\documentclass{article}

\usepackage{arxiv}

\usepackage[utf8]{inputenc} 
\usepackage[T1]{fontenc}    
\usepackage{hyperref}       
\usepackage{url}            
\usepackage{booktabs}       
\usepackage{nicefrac}       
\usepackage{microtype}      
\usepackage{lipsum}		
\usepackage{natbib}
\usepackage{doi}
\usepackage{amsmath,amssymb,amsfonts,amsthm}
\usepackage{algorithmic,algorithm}
\usepackage{graphicx}
\usepackage{textcomp}
\usepackage{xcolor}
\usepackage{xcolor}
\newtheorem{theorem}{Theorem}
\newtheorem{definition}{Definition}
\newtheorem{lemma}{Lemma}

\newtheorem{assumption}{Assumption}
\newtheorem{rem}{Remark}
\usepackage{comment}
\usepackage{subfig}

\title{Safe and Quasi-Optimal Autonomous Navigation in Environments with Convex Obstacles}

\author{ Ishak Cheniouni \thanks{A preliminary version of the present work has been presented in the 2023 American Control Conference \citep{ACC23}.}\\
	Department of Electrical Engineering\\
	Lakehead University\\
	Thunder Bay, ON P7B 5E1, Canada \\
	\texttt{cheniounii@lakeheadu.ca} \\
	\And
	 Soulaimane Berkane\thanks{Soulaimane Berkane is also with the Department of Electrical Engineering, Lakehead University, Thunder Bay, ON P7B 5E1, Canada.}\\
	Department of Computer Science and Engineering\\
	University of Quebec in Outaouais\\
	101 St-Jean Bosco, Gatineau, QC, J8X 3X7, Canada \\
	\texttt{soulaimane.berkane@uqo.ca}\\
        \And
        Abdelhamid Tayebi \\
	Department of Electrical Engineering\\
	Lakehead University\\
	Thunder Bay, ON P7B 5E1, Canada \\
	\texttt{atayebi@lakeheadu.ca} \\
}

\date{}

\hypersetup{
pdftitle={A template for the arxiv style},
pdfsubject={q-bio.NC, q-bio.QM},
pdfauthor={David S.~Hippocampus, Elias D.~Striatum},
pdfkeywords={First keyword, Second keyword, More},
}

\begin{document}
\maketitle

\begin{abstract}

We propose a continuous feedback control strategy that steers a point-mass vehicle safely to a destination, in a {\it quasi-optimal} manner, in sphere worlds. 
The main idea consists in avoiding each obstacle via the shortest path on the cone's surface enclosing the obstacle and heading straight toward the target when the vehicle has a clear line of sight to the target location. In particular, almost global asymptotic stability of the target location is achieved in two-dimensional (2D) environments under a particular assumption on the obstacles configuration. We also propose a reactive (sensor-based) approach, suitable for real-time implementations in \textit{a priori} unknown 2D environments with sufficiently curved convex obstacles, guaranteeing almost global asymptotic stability of the target location. Simulation results are presented to illustrate the effectiveness of the proposed approach.

\end{abstract}


\section{Introduction}
\label{sec:introduction}
\subsection{Motivation} 
Autonomous navigation consists in steering a robot from an initial position to a final destination while avoiding obstacles. The existing solutions for this problem can be classified into two main approaches. The first approach is the plan-and-track approach, which consists in generating, from a map of the environment, a collision-free path to be tracked via a feedback controller. The second approach, referred to as feedback-based approach, is a direct approach which consists in designing, in one shot, a feedback control strategy that steers the robot to the target location along a collision-free path. While safe global (or almost global) convergence to a target is achieved in environments with specific geometries, the existing feedback-based approaches do not generally generate the shortest collision-free paths. In this paper, we address this problem by proposing a continuous feedback control strategy that generates {\it quasi-optimal}\footnote[1]{This term will be rigorously defined later.} trajectories and guarantees safe navigation in $n$-dimensional sphere worlds and two-dimensional arbitrary convex environments.
\subsection{Prior Literature}
Among the path-finding algorithms of the plan-and-track category, one can cite the Dijkstra algorithm \citep{Dijkstra1959ANO} or the A$^{\star}$ (A star) algorithm \citep{Astar}, which rely on grids or graphs representing the environment where the shortest path is determined. Sampling-based algorithms are an example of approaches using such strategies. One can find single query methods such as the family of rapidly exploring random trees (RRTs) that incrementally construct a search tree for a single initial/final goal pair with probabilistic completeness \citep{lavalle_2006}. The variant RRT* is asymptotically optimal \citep{Karaman}. For multiple query methods, the probabilistic roadmaps (PRMs) are another family of probabilistically complete algorithms that construct a roadmap (topological graph) for a given workspace on which different pairs of initial/final goals can be connected. Once an initial/final goal pair is connected to the roadmap, the search can be performed on the roadmap \citep{lavalle_2006}. The variant PRM* is endowed with asymptotic optimality \citep{Karaman}. While the optimality of the RRT* and PRM* depends on the density of the constructed graphs, the tangent visibility graph (TVG) approach designed for 2D environments with polygonal obstacles \citep{ROHNERT198671}, then extended to 2D curved obstacles \citep{LAUMOND198741,Arimoto}, is a multiple query method that provides the exact shortest paths. A more detailed discussion of such approaches can be found in \citep{lavalle_2006, latombe2012robot}. One can also find reactive motion planning algorithms, such as the family of Bug algorithms \citep{Bug1,Bug2} which are used to navigate in planar environments without guarantees on the optimality of the generated paths. 
Artificial potential field methods are an example of a feedback-based approach. They consider a robot moving in a force field where the destination generates an attractive force, and the obstacles generate repulsive forces \citep{khatib}. The destination is the minimum of the potential function, and the negative gradient leads safely to it. These methods suffer from two problems, namely, the generation of local minima where the robot may get trapped instead of reaching the goal, and if the goal is reached the generated path is not generally the shortest collision-free path. To address local minima problem, the authors in \citep{k_R_90} proposed a navigation function (NF) whose negative gradient is the control law that steers the robot from almost all initial conditions to the target location in an \textit{a propri} known sphere world.
In order to navigate in more general spaces, diffeomorphisms from sphere worlds to more complex worlds were proposed in \citep{Pr_K_R_91,R_k_92}. 
The authors in \citep{LoizouNT1, LoizouNT3} proposed tuning-free navigation functions and diffeomorphisms from a point world to a sphere world or a star world. A sufficient condition was given in \citep{SnsNF6} for an artificial potential to be a navigation function in environments containing smooth, non-intersecting, and strongly convex obstacles. More recently, a tuning-free navigation function based on harmonic functions has been proposed in \citep{loizou2021correctbyconstruction} for sensor-based autonomous navigation \citep{LoizouNT1,LoizouNT3}.

In \citep{Arslan2019}, the authors proposed a new sensor-based autonomous navigation strategy (different from the NF-based approach) by constructing a compact obstacle-free local set around the robot using the hyperplanes separating the robot from the neighboring obstacles and then steering the robot towards the projection of the target location onto the boundary of this compact set. This approach ensures safe navigation through unknown strongly convex obstacles and convergence to the destination from everywhere, except from a set of zero Lebesgue measure. This work has been extended for non-convex star-shaped obstacles in \citep{vasilopoulos1}, and polygonal obstacles with possible overlap in \citep{Vasilopoulos2}.
A sensor-based autonomous navigation approach, relying on Nagumo's theorem  \citep{Nagumo} and using tangent cones, was proposed in \citep{souVeloCones}. This approach guarantees safety through an appropriate switching between a stabilizing controller and an obstacle avoidance controller. Control Barrier Functions (CBFs) and Control Lyapunov Functions (CLFs) were used in \citep{Barrier6, Barrier7} along with a quadratic program to design navigation controllers ensuring the stabilization of the desired target location with safety guarantees. Hybrid feedback was used, for instance in \citep{hybr1,SoulaimaneHybTr,Mayur2022}, to achieve global convergence; a feature that is not possible to obtain via continuous time-invariant control due to topological obstructions \citep{k_R_90}.
 Unfortunately, the feedback-based approaches in the above-mentioned papers, although endowed with almost global or global asymptotic stability, do not take into account the optimality of the generated trajectories.
\subsection{Contributions}
The present paper proposes a continuous feedback control strategy that generates {\it quasi-optimal} trajectories, as per Definition \ref{def2} that will be provided in subsection \ref{section:Definition1}, and ensures safe autonomous navigation in $n$-dimensional sphere worlds and two-dimensional environments with arbitrary convex and sufficiently curved obstacles. Our approach relies on iteratively projecting the nominal feedback controller on the obstacles' enclosing cones for known environments, which generates locally optimal collision-free trajectories. A sensor-based implementation of our control approach is proposed for {\it a priori} unknown 2D environments with arbitrary convex and sufficiently curved obstacles. The main contributions of the proposed approach are summarized as follows:
\begin{itemize}
    \item The proposed continuous feedback control generates {\it quasi-optimal} trajectories in terms of length. The generated trajectories are often the shortest, as illustrated through extensive simulation results. 
    \item Except for the restrictions imposed by the standard separation conditions of Assumptions \ref{as:1} and \ref{as:2}, the environment can be highly dense, and the destination can be located arbitrarily close to the boundaries of the obstacles.
    \item The reactive (sensor-based) version of our approach applies to {\it a priori} unknown 2D environments with arbitrary convex and sufficiently curved obstacles, and ensures almost global asymptotic stabilization of the target location.
\end{itemize}
\subsection{Organization}
The remainder of this paper is organized as follows: Section II provides the preliminaries that will be used throughout this article. In Section III, we formulate our autonomous navigation problem. In Section IV, we define the subsets of the free space that
are needed for our proposed control design. In Sections V–VI, we present our control strategy and its properties. In Section VII, our control strategy is adapted to the sensor-based scenarios in two-dimensional sphere worlds and sufficiently curved convex worlds. Simulation results are presented in Section VIII. Section X concludes with a summary of our contributions and prospects for future work.

\section{Notations and Preliminaries}
Throughout the paper, $\mathbb{N}$, $\mathbb{R}$ and $\mathbb{R}_{>0}$ denote the set of natural numbers, real numbers and positive real numbers, respectively. The Euclidean space and the unit $n$-sphere are denoted by $\mathbb{R}^n$ and $\mathbb{S}^n$, respectively. The Euclidean norm of $x\in\mathbb{R}^n$ is defined as $\|x\|:=\sqrt{x^\top x}$ and the angle between two non-zero vectors $x,y\in\mathbb{R}^n$ is given by $\angle (x,y):=\cos^{-1}(x^\top y/\|x\|\|y\|)$ . The Jacobian matrix of a vector field $f: \mathbb{R}^n\rightarrow\mathbb{R}^n$ is given by $J_x(f(x))=[\nabla_x f_1\dots\nabla_x f_n]^\top$ where $\nabla_x f_i=[\frac{\partial f_i}{\partial  x_1}\,\dots\,\frac{\partial f_i}{\partial  x_n}]^\top$ is the gradient of the $i$-th element $f_i$.  Define the ball centered at $x\in\mathbb{R}^n$, of radius $r\in\mathbb{R}_{>0}$, by the set $\mathcal{B}(x,r):=\left\{q\in\mathbb{R}^n|\;\|q-x\| \leq r\right\}$. The interior and the boundary of a set $\mathcal{A}\subset\mathbb{R}^n$ are denoted by $\mathring{\mathcal{A}}$ and $\partial\mathcal{A}$, respectively. The relative complement of a set $\mathcal{B}\subset\mathbb{R}^n$ with respect to a set $\mathcal{A}$ is denoted by $\mathcal{B}^c_\mathcal{A}$. The distance of a point $x\in\mathbb{R}^n$ to a closed set $\mathcal{A}$ is defined as $d(x,\mathcal{A}):=\min\limits_{q\in\mathcal{A}}\|q-x\|$. The Minkowski sum of two convex sets $\mathcal{A}$ and $\mathcal{B}$ is defined as $\mathcal{A}\oplus\mathcal{B}:=\left\{a+b|a\in\mathcal{A},b\in\mathcal{B}\right\}$. The cardinality of a set $\mathcal{N}\subset\mathbb{N}$ is denoted by $\mathbf{card}(\mathcal{N})$. The half-line, starting from a point $x\in\mathbb{R}^n$, with direction $v\in\mathbb{R}^n\setminus\{0\}$ is defined as $\mathcal{L}_h(x,v):=\left\{q\in\mathbb{R}^n|q=x+\delta v,\;\delta\geq0\right\}$. The line segment connecting two points $x,y\in\mathbb{R}^n$ is defined as $\mathcal{L}_s(x, y):=\left\{q\in\mathbb{R}^n|q=x+\delta(y-x),\;\delta\in[0,1]\right\}$. The parallel and orthogonal projections are defined as follows:
\begin{align}
    \pi^{\parallel}(v):=vv ^\top,\quad\pi^{\bot}(v)&:=I_n-vv^\top,
\end{align}
where $I_n\in\mathbb{R}^{n\times n}$ is the identity matrix and $v\in\mathbb{S}^{n-1}$. Therefore, for any vector $x$, the vectors $\pi^{\parallel}(v) x$ and $\pi^{\bot}(v) x$ correspond, respectively, to the projection of $x$ onto the line generated by $v$ and onto the hyperplane orthogonal to $v$.
Let us define the set $\mathcal{P}_{\Delta}(x,v)=\left\{q\in\mathbb{R}^n|v^\top(q-x)~\Delta~0\right\}$, with $\Delta \in\{=,>,\geq,<,\leq\}$.
The hyperplane passing through $x\in\mathbb{R}^n$ and orthogonal to $v\in\mathbb{R}^n\setminus\{0\}$ is denoted by $\mathcal{P}_{=}(x,v)$. The closed negative half-space (resp. open negative half-space) is denoted by $\mathcal{P}_{\leq}(x,v)$ (resp. $\mathcal{P}_{<}(x,v)$) and the closed positive half-space (resp. open positive half-space) is denoted by $\mathcal{P}_{\geq}(x,v)$ (resp. $\mathcal{P}_{>}(x,v)$).
A conic subset of $\mathcal{A}\subseteq\mathbb{R}^n$, with vertex $x\in\mathbb{R}^n$, axis $v\in\mathbb{R}^n\setminus\{0\}$, and aperture $2\psi$ is defined as follows \citep{HybBerkaneECC2019}:
{\small
\begin{align}
    \mathcal{C}^{\Delta}_{\mathcal{A}}(x,v,\psi):=\left\{q\in\mathcal{A}|\|v\|\|q-x\|\cos(\psi)~\Delta~ v^\top(q-x)\right\},
\end{align}}
where $\psi\in(0,\frac{\pi}{2}]$ and $\Delta\in\left\{\leq,<,=,>,\geq\right\}$,  with $``="$, representing the surface of the cone, $``\leq"$ (resp. $``<"$) representing the interior of the cone including its boundary (resp. excluding its boundary), and $``\geq"$ (resp. $``>"$) representing the exterior of the cone including its boundary (resp. excluding its boundary). The set of vectors parallel to the cone $\mathcal{C}^=_{\mathbb{R}^n}(x,v,\psi)$ is defined as follows:
    \begin{align}\label{17}
        \mathcal{V}(v,\psi):=\left\{a\in\mathbb{R}^n|\;\;a^\top v=\|a\|\|v\|\cos(\psi)\right\}.
    \end{align}
\section{Problem Formulation}
Consider a point mass vehicle at position $x\in\mathbb{R}^n$ moving inside a spherical workspace $\mathcal{W}\subset\mathbb{R}^n$ centered at the origin $0$ and punctured by $m\in\mathbb{N}$ balls $\mathcal{O}_i$ such that:
\begin{align}
&\mathcal{W}:=\mathcal{B}(0,r_0),\\
&\mathcal{O}_i:=\mathcal{B}(c_i,r_i),\;\;i\in\mathbb{I}:=\{1,\dots,m\},\label{m2}
\end{align}
where $r_0>r_i>0$ for all $i\in\mathbb{I}$. The free space is, therefore, given by the closed set
\begin{align}\label{8}
   \mathcal{F}:=\mathcal{W}\setminus\bigcup\limits_{i=1}^{m}\mathring{\mathcal{O}}_i.
\end{align}
For $\mathcal{F}$ to be a valid sphere world, as defined in \citep{k_R_90}, the obstacles $\mathcal{O}_i$ must satisfy the following assumptions:
\begin{assumption}\label{as:1}
 The obstacles are completely contained within the workspace and separated from its boundary, {\it i.e.,}
\begin{align}\label{9}
    \min\limits_{a\in\mathcal{O}_i,b\in\partial\mathcal{W}}\|a-b\|>0,\,\forall i\in\mathbb{I}.
\end{align}
\end{assumption}
\begin{assumption}\label{as:2}
 The obstacles are disjoint, {\it i.e.,}
 \begin{align}\label{10}
\min\limits_{a\in\mathcal{O}_i,b\in\mathcal{O}_j}\|a-b\|>0,\,\forall i,j\in\mathbb{I},\,i\neq j.
\end{align}
\end{assumption}
Consequently, the boundary of the free space $\mathcal{F}$ is given by
\begin{align}\label{11}
    \partial\mathcal{F}:=\partial\mathcal{W}\bigcup\Bigl(\bigcup\limits_{i=1}^{m}\partial\mathcal{O}_i\Bigr).
\end{align}
Consider the following first-order dynamics
\begin{align}\label{12}
    \Dot{x}=u,
\end{align}
where $u$ is the control input. The objective is to determine a continuous Lipschitz state-feedback controller $u(x)$ that safely steers the vehicle from almost any initial position $x(0)\in\mathcal{F}$ to any given desired destination $x_d\in\mathring{\mathcal{F}}$. In particular, the closed-loop system
\begin{equation}\label{eq:closed-loop-system}
    \dot x=u(x),\quad x(0)\in\mathcal{F}
\end{equation}
must ensure forward invariance of the set $\mathcal{F}$, almost global asymptotic stability\footnote{An equilibrium point is almost globally asymptotically stable if it is stable and attractive from all initial conditions except from a set of zero Lebesgue measure.} of the equilibrium $x=x_d$, and generates {\it quasi-optimal} trajectories that will be rigorously defined later in subsection \ref{section:Definition1}.
\section{Sets Definition and Obstacles Classification}\label{section:sets}
In this section, we define the subsets of the free space that are needed for our proposed control design in Section \ref{section:control-design}. These are depicted in Fig. \ref{fig:fig1} and given as follows:
\begin{itemize}
\item The hat of a cone inside the workspace $\mathcal{W}$, enclosing an obstacle $\mathcal{O}_i$, of vertex $x\in\mathbb{R}^n$ and aperture $\theta_i$ is defined as follows:
     \begin{align}\label{17}
        \mathcal{H}(x,c_i):=\bigl\{q\in\mathcal{C}^{\leq}_{\mathcal{W}}(x,c_i-x,\theta_i(x))|(c_i-q)^{\top}(x-q)\leq0\bigr\},
    \end{align}
    where the angle $\theta_i(x)=\arcsin\left(r_i/\|c_i-x\|\right)\in(0,\frac{\pi}{2}]$.
    \item The shadow region of obstacle $\mathcal{O}_i$, which is the area hidden by obstacle $\mathcal{O}_i$, from which there is no line of sight to the destination, is defined as follows:
    \begin{align}\label{18}
        \mathcal{D}(x_d,c_i):=\bigl\{q\in\mathcal{C}^{\leq}_{\mathcal{F}}(x_d,c_i-x_d,\varphi_i)|(c_i-q)^\top(x_d-q)\geq0\bigr\},
    \end{align}
    where the angle $\varphi_i=\arcsin\left(r_i/\|c_i-x_d\|\right)\in(0,\frac{\pi}{2}]$.
    \item The exit set of obstacle $\mathcal{O}_i$ separates the set $\mathcal{D}(x_d,c_i)$ and its complement with respect to $\mathcal{F}$ and is defined as follows:
    \begin{align}\label{20}
        \mathcal{S}(x_d,c_i):=\bigl\{q\in\mathcal{C}^=_{\mathcal{F}}(x_d,c_i-x_d,\varphi_i)|(c_i-q)^\top(x_d-q)\geq0\bigr\}.
    \end{align}
    \item The blind set is a subset of $\mathcal{F}$ where there is no line of sight to the destination, and is defined as follows:
    \begin{align}
      \mathcal{BL}&:=\left\{q\in\mathcal{F}|\mathcal{L}_s(q,x_d)\cap\mathcal{O}_k\neq\varnothing,\,k\in\mathbb{I}\right\},\label{bl_1}\\
      &:=\bigcup\limits_{i\in\mathbb{I}}\mathcal{D}(x_d,c_i).
    \end{align}
    \item The visible set is the complement of the blind set with respect to the free space
    \begin{align}
        \mathcal{VI}:= \mathcal{BL}^c_\mathcal{F}.
    \end{align}
    \item The set of blocking obstacles between two given positions $x$ and $y$ is the set of obstacles crossed by the line-segment $\mathcal{L}_s(x,y)$, and  is defined as follows:
    \begin{align}
        \mathcal{LO}(x,y):=\{k\in\mathbb{I}|\mathcal{O}_k\cap\mathcal{L}_s(x,y)\neq\varnothing\}.
    \end{align}
\end{itemize}
\begin{figure}[h!]
     \centering
     \subfloat[]{\includegraphics[width=0.48\linewidth,height=0.48\linewidth,keepaspectratio]{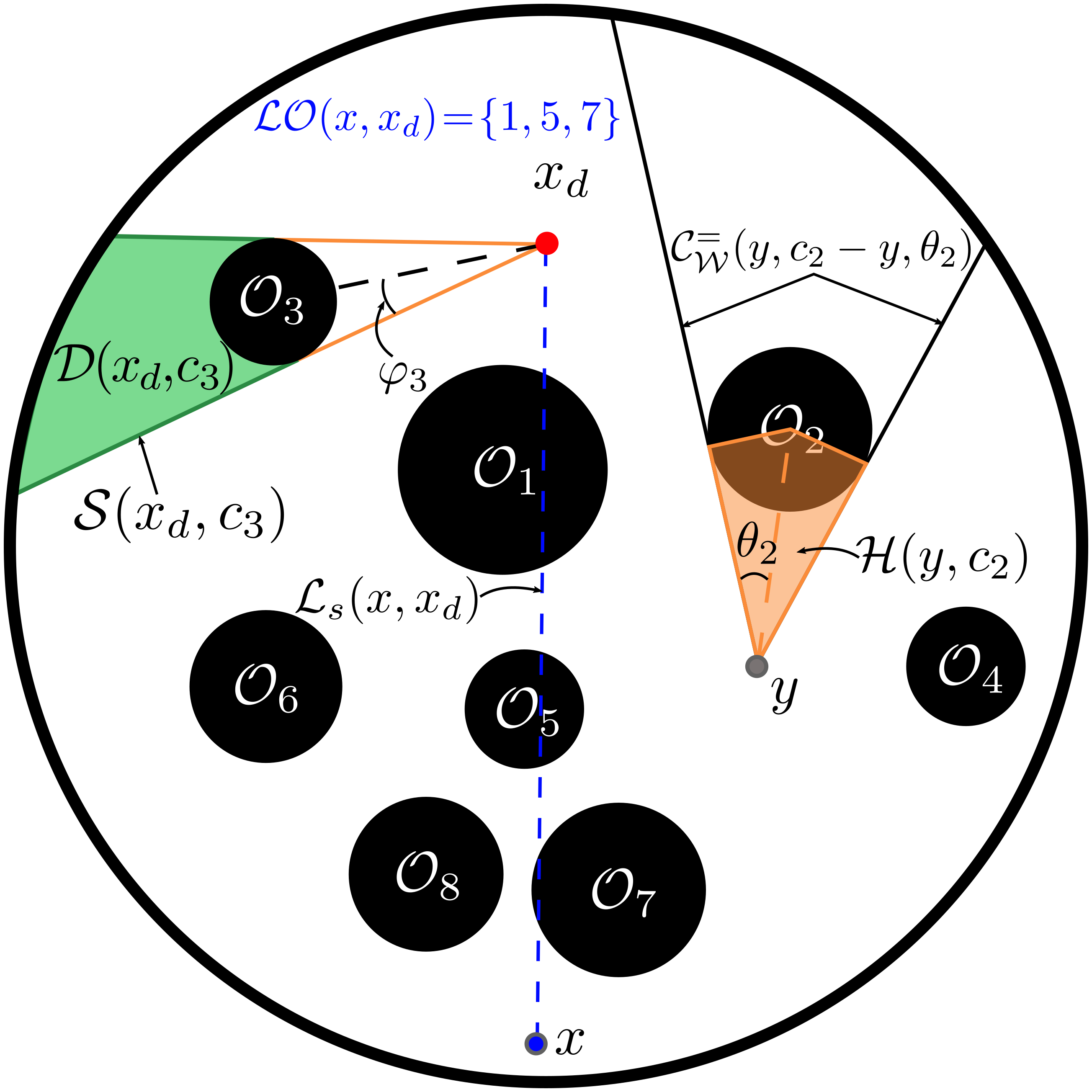}\label{fig:fig1_left}}
     \subfloat[]{\includegraphics[width=0.48\linewidth,height=0.48\linewidth,keepaspectratio]{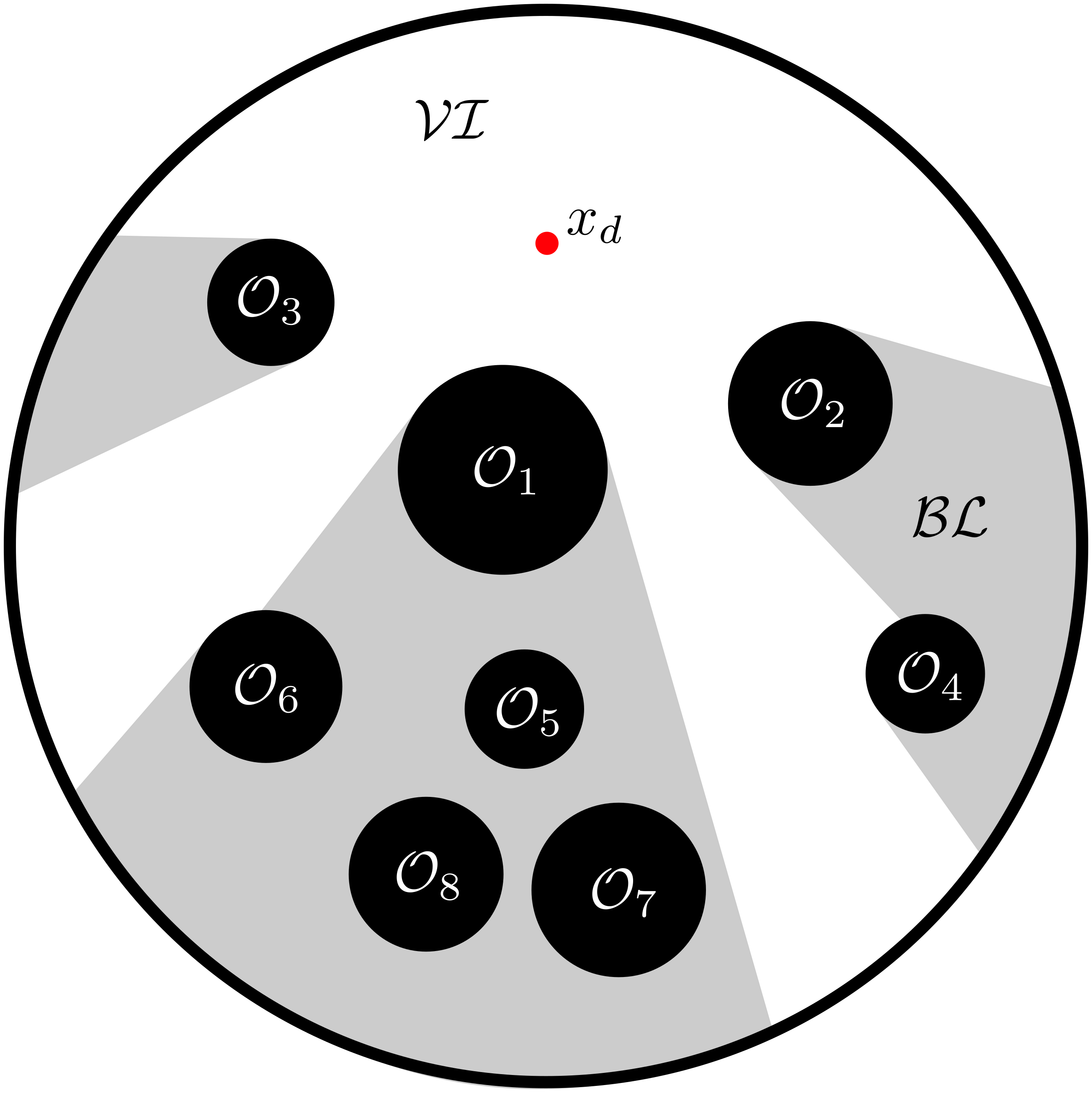}\label{fig:fig1_right}}
     \caption{2D representation of the sets in Section IV.}
     \label{fig:fig1}
\end{figure}
\section{Control Design}\label{section:control-design}
\subsection{Single Obstacle Case}
We design a preliminary control law for the single obstacle case, which will be used as a baseline in the multiple obstacle case. Let us start by considering a single obstacle $\mathcal{O}_i$ and ignoring all others. In the case where the path is clear (\textit{i.e.,} $x$ belongs to the visible set  $\mathcal{VI}$), the vehicle follows a straight line to the destination under the control law $u_d(x)=-\gamma(x-x_d)$ where $\gamma \in\mathbb{R}_{>0}$. In the case where the path is not clear (\textit{i.e.,} $x\in \mathcal{D}(x_d,c_i)$), we generate a control input (vehicle's velocity) that is in the direction of the cone $\mathcal{C}^{=}_{\mathcal{F}}(x,c_i-x,\theta_i)$ enclosing the obstacle. In particular, the direction of the control input should minimize the angle between the nominal control direction, given by $(x_d-x)$, and the set of all vectors parallel to the enclosing cone, that is 
    \begin{equation}\label{min}
    u(x)\in\mathcal{U}(x):=\mathrm{arg}\min\limits_{v_i\in\mathcal{V}(c_i-x,\theta_i)}\angle(x_d-x,v_i),\,x\in\mathcal{D}(x_d,c_i).
    \end{equation}
    Moreover, to ensure continuity of the control input, we impose further that the control is equal to $u_d(x)$ at the exit set $\mathcal{S}(x_d,c_i)\subset\mathcal{D}(x_d,c_i)$, namely
    \begin{align}\label{constraint}
        \forall x\in\mathcal{S}(x_d,c_i),\;u(x)=u_d(x),
    \end{align}
    The following lemma provides the solution of the optimization problem \eqref{min}-\eqref{constraint} and shows its uniqueness.
\begin{lemma}\label{lem1}
The solution of the optimization problem \eqref{min}-\eqref{constraint} is unique and is given by
\begin{align}
    u(x)=\xi(u_d(x),x,i),
\end{align}
where $\xi:\mathbb{R}^n\times\mathbb{R}^n\times\mathbb{N}\to\mathbb{R}^n$ is given by
    \begin{align}\label{alg}
   \xi(u,x,i):=\frac{\sin(\beta_i(u,x))\sin^{-1}(\theta_i(x))}{\cos(\theta_i(x)-\beta_i(u,x))}\pi^{\parallel}(\bar\xi_i)u,
    \end{align}
    with $\bar\xi_i\in\mathcal{V}(c_i-x,\theta_i)$,
    \begin{align*}
        &\bar\xi_i:=\frac{\sin(\theta_i(x))u}{\sin(\beta_i(u,x))\|u\|}-\frac{\sin(\theta_i(x)-\beta_i(u,x))}{\sin(\beta_i(u,x))}\frac{(c_i-x)}{\|c_i-x\|},\\
        &\beta_i(u,x):=\angle(u,c_i-x)\leq\theta_i(x).
    \end{align*}
\end{lemma}
\begin{proof}
See Appendix \ref{appendix:Lemma 1}.
\end{proof}
In other words, Lemma \ref{lem1} shows that,  when $x\in\mathcal{D}(x_d,c_i)$, the control $u(x)$ is a scaled parallel projection of the nominal controller $u_d(x)$ in the direction of $\bar\xi_i$ which represents a unit vector on the cone enclosing the obstacle.
Finally, one obtains the following continuous control strategy in the case of a single obstacle
\begin{align}
   u(x)= \begin{cases}\label{25}
      \displaystyle u_d(x), & x\in\mathcal{VI},\\
      \displaystyle\xi(u_d(x),x,i), & x\in\mathcal{D}(x_d,c_i).
    \end{cases} 
\end{align}
The trajectory of the closed-loop system \eqref{12}-\eqref{25} is length-optimal as shown in the following lemma and illustrated in Fig. \ref{fig:fig6}.
\begin{lemma}\label{lem2}
The path generated by the closed-loop system \eqref{12}-\eqref{25} is the shortest path to the destination $x_d$ from every initial condition $x(0)\in\mathcal{F}\setminus\mathcal{L}_d(x_d,c_i)$ where $\mathcal{L}_d(x_d,c_i):=\mathcal{D}(x_d,c_i)\cap\mathcal{L}_h(c_i,c_i-x_d)$.
\end{lemma}
\begin{proof}
See Appendix \ref{appendix:Lemma 2}.
\end{proof}
\begin{figure}[h!]
\centering
\includegraphics[scale=0.42]{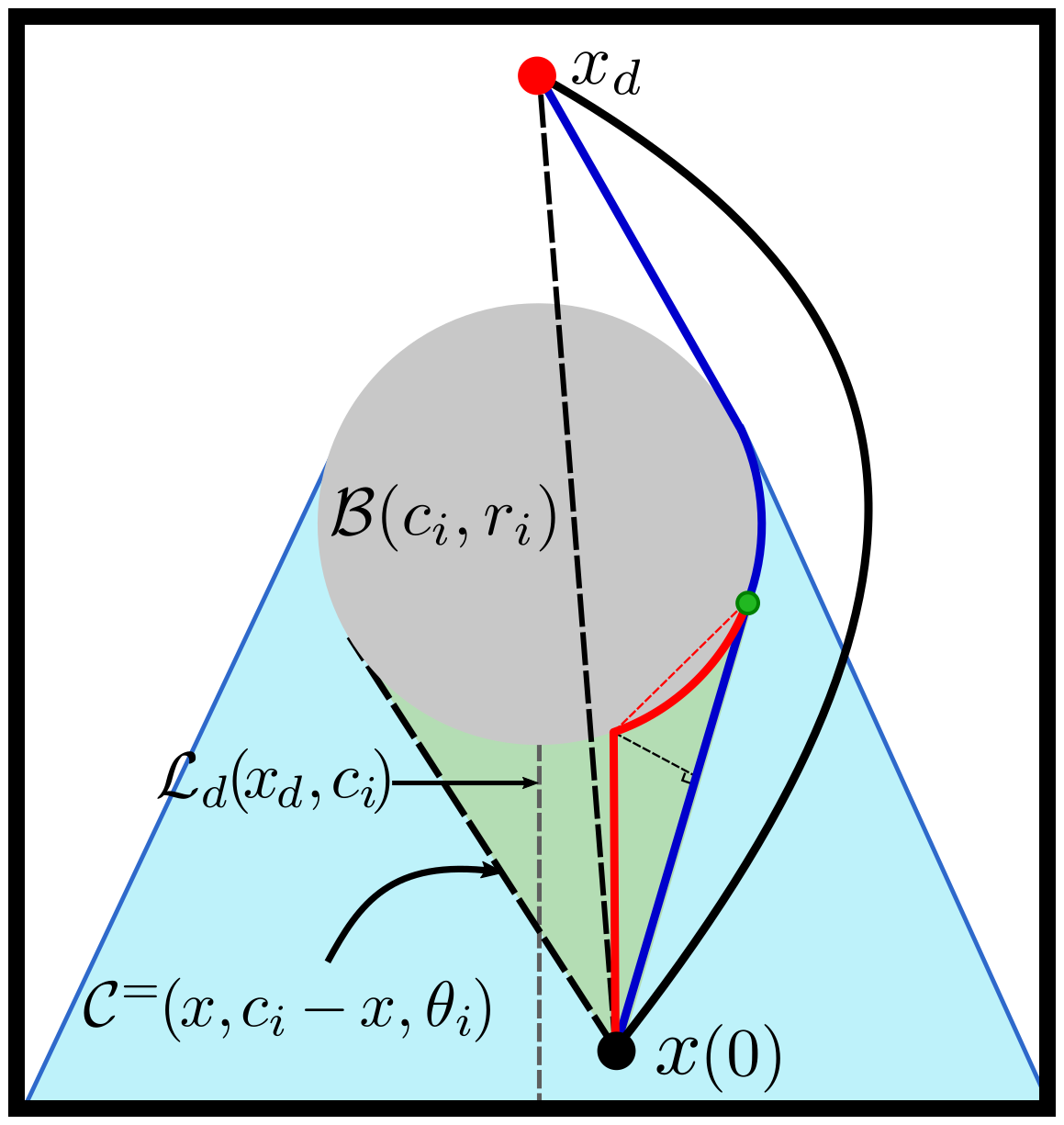}
\caption{Shortest path in a single-obstacle sphere world (blue curve).}
\label{fig:fig6}
\end{figure}
\subsection{Multiple Obstacles Case}
The objective in this subsection is to extend the controller \eqref{25} to the multiple obstacle case. The robot moves under the nominal control $u_d(x)$ when the robot has a clear line of sight to the destination ({\it i.e.,} $x\in\mathcal{VI}$). When there is no clear line of sight to the destination ({\it i.e.,} $x\in\mathcal{BL}$), one proceeds with multiple projections as described hereafter. As per \eqref{bl_1}, at every position $x\in\mathcal{BL}$, the blocking obstacles between $x$ and $x_d$ are represented by the set $\mathcal{LO}(x,x_d)\neq\varnothing$. Among the set $\mathcal{LO}(x,x_d)$, we select the one closest to destination $\bigl(${\it i.e.,} $i=\arg\{\min d(x_d,\mathcal{O}_k),\,k\in\mathcal{LO}(x,x_d)\}\bigr)$, where $u_d(x)$ is projected onto the enclosing cone of the selected obstacle $\mathcal{O}_i$ using \eqref{alg}, as in the case of a single obstacle. The resulting control vector is denoted by $u_1(x)$. The next obstacle to be considered is selected from the set of blocking obstacles $\mathcal{LO}(x,\hat{c}_i(x))$, where $\hat{c}_i(x):=x+\pi^{\parallel}(u_1(x)/\|u_1(x)\|)(c_i-x)$ is the point at which the line directed by $u_1(x)$ is tangent to obstacle $\mathcal{O}_i$. One chooses the closest obstacle among the set $\mathcal{LO}(x,\hat{c}_i)$ in terms of the Euclidean distance to $\mathcal{O}_i$. If $\mathcal{LO}(x,\hat{c}_i(x))=\varnothing$, the path is free. Otherwise, $u_1$ will be considered as $u_d$ for the newly selected obstacle and the same approach is followed to obtain $u_2$. Obstacle $\mathcal{O}_i$ is called an \textbf{ancestor} to the selected obstacle and the selection and projection are repeated until the path is free (see Fig. \ref{fig:fig3}).
The obstacles selected during the successive projections at a position $x$, are grouped in an ordered list $\mathcal{I}(x)\subset\mathbb{I}$ from the first obstacle $\bigl(\mathcal{O}_i$, such that $i=\arg\{\min d(x_d,\mathcal{O}_k),\,k\in\mathcal{LO}(x,x_d)\}\bigr)$ to the last one (obstacle involved in the last projection). Let $h(x)=\mathbf{card}(\mathcal{I}(x))$ be the number of required projections at  position $x$. Define the map $\iota_x :\{1,\dots,h(x)\}\rightarrow\mathcal{I}(x)$ which associates to each projection $p\in\{1,\dots,h(x)\}$ the corresponding obstacle $\iota_x(p)\in\mathcal{I}(x)$. 
The set of positions involving obstacle $k$ in the successive projections is called active region and defined as $\mathcal{AR}_k:=\left\{q\in\mathcal{BL}|k\in\mathcal{I}(q)\right\}$. To sum up, the intermediary control at a step $p\in\{1,\dots,h(x)\}$ and position $x\in\mathcal{AR}_{\iota_x(p)}$ is given by the recursive formula
    \begin{align}\label{eq:recursive-control}
        u_p(x)=\xi(u_{p-1}(x),x,\iota_x(p)),
    \end{align}
    with $u_0(x)=u_d(x)$ and $\xi(\cdot,\cdot,\cdot)$ as defined in Lemma \ref{lem1}. The point at which the line directed by $u_p(x)$ is tangent to obstacle $\mathcal{O}_{\iota_x(p)}$ is given by $\hat{c}_{\iota_x(p)}(x):=x+\pi^{\parallel}(u_p(x)/\|u_p(x)\|)(c_{\iota_x(p)}-x)$.
Finally, the proposed control law is obtained by performing $h(x)$ successive projections and is given by
\begin{align}\label{36}
   u(x)= \begin{cases}
      u_d(x), & x\in\mathcal{VI},\\
      u_{h(x)}(x),&
      x\in \mathcal{BL}.
    \end{cases} 
\end{align}
    \begin{figure}[h!]
    \centering
    \includegraphics[scale=0.4]{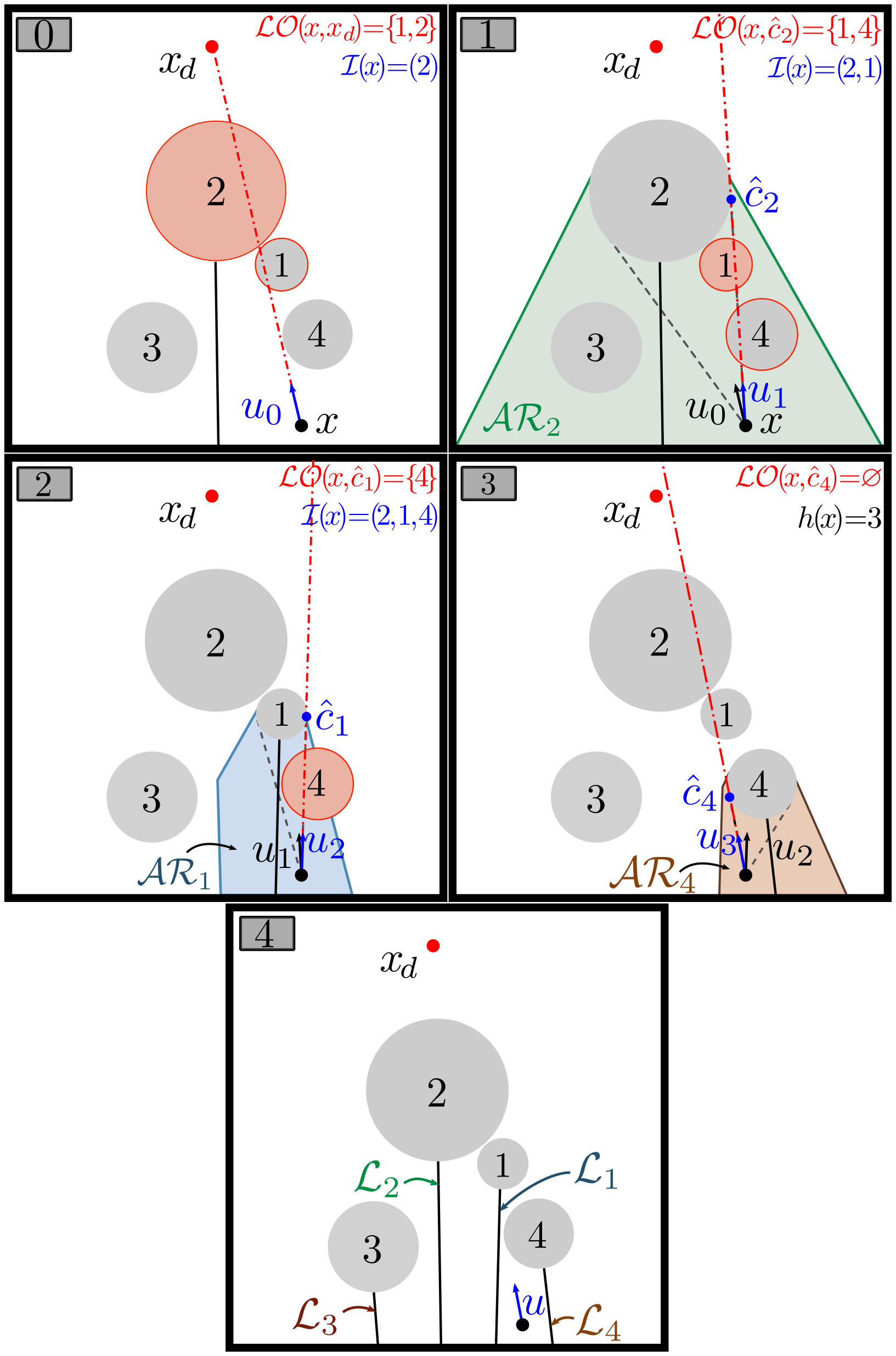}
    \caption{Successive projections of the control $u_d$ in a two-dimensional sphere world with four obstacles.}
    \label{fig:fig3}
    \end{figure}
    The implementation of the control strategy \eqref{36}  is summarized in Algorithm \ref{alg1}.
    \begin{rem}
     It is worth pointing out that the successive projections, involved in the control design, start from the closest obstacle to the destination. This approach enables our controller to enjoy the following features: 1) generates {\it quasi-optimal} trajectories; 2) guarantees the continuity of the control input.
    \end{rem}
 \begin{algorithm}
 \caption{Implementation of the control law \eqref{36} in the closed-loop system \eqref{eq:closed-loop-system}}\label{alg1}
 \begin{algorithmic}[1]
 \renewcommand{\algorithmicrequire}{\textbf{Initialization:}}
  \REQUIRE : $x_d$, $e_c$;\\
  \WHILE{true}
  \STATE Measure $x$;
  \IF{$\|x-x_d\|\leq e_c$}
  \STATE Break;
  \ELSE
  \IF{ $x\in\mathcal{BL}$}  
  \STATE $i\leftarrow\arg\min\limits_{k\in \mathcal{LO}(x,x_d)}d(x_d,\mathcal{O}_k)$;
  \WHILE{$i\neq\{\varnothing\}$}
    \STATE Update $u$ using \eqref{eq:recursive-control};
    \IF{ $\mathcal{LO}(x,\hat{c}_i(x))=\varnothing$}
    \STATE $i\leftarrow\{\varnothing\}$;
    \ELSE
    \STATE $i\leftarrow\arg\min\limits_{k\in \mathcal{LO}(x,\hat{c}_i(x))} d(\hat{c}_i(x),\mathcal{O}_k)$;
    \ENDIF
  \ENDWHILE
  \ELSE
  \STATE $u\leftarrow u_d$;
  \ENDIF
  \STATE Execute $u$ in \eqref{eq:closed-loop-system};
  \ENDIF
 \ENDWHILE
 \end{algorithmic} 
 \end{algorithm}
\subsection{Characterization of the generated trajectories}\label{section:Definition1}
The proposed control strategy steers the robot from an initial location $x_0\in\mathcal{F}$ to a final destination $x_d\in\mathring{\mathcal{F}}$ by tracking a position-dependant virtual destination. A virtual destination at a position $x\in\mathcal{BL}$ is given by $P(x):=P_{h(x)}(x)$, where $P_{h(x)}(x)$ is the last in a list of successive intermediary destinations $P_p(x):=x+u_{p}(x)$, with $p\in\{1,\dots,h(x)\}$, $h(x)=\mathbf{card}(\mathcal{I}(x))$ and $P_0(x)=x_d$. The point $P_p(x)$ lies on the surface of the cone enclosing the obstacle of index $\iota_x(p)\in\mathcal{I}(x)$. The intermediary destinations are designed to guarantee a minimum deviation between $(P_{p-1}(x)-x)$ and $(P_p(x)-x)$ for all $p\in\{1,\dots,h(x)\}$. This deviation represented by the angle $\angle((P_{p-1}(x)-x),(P_{p}(x)-x))=\angle(u_{p-1}(x),u_p(x))$ is the smallest possible since $u_p(x)=\xi(u_{p-1}(x),x,\iota_x(p))$, where the operator $\xi(\cdot,\cdot,\cdot)$, defined in Lemma \ref{lem1}, minimizes the angle $\angle(u_{p-1}(x),u_p(x))$ such that $u_p(x)\in\mathcal{V}(c_{\iota_x(p)}-x,\theta_{\iota_x(p)}(x))$. Recall that the set $\mathcal{V}(c_{\iota_x(p)}-x,\theta_{\iota_x(p)}(x))$ is the set of vectors parallel to the cone $\mathcal{C}^{=}(x,c_{\iota_x(p)}-x,\theta_{\iota_x(p)}(x))$ enclosing obstacle $\mathcal{O}_{\iota_x(p)}$. The virtual destination $P(x)$, at a position $x\in\mathcal{BL}$, is the final intermediary destination obtained through the following recursive minimization process:
\begin{equation}\label{min_process}
\begin{aligned}
    &P_p(x):=\arg\min\limits_{y\in\mathcal{C}^{=}_{\mathcal{F}}(x,c_{\iota_x(p)}-x,\theta_{\iota_x(p)}(x))\setminus\{x\}}\angle(y,P_{p-1}(x)),\\
    &P_0(x)=x_d,~~p\in\{1,\dots,h(x)\}. 
\end{aligned}
\end{equation}
The virtual destination coincides with the final destination (\textit{i.e.,} $P(x):=x_d$) when $x\in\mathcal{VL}$.\\
Throughout this paper, the trajectories generated with our optimized successive projections approach are referred to as \textit{quasi-optimal} trajectories, and are defined as follows:
\begin{definition}\label{def2}
    Given an initial position $x_0\in\mathcal{F}$ and a final destination $x_d\in\mathring{\mathcal{F}}$, a continuously differentiable trajectory connecting $x_0$ and $x_d$ is said to be {\it quasi-optimal} if it has the shortest length when $x_0\in\mathcal{VI}$ and when $x_0\in\mathcal{BL}$, the tangent vector to the trajectory, at each $x$, points towards the virtual destination $P(x)$ obtained by the recursive minimization process \eqref{min_process}. 
\end{definition}
\vskip 0.15cm
A {\it quasi-optimal} trajectory, as per Definition \ref{def2}, is a trajectory along which the vehicle's velocity, at a given location $x$, always points to a virtual destination (depending on $x$). The virtual destination, at position $x$ on the trajectory, is a result of a series of minimized deviations from the nominal direction (the direction from $x$ to $x_d$) with respect to the blocking obstacles, starting from the closest to the destination $x_d$. An example of a {\it quasi-optimal} trajectory is shown in Fig. \ref{fig:quasi_optimality_def} in blue color. Fig. \ref{fig:quasi_optimality_def}-\subref{fig:quasi_left}  and Fig. \ref{fig:quasi_optimality_def}-\subref{fig:quasi_middle} illustrate the characteristics of a {\it quasi-optimal} trajectory where at each position $x$ on the trajectory, the tangent to the trajectory points toward the green virtual destination. The green virtual destination $P(x)$ is obtained by first minimizing the deviation of the nominal direction to the final destination (red point) with respect to obstacle $4$, which gives the intermediary destination $P_1(x)$ (orange point). The same operation is repeated with all the intermediary destinations until the green virtual destination is obtained. A simulation video highlighting the characteristics of a {\it quasi-optimal} trajectory can be found at \url{https://youtu.be/CzIjtsy6HBA}. The generated {\it quasi-optimal} trajectory shown in Fig. \ref{fig:quasi_optimality_def}-\subref{fig:quasi_right} coincides with the shortest path (green curve). However, it is not always the case, as shown in Fig. \ref{fig:optimal}, where one can observe that, for the initial position $x_0^2$, the {\it quasi-optimal} trajectory (blue) coincides with the shortest path (green), while for the initial position $x_0^1$, it does not. 
The following remark provides some additional interesting features of {\it quasi-optimal} trajectories in two-dimensional environments.  
\begin{rem}
In two-dimensional environments, the {\it quasi-optimal} trajectories are length-optimal between any two successive avoided obstacles. They are generated by smoothly connected lines (common tangents to pairs of obstacles) and arcs of obstacles' boundaries. These trajectories belong to the tangent visibility graph (TVG) (also known as the reduced visibility graph) that was introduced in \citep{ROHNERT198671} for two-dimensional environments with polygonal obstacles and shown to contain the shortest path, then extended to two-dimensional environments with curved obstacles \citep{LAUMOND198741,Arimoto}.
\end{rem}

\begin{figure}[h!]
     \centering
     \subfloat[]{\includegraphics[width=0.29\linewidth,height=0.89\linewidth,keepaspectratio]{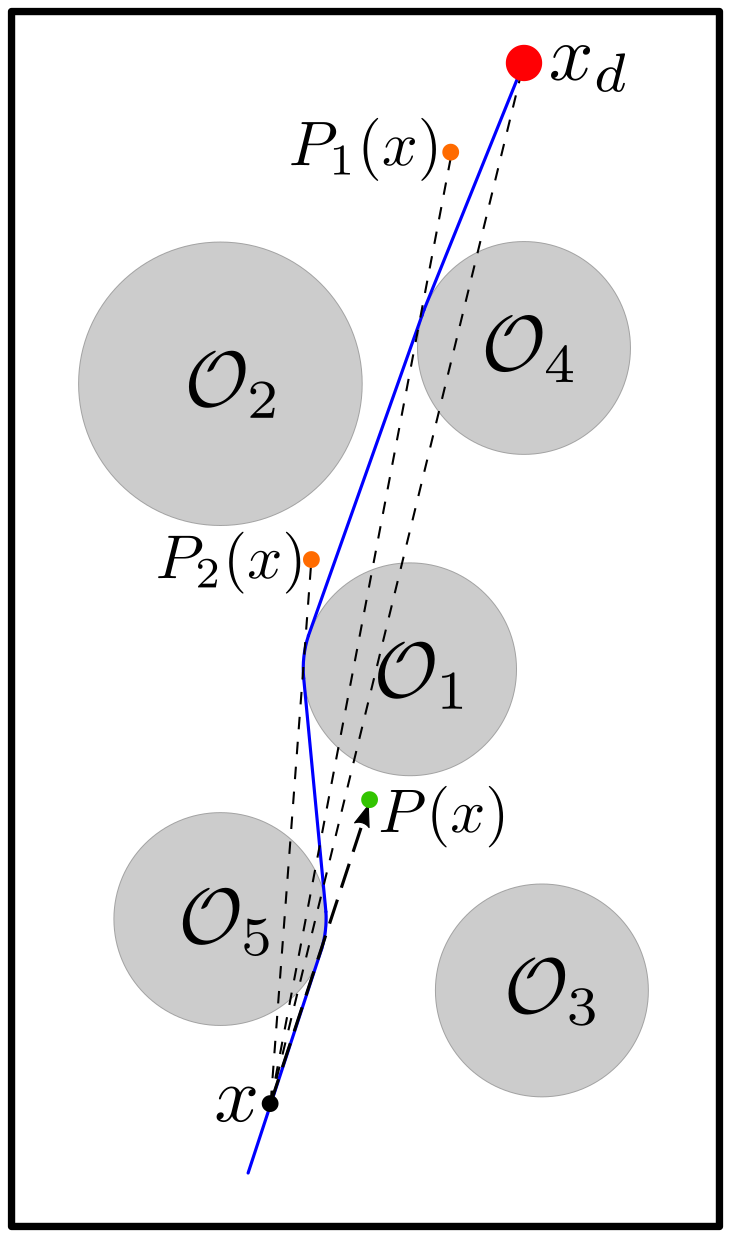}\label{fig:quasi_left}}
     \subfloat[]{\includegraphics[width=0.29\linewidth,height=0.89\linewidth,keepaspectratio]{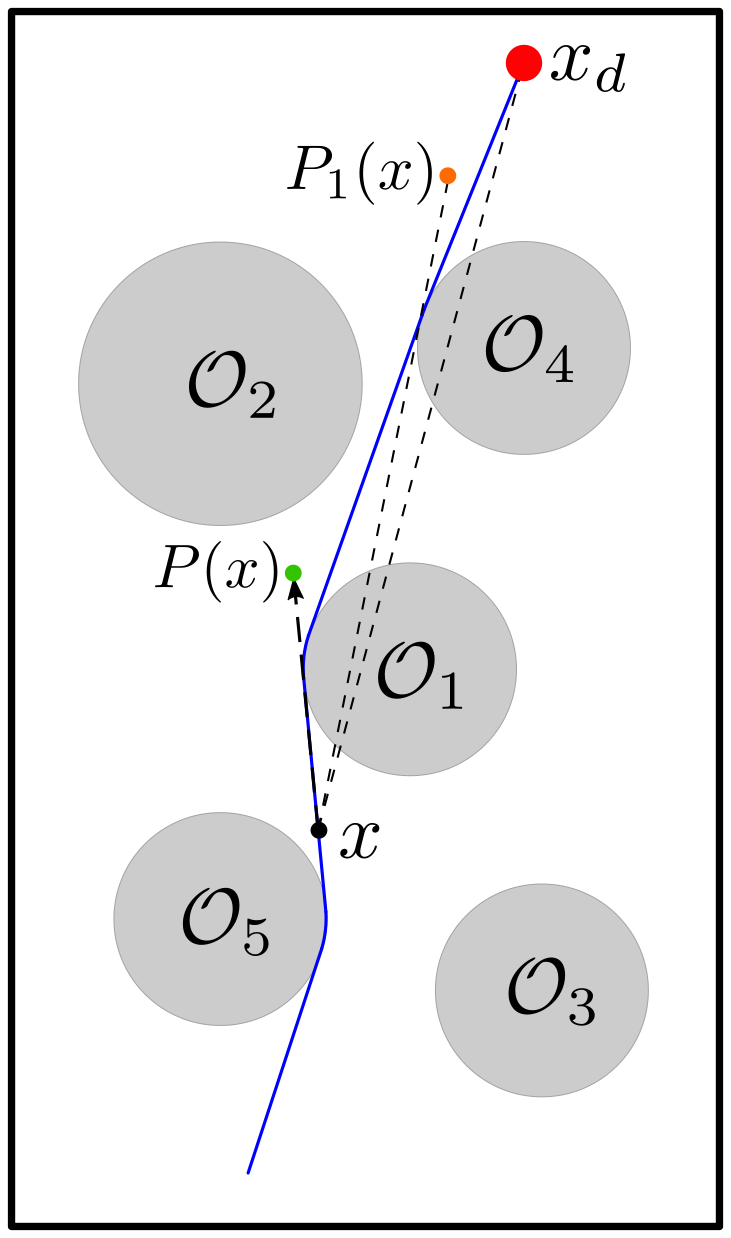}\label{fig:quasi_middle}}
      \subfloat[]{\includegraphics[width=0.29\linewidth,height=0.89\linewidth,keepaspectratio]{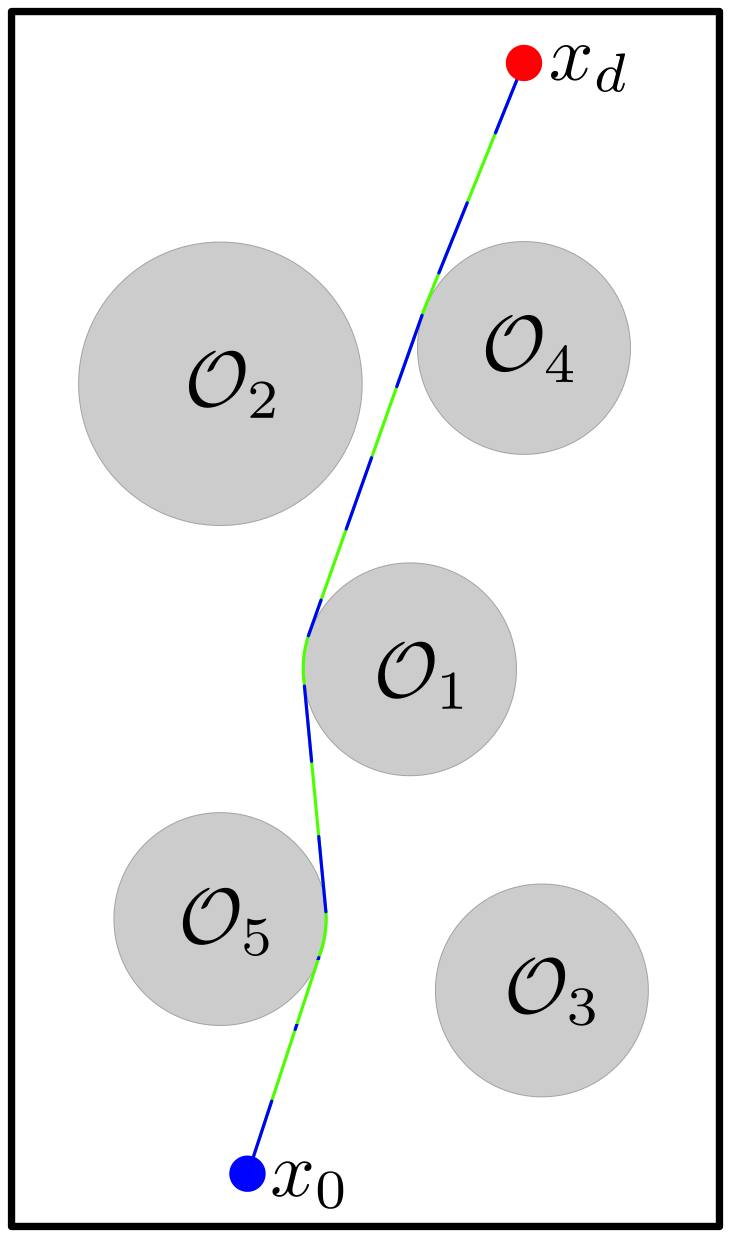}\label{fig:quasi_right}}
      \caption{Quasi-optimal trajectory in 2D workspace.}
     \label{fig:quasi_optimality_def}
\end{figure}

\begin{figure}[h!]
     \centering
     \includegraphics[width=2\linewidth,height=0.3\linewidth,keepaspectratio]{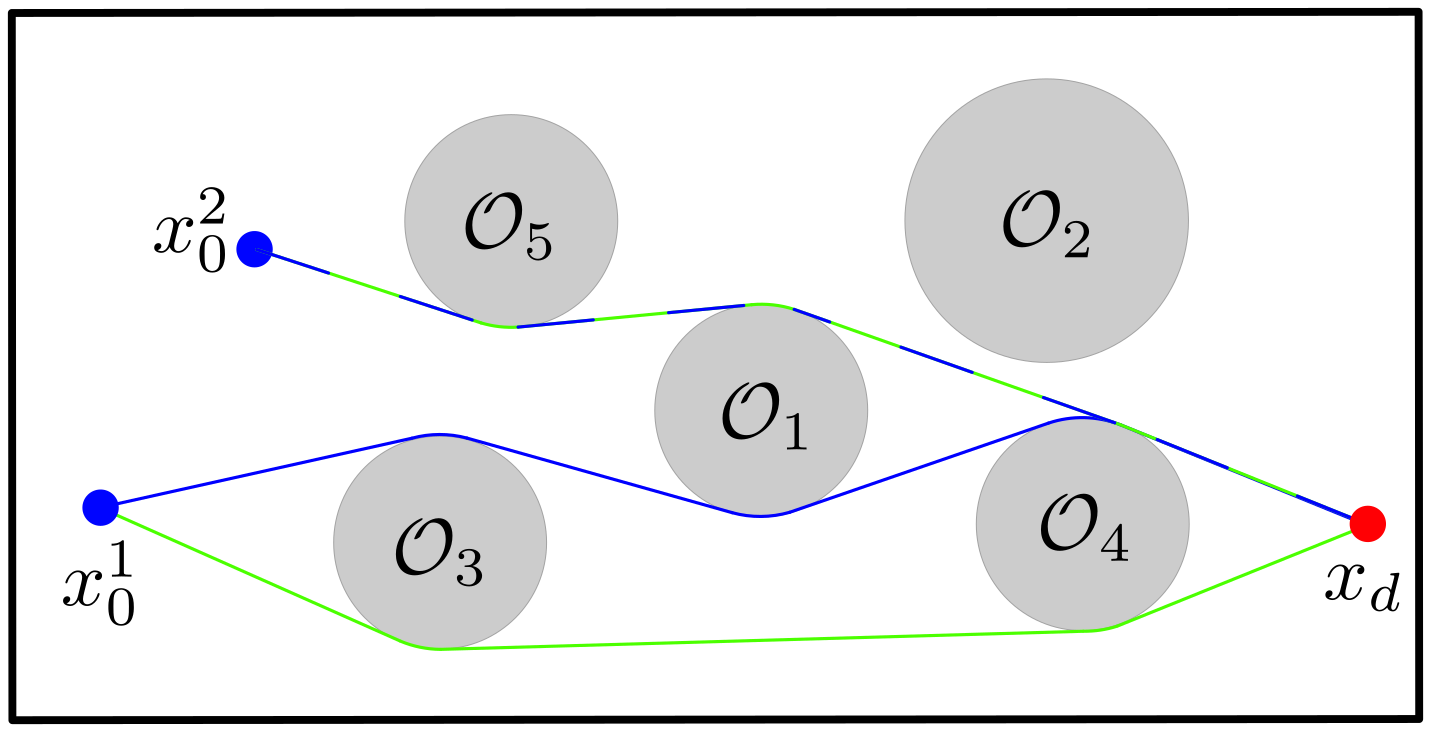}
     \caption{Optimal and quasi-optimal trajectories.}
     \label{fig:optimal}
\end{figure}
\section{Safety and Stability Analysis}
In this section, the safety and stability of the trajectories of the closed-loop system \eqref{12}-\eqref{36} will be analyzed.
Nagumo's theorem (\citep{Nagumo,Set-Theoretic_Methods_in_Control}), offers an important tool to prove safety. One of the statements of this theorem is the one based on Bouligand's tangent cones \citep{Bouligand}. 
\begin{definition}Given a closed set $\mathcal{K}$, the tangent cone to $\mathcal{K}$ at $x$ is
$
    \mathcal{T}_{\mathcal{K}}(x):=\{z:\lim_{\tau\rightarrow0}\mathrm{inf}\,\tau^{-1}d(x+\tau z, \mathcal{K})=0\}.
$
\end{definition}
In our case, when $x\in\mathring{\mathcal{F}}$, the tangent cone is the Euclidean space ($\mathcal{T}_{\mathcal{F}}(x)\equiv\mathbb{R}^n$), and since the free space is a sphere world (smooth boundary), the tangent cone at its boundary is a half-space (see Fig. \ref{fig:Tangent_cones}). Nagumo's theorem guarantees, in a navigation problem, that the robot stays inside the free space $\mathcal{F}$. For Nagumo's condition to be satisfied, the velocity vector $u(x)$ must point inside (or is tangent to) the free space \citep{souVeloCones}. 
In what follows, we rely on Nagumo's theorem to prove the safety of the trajectories generated by our closed-loop system.
\begin{lemma}[Safety]\label{lem3}
Consider the closed set $\mathcal{F}$ described in \eqref{8}
and the kinematic system \eqref{12} under the control law \eqref{36}. Then, the closed-loop system \eqref{12}-\eqref{36} admits a unique solution for all $t\geq0$ and the set $\mathcal{F}$ is forward invariant.
\end{lemma}
\begin{proof}
See Appendix \ref{appendix:lemma 3}.
\end{proof}
\begin{figure}[!h]
\centering
\includegraphics[scale=0.4]{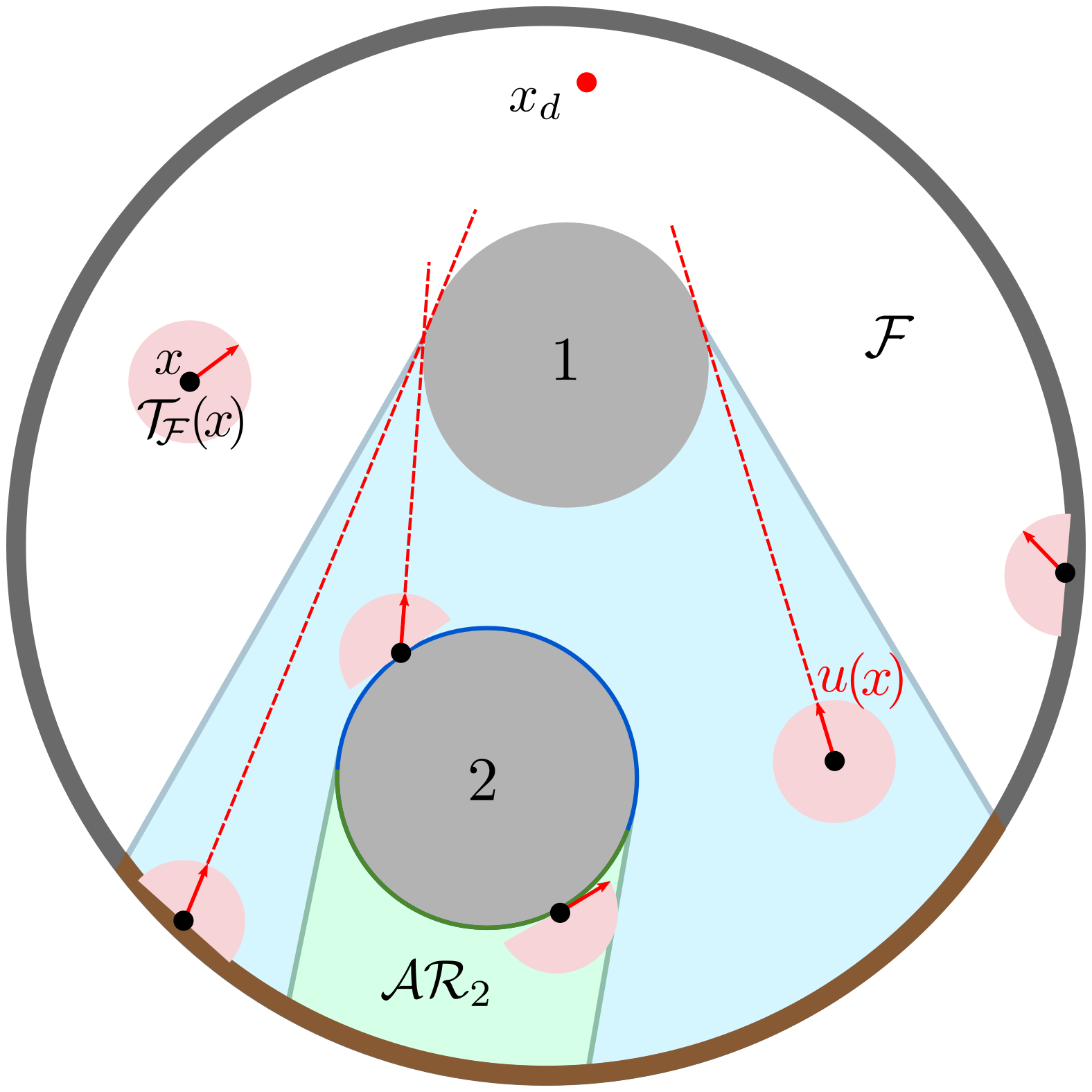}
\caption{Bouligand's tangent cones.}
\label{fig:Tangent_cones}
\end{figure}
Let us look for the equilibria of the closed-loop system \eqref{12}-\eqref{36} by setting $u(x)=0$ in \eqref{36}. Then, from the first equation of \eqref{36}, the equilibrium point is $x_d$. From \eqref{eq:recursive-control}, one can rewrite the control at step $p\in\{1,\dots,h(x)\}$ and position $x\in\mathcal{AR}_{\iota_x(p)}$, as $u_p=\sin(\beta_i)\sin^{-1}(\theta_i)\|u_{p-1}\|\bar\xi_i$\footnote{For simplicity, the arguments $(x,u)$ for the angles $\beta_i$ and $\theta_i$ are omitted whenever
clear from context.} where $\iota_x(p)=i$. In the case where $u_{p-1}\neq0$, and since $\bar\xi_i\in\mathbb{S}^{n-1}$, $u_p=0$ if and only if $\beta_i=0$. The set of positions leading to $\beta_i=0$ is the segment (or segments) of the line, tangent to the ancestor obstacle $k=\iota_x(p-1)$, crossing the center of obstacle $i$,  within the active region of obstacle $i$.  
When $\beta_i=0$, the control input, at step $p-1$, is aligned with $(c_i-x)$, which is also tangent to the ancestor of obstacle $i$. 
We define the set of undesired equilibria (shown in Fig. \ref{fig:fig3}) generated by obstacle $i$ as follows: 
\begin{align}\label{m22}
     \mathcal{L}_i:=\left\{q\in\mathcal{AR}_i|\;\; \beta_i(u_{p-1}(q),q)=0,p=\iota^{-1}_q(i)\right\}.
\end{align}
The central half-line generated by obstacle $i$ in the workspace, starts from the center $c_i$ and extends the set of undesired equilibria $\mathcal{L}_i$ (as shown in Fig. \ref{fig:obswithnoequilibriumpt}), and is defined as $\mathcal{L}_i^e:=\mathcal{L}_h(c_i,y-c_i)$, where $y\in\mathcal{L}_i$. Some obstacles may not generate undesired equilibria, in specific configurations, as will be shown later, and in this case $\mathcal{L}_i$ and $\mathcal{L}_i^e$ are empty sets. Therefore, $u(x)=0$ if $x\in\mathcal{L}_i$ where $i\in\mathcal{Z}$ and $\mathcal{Z}$ is the set of obstacles generating undesired equilibria. Finally, one can conclude that the set of equilibrium points of the system \eqref{12}-\eqref{36} is given by
$
     \zeta:=\{x_d\}\cup(\cup_{i\in\mathcal{Z}}\mathcal{L}_i)
$. The previous developments can be summarized in the following lemma:
\begin{lemma}\label{lem4}
All trajectories of the closed-loop system \eqref{12}-\eqref{36} converge to the set $ \zeta=\{x_d\}\cup(\cup_{i\in\mathcal{Z}}\mathcal{L}_i)$.\hfill\(\Box\)
\end{lemma}

The indices of obstacles crossed by the central half-line $\mathcal{L}^e_i$ of obstacle $\mathcal{O}_i$ are grouped in the set defined as $\mathcal{M}_i:=\left\{j\in\mathbb{I}\setminus\{i\}\|\mathcal{L}^e_i\cap\mathcal{O}_j\neq\varnothing \right\}$ and $N_i=\textbf{card}(\mathcal{M}_i)$. 
Define the map $\kappa_i:\mathcal{M}_i\rightarrow \{1,\dots,N_i\}$ that associates to each index $k\in\mathcal{M}_i$ the corresponding order of the obstacle $\mathcal{O}_k$ according to its proximity with respect to obstacle $\mathcal{O}_i$ among the obstacles of indices in the set $\mathcal{M}_i$, where the order goes from the closest to the farthest obstacle.  The set $\mathcal{M}_i^p:=\{\kappa_i^{-1}(1),\dots,\kappa_i^{-1}(p)\}$, $p\leq N_i$, contains the indices of the set $\mathcal{M}_i$ representing the $p$ first obstacles in increasing order of their distance from obstacle $i$, and $\mathcal{M}_i^0:=\varnothing$. In the following lemma, we show that under certain conditions, obstacles in the set $\mathcal{M}_i$ can be spared from generating undesired equilibria.
\begin{lemma}\label{lem5}
Let $i\in\mathbb{I}$ such that $\mathcal{M}_i\neq\varnothing$. Obstacles of indices in the set $\mathcal{M}_i^p$, where $p\leq N_i$, do not generate undesired equilibria if, for all $k\in\mathcal{M}_i^p$, the following conditions are satisfied:
\begin{enumerate}
    \item $c_k\in\mathring{\mathcal{H}}(x^*_{k,i},c_i)\cup(\cup_{j\in\mathcal{M}_i^{p-1}}\mathring{\mathcal{H}}(x^*_{k,i},c_j))$,
    \item $(\mathring{\mathcal{H}}(x^*_{k,i},c_i)\cup(\cup_{j\in\mathcal{M}_i^{p-1}}\mathring{\mathcal{H}}(x^*_{k,i},c_j)))\cap\mathcal{O}_l=\varnothing$ for all $l\in\mathbb{I}\setminus(\mathcal{M}_i^{p}\cup\{i\})$,
\end{enumerate}
where $x^*_{k,i}=\mathrm{arg}\max\limits_{q\in\mathcal{L}_i^e\cap\partial\mathcal{O}_k}||c_i-q||$. Moreover, if $p=N_i$, or ($p<N_i$ and the obstacle of index $k=\kappa_i^{-1}(p+1)$ does not satisfy conditions 1) and 2)), the set $\mathcal{M}_i$ is said to be of order $\Bar{N}_i=p$ which is the total number of obstacles, of indices in the set $\mathcal{M}_i$, that do not generate undesired equilibria and the set $\mathcal{M}_i^{\Bar{N}_i}$ groups them.
\end{lemma}
\begin{proof} See Appendix \ref{appendix:Lemma 5}.
\end{proof}
Lemma \ref{lem5} provides sufficient conditions for the first $p$ obstacles, with indices in the set $\mathcal{M}_i$ and ordered according to their proximity with respect to obstacle $\mathcal{O}_i$, to be free of undesired equilibria, and if $p=N_i$, or ($p<N_i$ and the $(p+1)-th$ obstacle does not satisfy these conditions), the set $\mathcal{M}_i^{\Bar{N}_i}$ groups all the obstacles, with indices in the set $\mathcal{M}_i$, which do not generate undesired equilibria where $\Bar{N}_i=p$ is the number of these obstacles and the order of the set $\mathcal{M}_i$. Condition 1)
    requires the center of each obstacle $k\in\mathcal{M}_i^p$ to be inside the union of the hats of the cones, of vertex $x^*_{k,i}$, enclosing obstacle $i$ and the obstacles of the list $\mathcal{M}_i^p$ closer to obstacle $i$ than obstacle $k$. Condition 2) requires that the union of hats considered in condition 1) does not intersect any obstacle other than those considered in condition 1)  ({\it i.e.,} obstacles $i$, $k$, and the obstacles closer to obstacle $i$ than obstacle $k$ among the list $\mathcal{M}_i^p$). Let us use obstacle $1$ in Fig. \ref{fig:obswithnoequilibriumpt} to verify (visually) the two conditions. The union of the hats of cones enclosing obstacles $3$ and $4$ (blue and green conic subsets in the left figure) includes the center of obstacle $1$ and does not intersect with any obstacle other than obstacles $3$, $1$, and $4$. Obstacle $4$ satisfies the conditions, but obstacle $2$ does not, as its center is outside the union of the hats enclosing obstacles $3$, $4$, and $1$ (red, blue, and green conic subsets in the right figure).
\begin{figure}[!h]
\centering
\includegraphics[scale=0.5]{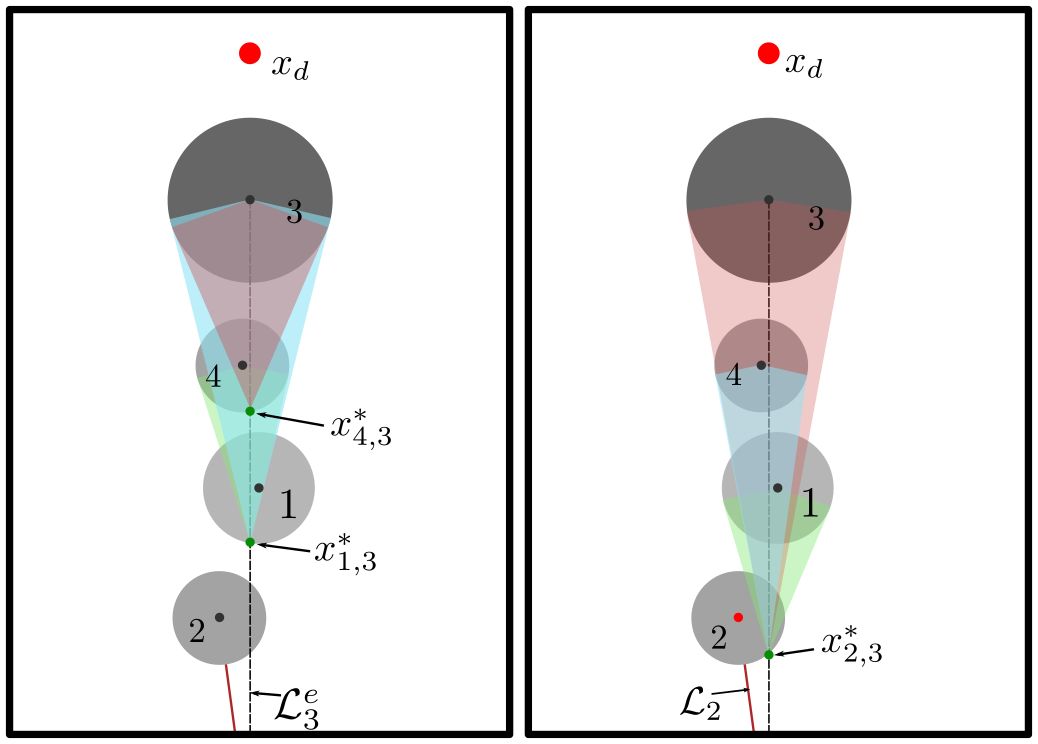}
\caption{Illustration of obstacles that generate, and those that do not generate, undesired equilibria.}
\label{fig:obswithnoequilibriumpt}
\end{figure}
Our main result in $n$-dimensional workspaces is stated in the following theorem.

\begin{theorem}\label{the1}
Consider the free space $\mathcal{F}\subset\mathbb{R}^n$ described in \eqref{8}, for $n\geq2$, and the closed-loop system \eqref{12}-\eqref{36}. Under Assumptions \ref{as:1} and \ref{as:2}, the following statements hold:
\begin{itemize}
\item [i)] The set $\mathcal{F}$ is forward invariant.
\item [ii)] All trajectories converge to the set $ \zeta=\{x_d\}\cup(\cup_{i\in\mathcal{Z}}\mathcal{L}_i)$.
\item [iii)] The set of equilibrium points $\cup_{i\in\mathcal{Z}}\mathcal{L}_i$ is unstable.
\item[iv)] The equilibrium point $x_d$ is locally exponentially stable on $\mathcal{F}$.
\item[v)] The generated trajectories are {\it quasi-optimal}.
\end{itemize}
\end{theorem}
\begin{proof}
See Appendix \ref{appendix:the1}.
\end{proof}
Theorem \ref{the1} shows that the desired equilibrium point $x_d$ is locally exponentially stable and that all trajectories converging to it are safe and {\it quasi-optimal}, in the sense of Definition \ref{def2}. The region of attraction of the desired equilibrium is characterized in the next section for two-dimensional workspaces. Unfortunately, a complete characterization of the region of attraction has not been proved for higher dimensions $n\geq3$. Nevertheless, our insights and extensive simulations in three-dimensional environments led us to conjecture that the equilibrium point $x_d$ is almost globally asymptotically stable, at least for $n=3$.

\subsection*{Invariant sets in two-dimensional spaces ($n=2$)}
Let $\mathcal{R}_i:=\left\{k\in\mathcal{Z}|\mathcal{L}^e_k\cap\mathcal{O}_i\neq\varnothing,\,\mathcal{L}_k\cap\mathcal{AR}_i\neq\varnothing\right\}$ be the set of indices of central half-lines crossing obstacle $i$ and their set of undesired equilibria intersecting with its active region $\mathcal{AR}_i$, and note that $\mathcal{R}_i\neq\varnothing$ for all $i\in\mathbb{I}$.
Obstacles crossed by more than one central half-line are represented by the set of indices $\mathbb{L}:=\left\{k\in\mathbb{I}|\textbf{card}(\mathcal{R}_k)\geq2\right\}$. For every $i\in\mathbb{L}$, we select the out-most line segments $\mathcal{L}_k$, $k\in\mathcal{R}_i$, and we determine their intersection with the boundary of obstacle $i$, the left and right intersections being denoted by $y^l_{i,0}$ and $y^r_{i,0}$ respectively. We go through the two out-most line segments separately until they intersect with one of the line segments having an index in the set $\mathcal{R}_i$, or with the boundary of the workspace. We denote the left and right intersections by $y^l_1$ and $y^r_1$, respectively. If the workspace has yet to be reached and ($y^l_1\neq y^r_1$), we continue in the same way with the new line segments up to the intersection with the workspace boundary or up to the intersection between the left and right line segments ({\it i.e.}, $y^l_k=y^r_j$, $k,p>0$). We group the intersection points obtained on the left and right into two lists, $Y_i^l=\{y^l_{i,0},y^l_{i,1},\dots\}$ and $Y_i^r=\{y^r_{i,0},y^r_{i,1},\dots\}$, respectively (see Fig. \ref{fig:Inv_nests}). For every two successive points $\{y^l_{i,p},y^l_{i,p+1}\}$ of $Y_i^l$ (resp. $\{y^r_{i,p},y^r_{i,p+1}\}$ of $Y_i^r$), we generate the right (resp. left) half-plane bounded by the line passing through these two points. The intersection of the union of the half-planes of each list forms an area that, when restricted to the active region, gives a characteristic region defined as $\chi_i:=\left(\cup_{p=0}^{\textbf{card}(Y^r_i)-2}\mathcal{P}_{\scriptscriptstyle \geq}\left(y^r_{i,p},R(y^r_{i,p}-y^r_{i,p+1})\right)\right)\cap\left(\cup_{p=0}^{\textbf{card}(Y^l_i)-2}\mathcal{P}_{\scriptscriptstyle \leq}\left(y^l_{i,p},R(y^l_{i,p}-y^l_{p+1})\right)\right)\cap\mathcal{AR}_i$ where $R=\big(\begin{smallmatrix}
  0 & 1\\
  -1 & 0
\end{smallmatrix}\big)$. Finally, we create a cell, deleting the characteristic regions of other obstacles inside the characteristic region of obstacle $i$, and define it as follows $\textbf{Cell}_i:=\chi_i\setminus\cup_{k\in\mathbb{L}_i}\mathring{\chi}_k$ where $\mathbb{L}_i:=\left\{k\in\mathbb{L}|\chi_i\cap\chi_k\neq\varnothing;\forall x\in\chi_i\cap\chi_k,\,\iota_x^{-1}(k)>\iota_x^{-1}(i)\right\}$. Note that the construction of these cells requires their boundaries to be formed by undesired equilibria and the boundary of the free space, which endows them with the invariance property stated in the following lemma.
\begin{lemma}\label{lem6}
    Let $i\in\mathbb{L}$. The cell $\textbf{Cell}_i$ is forward invariant for the closed-loop system \eqref{12}-\eqref{36}.
\end{lemma}
\begin{proof}
See Appendix \ref{appendix:lemma 6}.
\end{proof}
\begin{lemma} \label{lem7}
    The set $\cup_k\textbf{Nest}_k$ is the region of attraction of the undesired equilibria $\cup_{i\in\mathcal{Z}}\mathcal{L}_i$.
\end{lemma}
\begin{proof}
See Appendix \ref{appendix:lemma 7}.
\end{proof}
\begin{figure}[!h]
\centering
\includegraphics[scale=0.27]{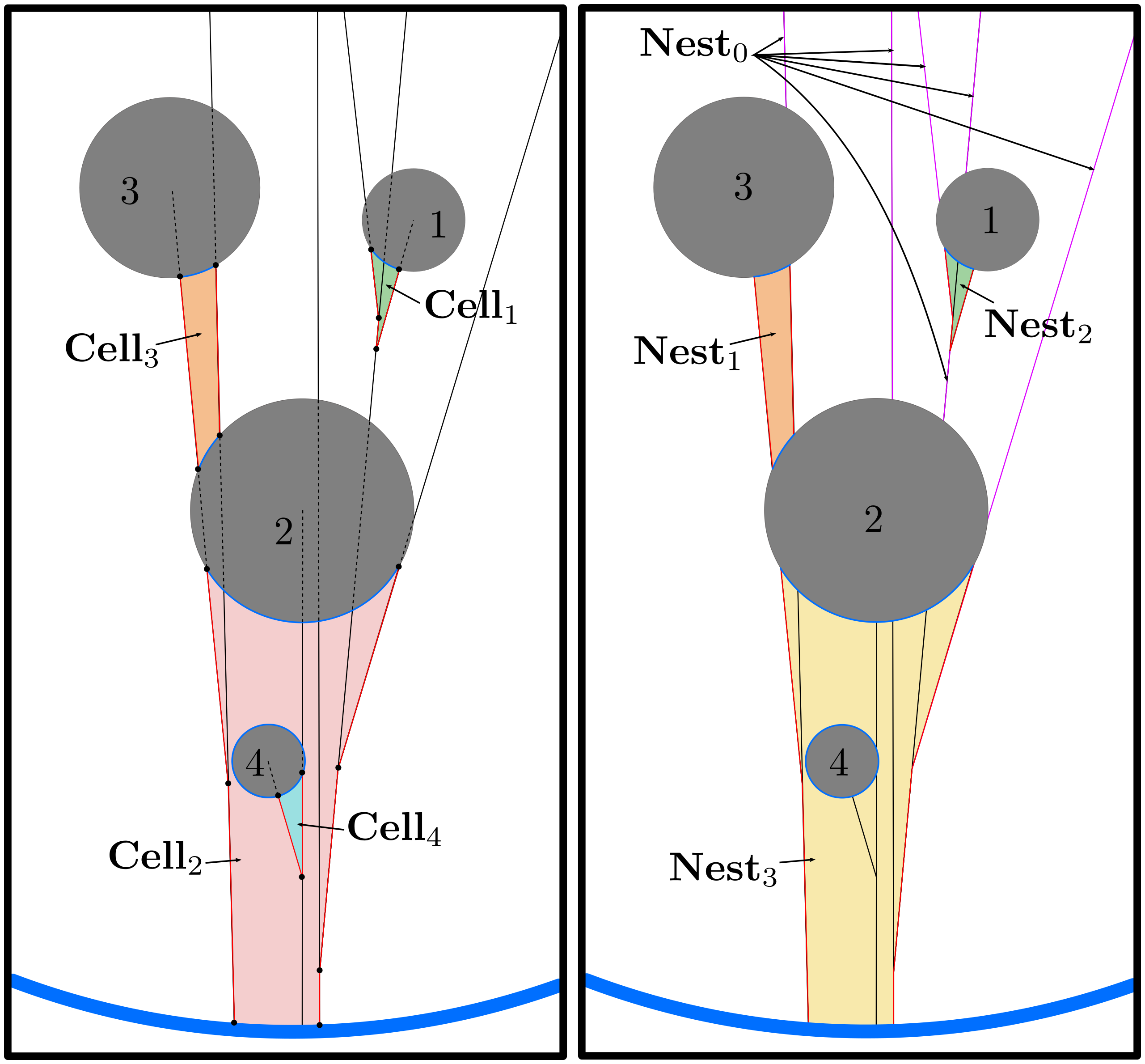}
\caption{Invariant cells and nests}
\label{fig:Inv_nests}
\end{figure}
Two cells are adjacent if they share undesired equilibria on their boundary, which is true only if $\partial\textbf{Cell}_i\cap\partial\textbf{Cell}_k\cap(\cup_{i\in\mathcal{Z}}\mathcal{L}_i)\neq\varnothing$. We construct nests by the union of adjacent cells, where each cell has at least one adjacent cell among the cells in that nest. Cells without adjacent cells form a nest with a single element. We also construct a special nest whose cells are segments of undesirable equilibria that belong to no other regular cell.
Since nests are the union of invariant cells or of undesired equilibria (the special nest), nests are invariant and are denoted by $\textbf{Nest}_k$ where 
$\textbf{Nest}_0:=\cup_{i\in\mathcal{Z}}\mathcal{L}_i\setminus\cup_{k\in\mathbb{L}}\textbf{Cell}_k$ is the special nest (see Fig. \ref{fig:Inv_nests}). Unfortunately, a nest can form a barrier around the workspace, reducing the navigable area of the free space. Such a nest can be generated by creating a circular band of adjacent cells, as shown in Fig. \ref{fig:Quasi_nn_navig}. In the following lemma, nests are shown to be the attraction region of the undesired equilibria.

\begin{figure}[!h]
\centering
\includegraphics[scale=0.4]{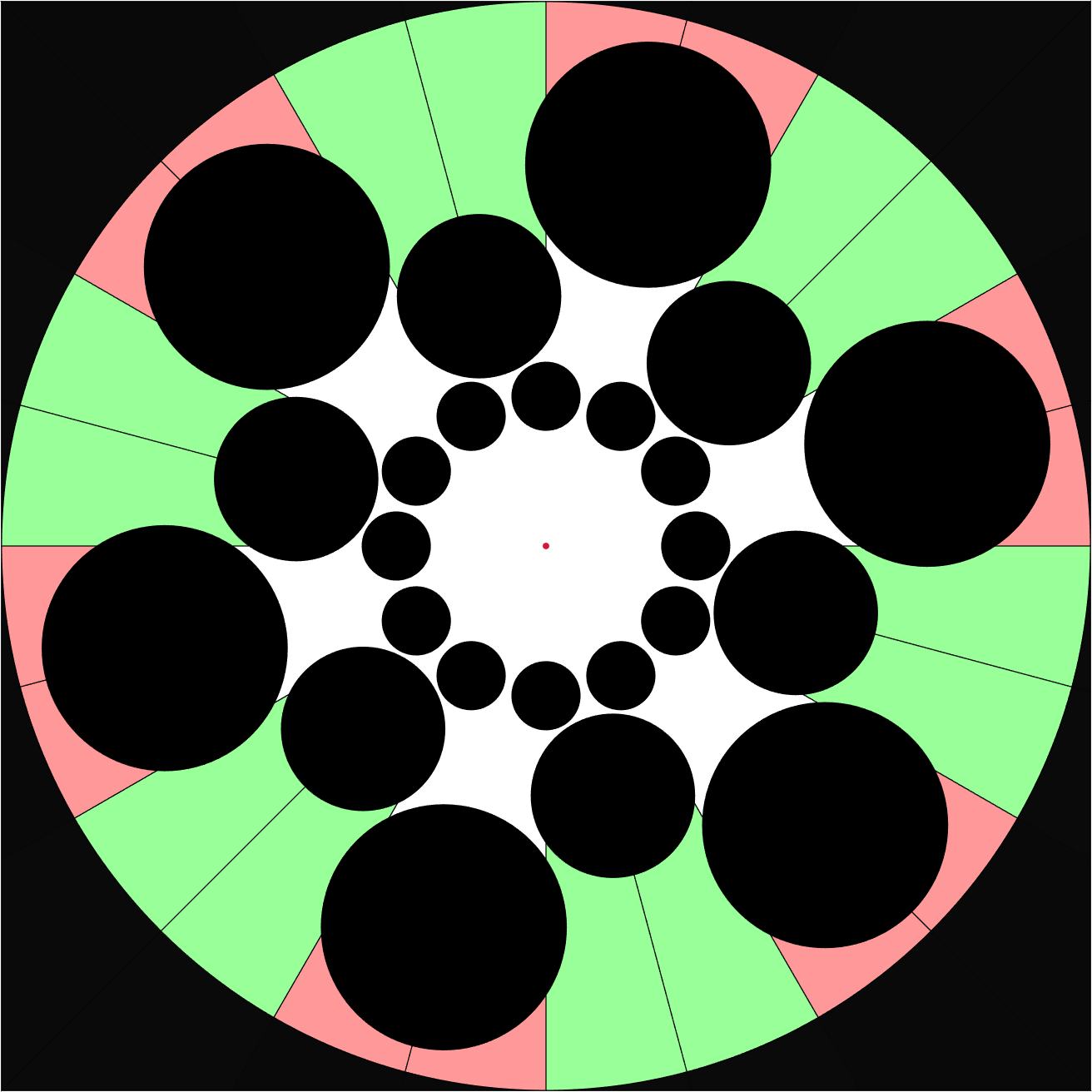}
\caption{Quasi-non-navigable two-dimensional space.}
\label{fig:Quasi_nn_navig}
\end{figure}
Now, to ensure almost global asymptotic stability of the equilibrium point $x_d$ in two-dimensional spaces, we reduce the nests to the set of undesired equilibria by imposing the following assumption:
\begin{assumption}\label{as:3}
 For any $i\in\mathbb{I}$ and $k\in\mathcal{Z}$ where $i\neq k$, $\mathcal{L}^e_k\cap\mathcal{O}_i=\varnothing$, or ($\mathcal{L}^e_k\cap\mathcal{O}_i\neq\varnothing$ and $i\in\mathcal{M}_k^{\Bar{N}_i}$.)
\end{assumption}
Assumption \ref{as:3} rules out the possibility of creating the invariant cells by imposing obstacle configurations such that $\mathbb{L}=\varnothing$ making the undesired equilibria repellers. In addition to the results of Theorem \ref{the1}, the next theorem characterizes the attraction region of the undesired equilibria and shows almost global asymptotic stability of the destination under Assumption \ref{as:3}. 
\begin{theorem}\label{the2}
Consider the free space $\mathcal{F}\subset\mathbb{R}^n$ described in \eqref{8}, for $n=2$, and the closed-loop system \eqref{12}-\eqref{36}. Let Assumptions \ref{as:1} and \ref{as:2} hold. Then, items i), ii), iii), iv) of Theorem \ref{the1}, and the following statements hold:
\begin{itemize}
\item [i)] The equilibrium point $x_d$ is attractive from all $x(0)\in\mathcal{F}\setminus\cup_{k}\textbf{Nest}_k$.
\item [ii)] From any initial position $x(0)\in\mathcal{F}\setminus\cup_{k}\textbf{Nest}_k$, the trajectory $x(t)$ is {\it quasi-optimal}.
\item [iii)] Under Assumption \ref{as:3}, $\cup_k\textbf{Nest}_k=\cup_{i\in\mathcal{Z}}\mathcal{L}_i$ and the destination $x_d$ is almost globally asymptotically stable.
\end{itemize}
\end{theorem}
\begin{proof}
See Appendix \ref{appendix:the2}.
\end{proof}
Theorem \ref{the2} shows the attraction of the target location from any position in the free space, except for the nests (region of attraction of the undesired equilibria), which reduces, under Assumption \ref{as:3}, to the undesired equilibria $\cup_{i\in\mathcal{Z}}\mathcal{L}_i$ having measure zero.
Fortunately, the nests will naturally disappear in the sensor-based case as we will see in the next section.
\section{Sensor-based navigation using a 2D LiDAR range scanner}
 We now present a more practical version of our approach using a LiDAR range scanner in an unknown two-dimensional sphere world. Assume that the robot is equipped with a sensor of $360^{\circ}$ angular sensing range, a resolution $d\theta>0$, and a radial sensing range $R>0$. The measurements of the sensor, at a position $x$, are modeled by the polar curve $\rho(x,\theta):\mathcal{F}\times\hat{\mathcal{A}}\rightarrow[0, R]$, where $\hat{\mathcal{A}}:=\left\{0,d\theta,2d\theta,\dots,360-d\theta\right\}$ is the set of scanned angles, defined as follows:
\begin{align}
    \rho(x,\theta):=\min\left(\begin{array}{l}
        R,
        \min\limits_{\substack{y\in\partial\mathcal{F}\\\overline{\mathrm{atan2}}(y-x)=\theta}}\|x-y\|,
    \end{array}\right),
\end{align}
where $\overline{\mathrm{atan2}}(v)=\mathrm{atan2}(v(2),v(1))$ for $v\in\mathbb{R}^2$.\\
The Cartesian coordinates of the scanned points are modeled by the mapping  
 $\delta(x,\theta):\mathcal{F}\times\hat{\mathcal{A}}\rightarrow\mathcal{F}$ defined as follows:
 \begin{align}
    \delta(x,\theta):=x+\rho(x,\theta)[\cos(\theta)\;\sin(\theta)]^{\top}.
\end{align}
Let $G_x(\delta)$ be the graph of the mapping $\delta$ at a position $x$ (red curve in Fig. \ref{Sensor_process}). The set $\mathbb{I}_x\subset\mathbb{I}$ of the detected obstacles is defined as $\mathbb{I}_x:=\left\{i\in\mathbb{I}|d(x,\mathcal{O}_i)\leq R\right\}$. Assume that at each position $x$, the sensor returns a list of arcs $\mathcal{LA}(x):=\left\{L_1,L_2,\dots,L_{\tau(x)}\right\}$ from the detected obstacles corresponding to the intersection of the graph $G_x(\delta)$ and obstacles of the set $\mathbb{I}_x$, where $\tau(x)=\textbf{card}(\mathbb{I}_x)$ as shown in Fig. \ref{Sensor_process}-\subref{sensor2} by the magenta arcs. Since the available information about the environment is limited by the graph $G_x(\delta)$, successive projections are impossible to apply. Therefore, we apply the single obstacle control strategy given by
\begin{align*}
    u(x)= \begin{cases}
      \displaystyle u_d(x), & x\in\mathcal{VI},\\
      \displaystyle u_d(x)-\|u_d(x)\|\frac{\sin(\theta_i-\beta_i)}{\sin(\theta_i)}V_{ci}, & x\in\mathcal{D}(x_d,c_i),
    \end{cases} 
\end{align*}
where $\theta_i$ and $V_{ci}=(c_i-x)/\|c_i-x\|$ are, respectively, the aperture and the axis of the enclosing cone, $\beta_i$ is the angle between $u_d$ and $(c_i-x)$, and $\mathcal{D}(x_d,c_i)$ is the shadow region.
\begin{figure}[h!]
     \centering
     \subfloat[]{\includegraphics[width=0.4\linewidth,height=0.4\linewidth,keepaspectratio]{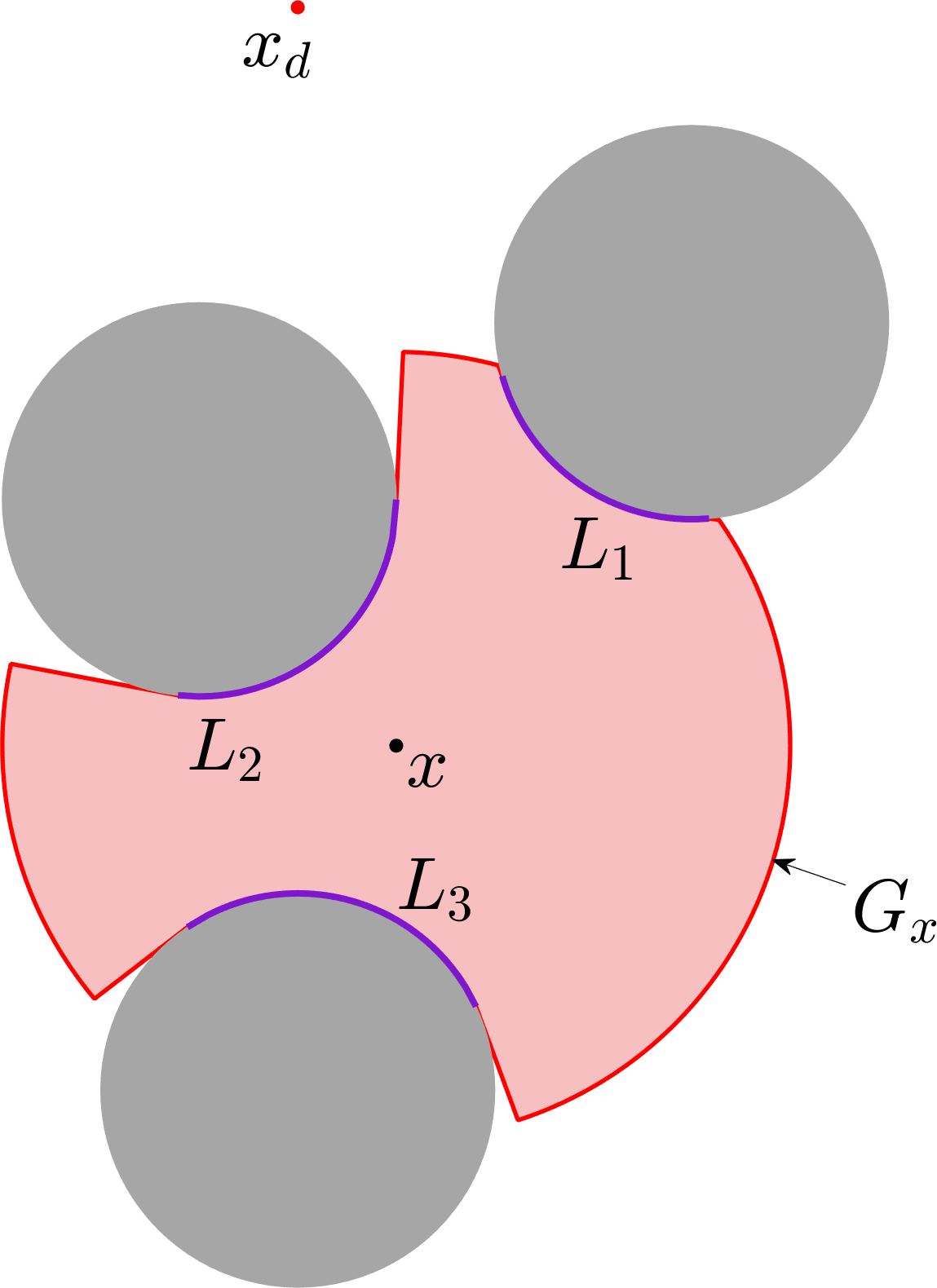}\label{sensor2}}
      \hspace{0.01cm}
     \subfloat[]{\includegraphics[width=0.4\linewidth,height=0.4\linewidth,keepaspectratio]{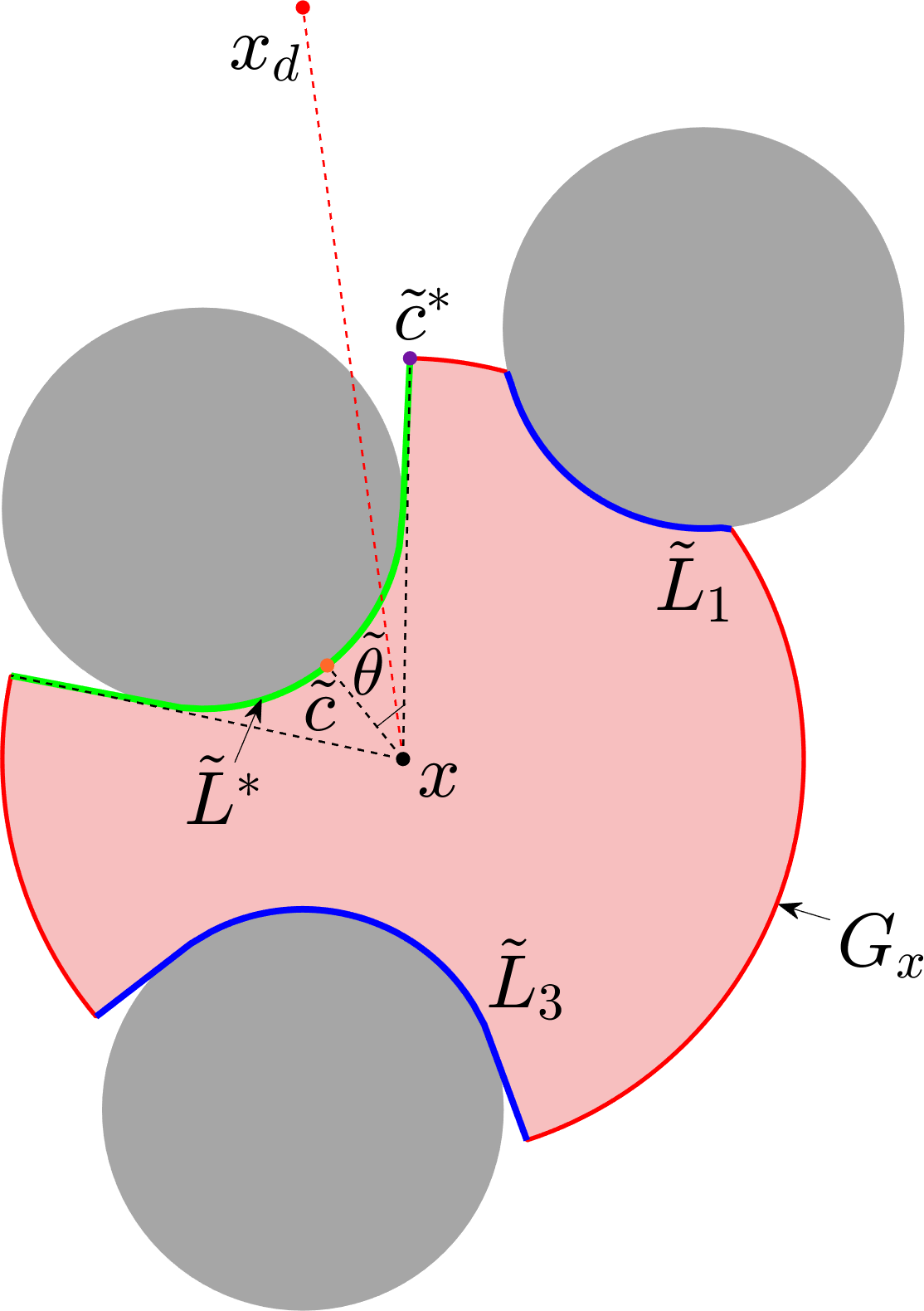}\label{sensor5}}
     \caption{Sensor-based control procedure for our approach.}
     \label{Sensor_process}
\end{figure}
To adapt the above control strategy to the sensor-based case, 
one proceeds as follows. At each position $x$, the detected arcs $\mathcal{LA}(x)$ are considered as obstacles. The arc crossed by the segment $\mathcal{L}(x,x_d)$ will help to create a virtual enclosing cone onto which the projection is performed. However, due to the practical model of the sensor, which may have low resolutions, safety is not always guaranteed when the robot is in the neighborhood of the obstacles where the velocity vector (projection of $u_d$ onto the virtual enclosing cone) may point inside the obstacle (see Fig.\ref{Pb_sensor}-\subref{fig:Pb_s1_lq}). To overcome this problem, a list of extended arcs $\mathcal{LA}_e(x):=\left\{\tilde{L}_1,\tilde{L}_2,\dots,\tilde{L}_{\tau(x)}\right\}$ is defined, where the endpoints of an arc $L_i$ are moved through the graph $G_x(\delta)$ until they have a radial polar coordinate equal to $R$ or they meet the endpoints of the neighboring arcs, as shown in Fig. \ref{Sensor_process}-\subref{sensor5}. Among the extended arcs of the list $\mathcal{LA}_e(x)$, the active extended arc crossed by the segment $\mathcal{L}(x,x_d)$ is selected and denoted by $\tilde{L}^*$. The active extended arc serves as an obstacle enclosed by a virtual cone (see Fig. \ref{Sensor_process}-\subref{sensor5}) from which we extract the following practical parameters:
\begin{itemize}
    \item The virtual center
    \begin{align}\label{vrt_center}
    \tilde{c}:=\mathrm{arg}\min\limits_{y\in\tilde{L}^*}\|x-y\|,
    \end{align}
    which gives the direction $(\tilde{c}-x)$.
    \item The virtual aperture
    \begin{align}\label{vrt_aperture}
    \tilde{\theta}:=\angle(\tilde{c}-x,\tilde{c}^*-x),
    \end{align}
    where $\tilde{c}^*$ is the endpoint of $\tilde{L}^*$ such that $u_d$ is between the directions $(\tilde{c}-x)$ and $(\tilde{c}^*-x)$.
    \item The angle
    \begin{align}\label{vrt_theta}
    \tilde{\beta}:=\angle(\tilde{c}-x,u_d).
    \end{align}
\end{itemize}
\begin{figure}[!h]
     \centering
     \subfloat[]{\includegraphics[width=0.35\linewidth,height=0.35\linewidth,keepaspectratio]{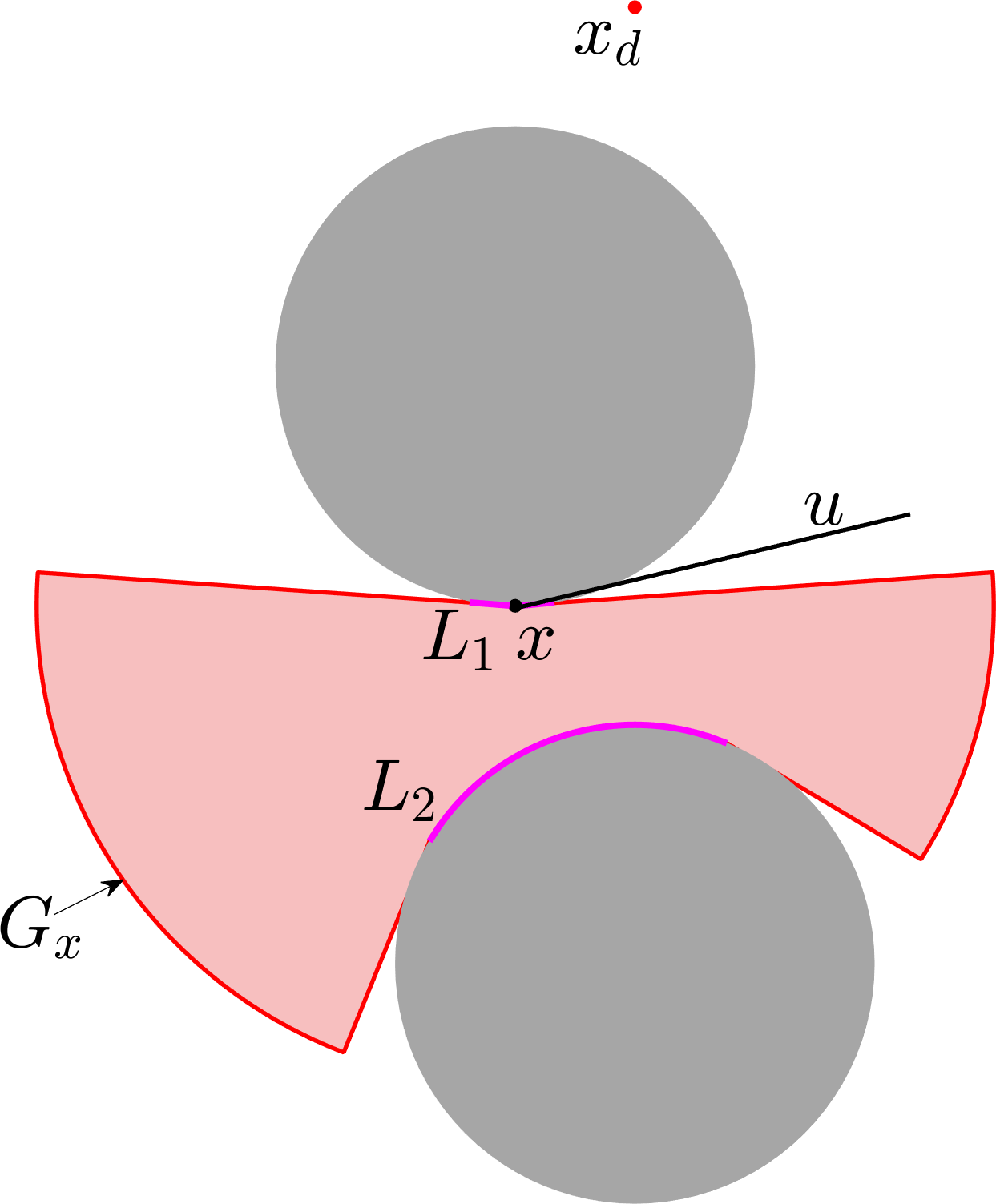}\label{fig:Pb_s1_lq}}
      \hspace{0.01cm}
     \subfloat[]{\includegraphics[width=0.35\linewidth,height=0.35\linewidth,keepaspectratio]{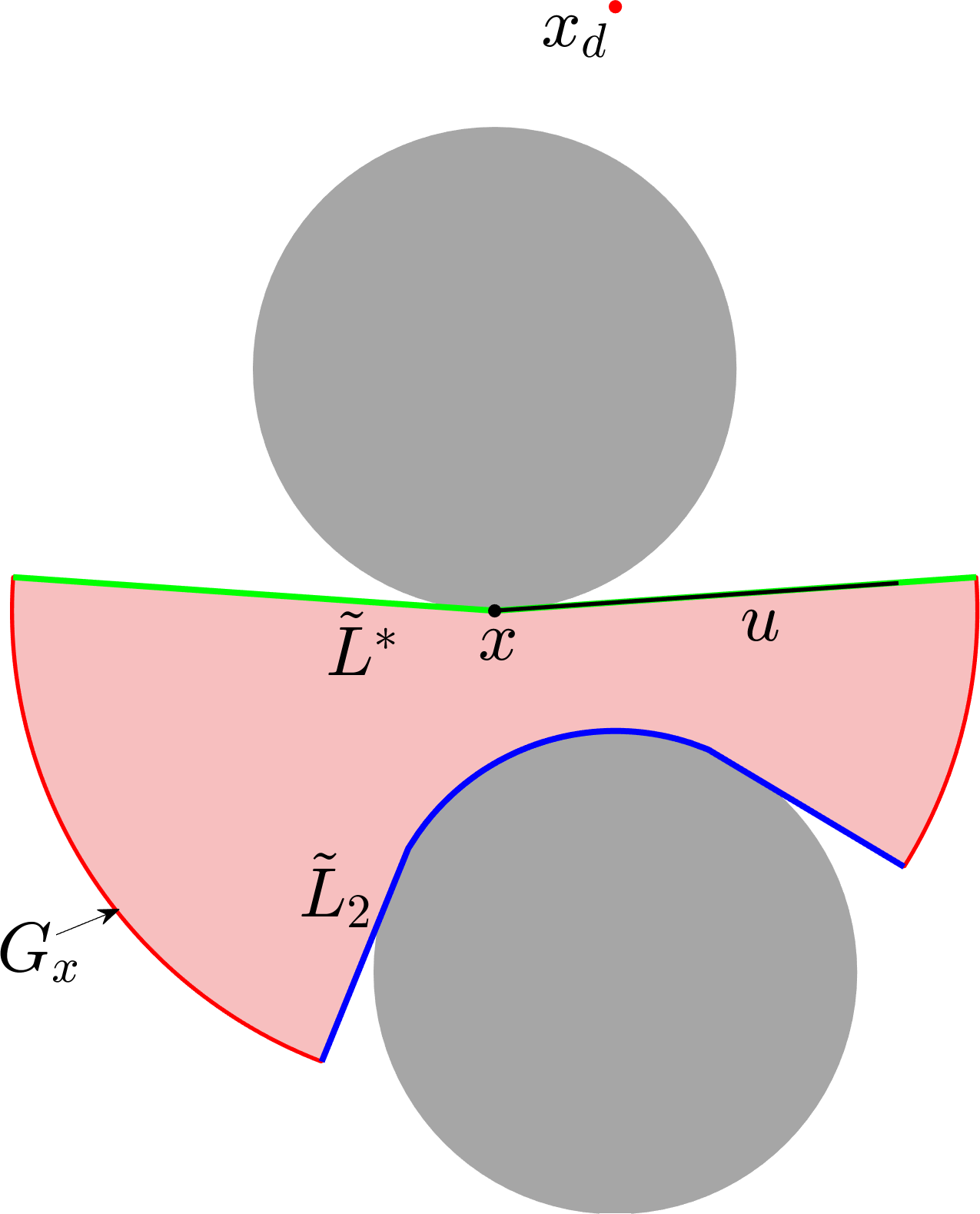}\label{fig:Pb_s2_lq}}
     \caption{Safety consideration in a sensor-based case. In Fig. (a), the projection $u$ of $(x_d-x)$ onto the cone enclosing the arc $L_1$ of the list $\mathcal{LA}(x)$ fails to satisfy the safety condition where $u$ crosses the obstacle. In Fig. (b), the projection lies on the active arc $\tilde L^*$ of the list $\mathcal{LA}_e(x)$ and meets the safety condition.} 
     \label{Pb_sensor}
\end{figure}
Before defining the new blind and visible sets, let us define the truncated shadow region by
\begin{align}\label{tr_sh}
    \mathcal{D}^t(x_d,c_i):=\mathcal{D}(x_d,c_i)\setminus\bigcup\limits_{j\in\mathcal{PR}_i}\mathcal{D}(x_d,c_j),
\end{align}
where $\mathcal{PR}_i:=\bigl\{j\in\mathbb{I}|\mathcal{D}(x_d,c_i)\cap\mathcal{D}(x_d,c_j)\neq\varnothing,\,d(x_d,\mathcal{O}_i)<d(x_d,\mathcal{O}_j)\bigr\}$ is the progeny of obstacle $i$ (see Fig. \ref{Practical workspace}-\subref{truncatedShadow}). Since the visibility of the robot is limited to the scanning range of the sensor, let us define the practical shadow region of an obstacle $i$ as follows:
\begin{align}
    \tilde{\mathcal{D}}(x_d,c_i,R):=\mathcal{D}^t(x_d,c_i)\cap\mathcal{B}(c_i,r_i+R).
\end{align}
Therefore, the practical blind set is defined as follows:
\begin{align}\label{blind_sens}
\widetilde{\mathcal{BL}}:=\bigcup\limits_{i\in\mathbb{I}}\tilde{\mathcal{D}}(x_d,c_i,R),
\end{align}
The practical visible set is then defined as $\widetilde{\mathcal{VL}}:=\widetilde{\mathcal{BL}}_{\mathcal{F}}^c$ (see Fig. \ref{Practical workspace}-\subref{practicalShadow}). Finally, the control is given by
\begin{align}\label{sensor_based_ctrl}
    u(x)= \begin{cases}
      \displaystyle u_d(x), & x\in\widetilde{\mathcal{VI}},\\
      \displaystyle u_d(x)-\|u_d(x)\|\frac{\sin(\tilde{\theta}-\tilde{\beta})}{\sin(\tilde{\theta})}\frac{\tilde{c}-x}{\|\tilde{c}-x\|}, & x\in\widetilde{\mathcal{BL}}.
    \end{cases} 
\end{align}
\begin{figure}[h!]
     \centering
     \subfloat[]{\includegraphics[width=0.38\linewidth,height=0.38\linewidth,keepaspectratio]{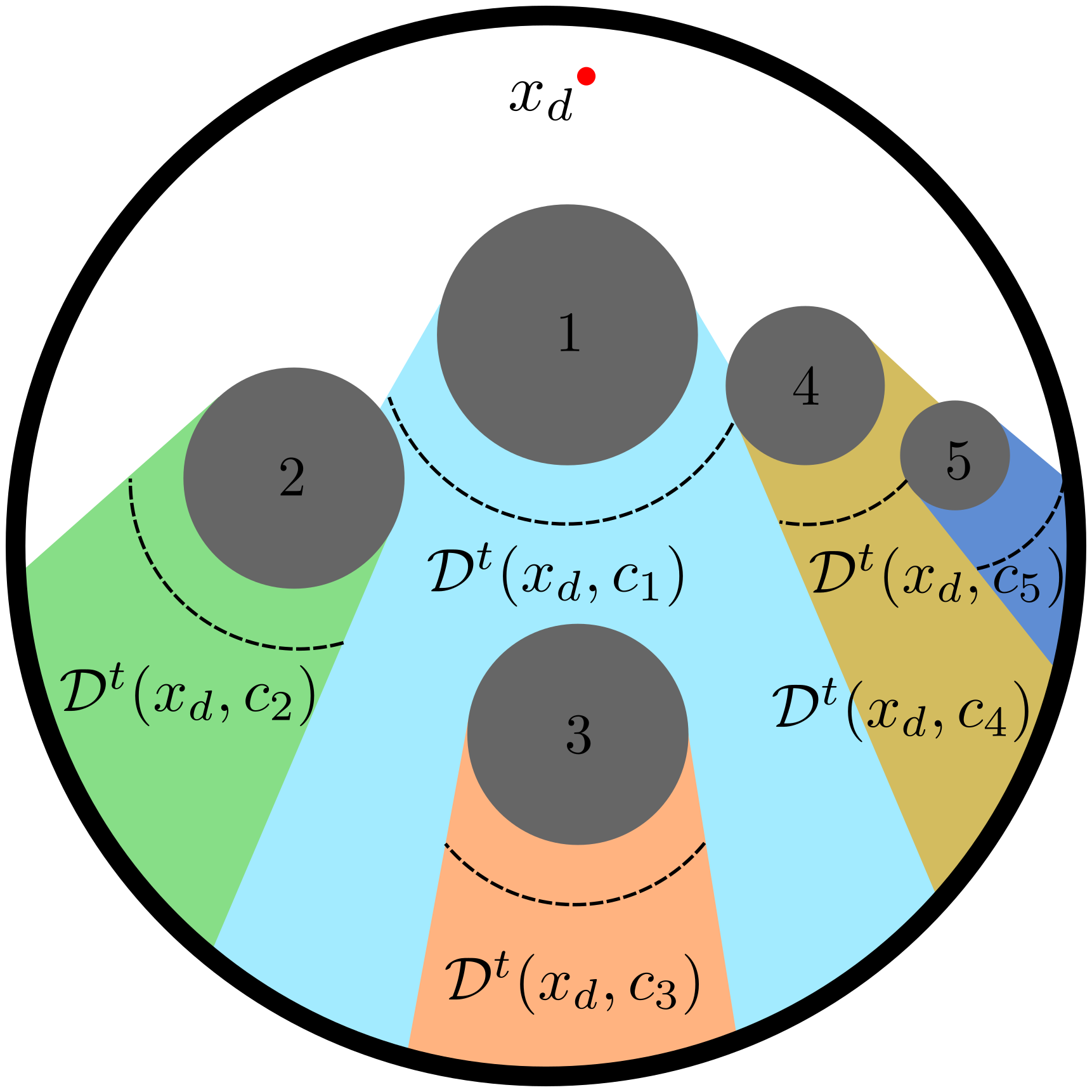}\label{truncatedShadow}}
    \hspace{0.01cm}
     \subfloat[]{\includegraphics[width=0.38\linewidth,height=0.38\linewidth,keepaspectratio]{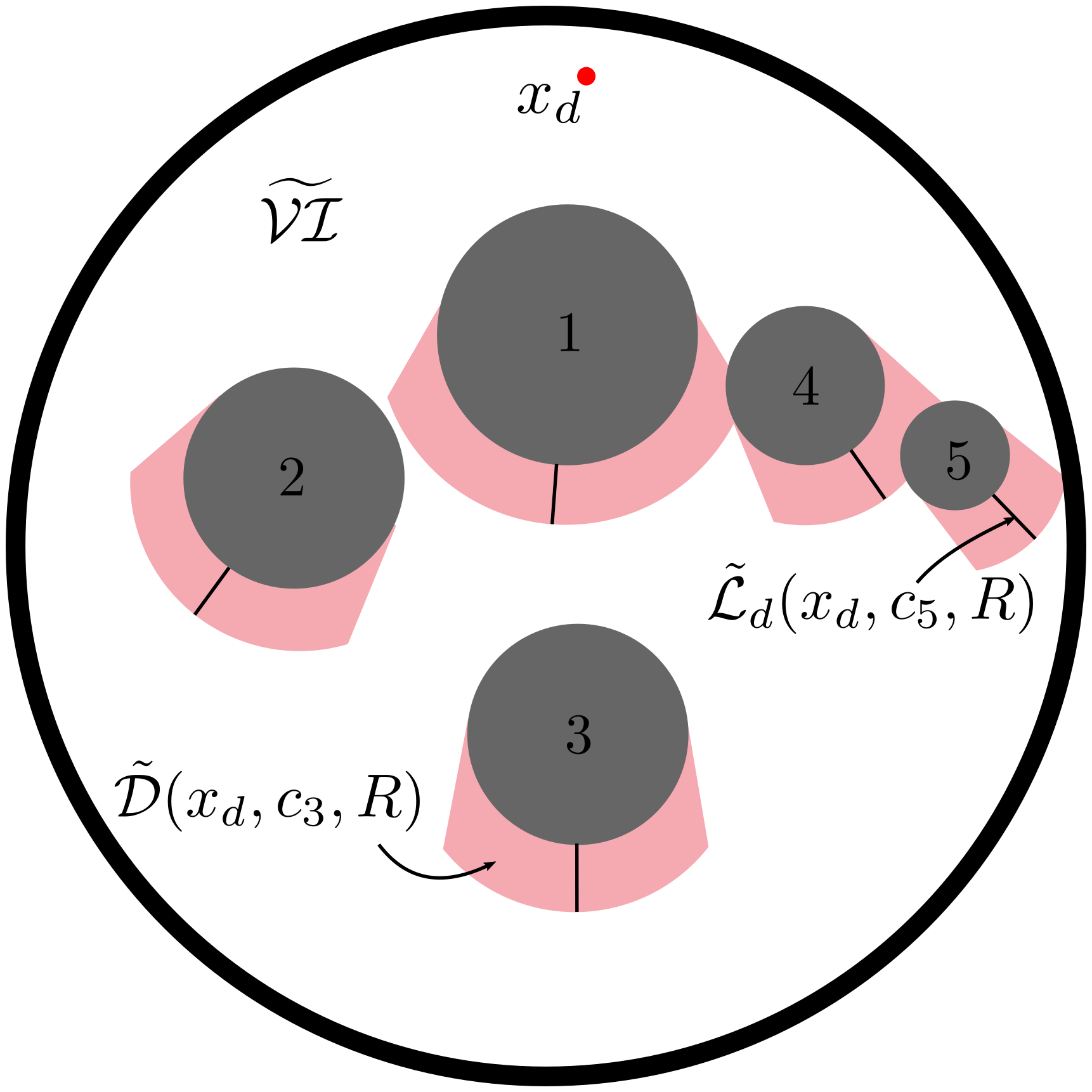}\label{practicalShadow}}\\
     \caption{Illustration of the workspace for the sensor-based case. Fig. (a) shows the truncated shadow regions of each obstacle where obstacles $\{2,3,4\}$ are the progeny of obstacle $1$, and obstacle $\{5\}$ is the progeny of obstacle $4$. Fig. (b) highlights the practical shadow regions of each obstacle in pink where the union of theses regions represents the practical blind set while the white region represents the practical visible set.}
     \label{Practical workspace}
\end{figure}
The implementation of the sensor-based control strategy is summarized in Algorithm \ref{alg_sensor} (see also Fig. \ref{Sensor_process}).\\
\begin{algorithm}[tbh]
 \caption{Implementation of the control law \eqref{sensor_based_ctrl} in the closed-loop \eqref{eq:closed-loop-system}}\label{alg_sensor}
 \begin{algorithmic}[1]
 \renewcommand{\algorithmicrequire}{\textbf{Initialization:}}
 \REQUIRE $x_d$, $e_c$;
 \WHILE{true}
 \STATE Measure $x$, $G_x(\delta)$, and $\mathcal{LA}(x)$.
 \IF {$\|x-x_d\|\leq e_s$}
 \STATE Break;
 \ELSE
  \IF{ $\mathcal{LA}(x)\neq\varnothing$} 
  \STATE Construct $\mathcal{LA}_e(x)$.
  \IF{$\mathcal{L}(x_d,x)$ crosses one of the extended arcs of $\mathcal{LA}_e(x)$}
  \STATE Identify $\tilde{L}^*$.
  \STATE Determine $\tilde{c}$, $\tilde{\theta}$ and $\tilde{\beta}$ using equations \eqref{vrt_center}, \eqref{vrt_aperture} and \eqref{vrt_theta}, respectively.
  \STATE $u\leftarrow u_d(x)-\|u_d(x)\|\frac{\sin(\tilde{\theta}-\tilde{\beta})}{\sin(\tilde{\theta})}\frac{\tilde{c}-x}{\|\tilde{c}-x\|}$;
  \ELSE
  \STATE $u\leftarrow u_d$;
  \ENDIF
  \ELSE
  \STATE $u\leftarrow u_d$;
  \ENDIF
  \STATE Execute $u$ in (9); 
  \ENDIF
  \ENDWHILE
 \end{algorithmic} 
 \end{algorithm}
The sensor-based control strategy \eqref{sensor_based_ctrl} can be seen as a special case of the control strategy in the map-based scenario ({\it a priori } known environments) if each obstacle is considered as the unique obstacle in the workspace. In this way, the active regions become the disjoint practical shadow regions that will limit the undesired equilibria generated by each obstacle to its own practical shadow region excluding the possibility of creating invariant cells. The following lemma determines the equilibria of the closed-loop system \eqref{12}-\eqref{sensor_based_ctrl}.
\begin{lemma}\label{lem8}
    All trajectories of the closed-loop system \eqref{12}-\eqref{sensor_based_ctrl} converge to the set $\zeta=\{x_d\}\cup\left(\cup_{i\in\mathbb{I}}\tilde{\mathcal{L}}_d(x_d,c_i,R)\right)$ where $\tilde{\mathcal{L}}_d(x_d,c_i,R):=\mathcal{L}_d(x_d,c_i)\cap\tilde{\mathcal{D}}(x_d,c_i,R)$. 
\end{lemma}
\begin{proof}
See Appendix \ref{appendix:Lemma 8}.
\end{proof}
Lemma \ref{lem8} shows that the set of undesirable equilibria of the closed-loop system \eqref{12}-\eqref{sensor_based_ctrl}, associated with an obstacle $\mathcal{O}_i$, is the line segment starting from the antipodal point of the destination on obstacle $\mathcal{O}_i$ to the boundary of the practical shadow region. The next theorem states formally the properties of the sensor-based control strategy in two-dimensional sphere worlds. 
\begin{theorem}\label{the3}
Consider the free space $\mathcal{F}\subset\mathbb{R}^n$ described in \eqref{8}, for $n=2$, and the closed-loop system \eqref{12}-\eqref{sensor_based_ctrl}. Under Assumptions \ref{as:1} and \ref{as:2}, the following statements hold:
\begin{itemize}
\item [i)] The set $\mathcal{F}$ is forward invariant.
\item [ii)] All trajectories converge to the set $ \zeta=\{x_d\}\cup\left(\cup_{i\in\mathbb{I}}\tilde{\mathcal{L}}_d(x_d,c_i,R)\right)$.
\item [iii)] The set of undesired equilibria $\cup_{i\in\mathbb{I}}\tilde{\mathcal{L}}_d(x_d,c_i,R)$ is unstable.
\item [iv)] The equilibrium point $x_d$ is almost globally asymptotically stable on $\mathcal{F}$.
\end{itemize}
\end{theorem}
\begin{proof}
See Appendix \ref{appendix:the3}.
\end{proof}
Theorem \ref{the3} provides the stability results obtained with the reactive sensor-based feedback control strategy, relying only on local information provided by the sensor, which is more practical and realistic than the global approach which requires \textit{a priori} knowledge of the entire workspace. Almost global asymptotic stability is guaranteed without imposing restrictions on the obstacle configurations as in Assumption \ref{as:3}. However, the control continuity and quasi-optimality of the generated trajectories are no longer guaranteed. Fig. \ref{Discontinuity} shows an example of a discontinuity in our control at time $t=t'$ when the active arc $\tilde{L}^*$ passed from one obstacle to another, resulting in a sudden change in the control's direction to follow the tangent of the new obstacle.
\begin{figure}[!h]
\centering
\includegraphics[scale=0.35]{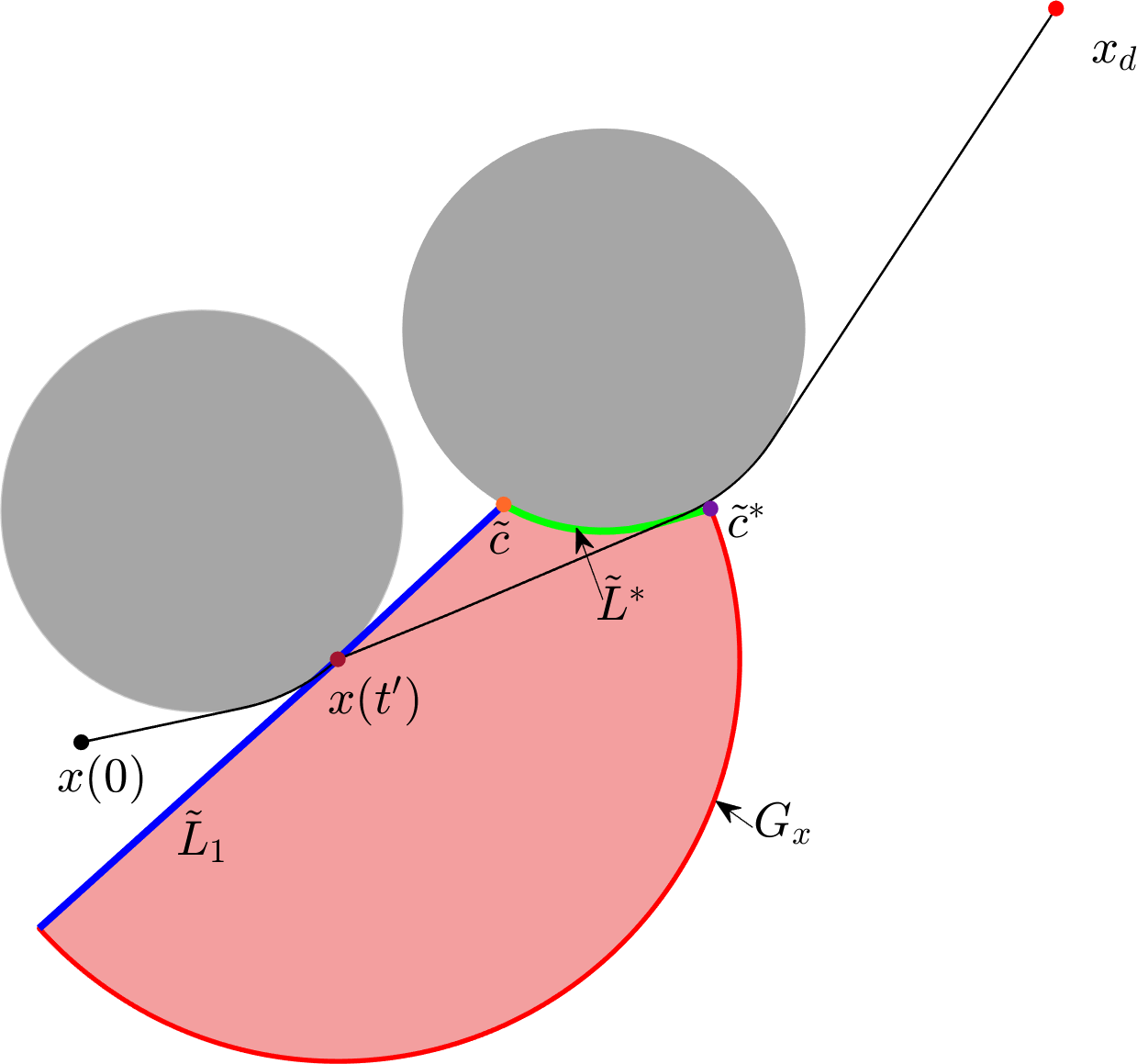}
\caption{A scenario of a discontinuity occurrence in a 2D workspace populated by two obstacles at time $t=t'$.}
\label{Discontinuity}
\end{figure}
\subsection{Convex obstacles}
We consider convex sets with smooth boundaries $\tilde{\mathcal{O}}_i$ as obstacles. The free space must satisfy the separation conditions of Assumptions \ref{as:1} and \ref{as:2}. We also assume that the following curvature condition (see, \textit{e.g.,} \citep{Arslan2019}) is satisfied.
\begin{assumption}\label{as:4}
    Obstacles are sufficiently curved at their farthest point from the target location $x_d\in\mathcal{F}$, {\it i.e.,}
 \begin{align}
    \tilde{\mathcal{O}}_i\subset\mathcal{B}(x_d,\|x_d-\mathrm{x}_i\|),\;\forall i\in\mathbb{I},
\end{align}
where $\mathrm{x}_i\in\left\{q\in\partial\tilde{\mathcal{O}}_i|ds_i(q)^{\top}(x_d-q)/\|x_d-q\|=1\right\}$, and $ds_i(q)\in\mathbb{S}^{n-1}$ is the inward-directed gradient of the surface of obstacle $\tilde{\mathcal{O}}_i$ at position $q\in\partial\tilde{\mathcal{O}}_i$.
\end{assumption}
Assumption \ref{as:4} somewhat limits the flatness of an obstacle at its farthest point from the target, as illustrated in Figure \ref{fig:conv_par}. \\
The shadow region for a convex obstacle is redefined as $\mathcal{D}(x_d,i):=\left\{q\in\mathcal{F}|\mathcal{L}(x_d,q)\cap\tilde{\mathcal{O}}_i\neq\varnothing\right\}$, where the center is replaced by the index of the obstacle as a parameter (see Fig. \ref{fig:conv_par}). The practical shadow region is then defined as $\tilde{\mathcal{D}}(x_d,i,R):=\left\{q\in\mathcal{D}^t(x_d,i)|d(q,\tilde{\mathcal{O}}_i)\leq R \right\}$, where $\mathcal{D}^t(x_d,i)$ is the truncated shadow region defined in \eqref{tr_sh}, substituting the center with the obstacle's index. The practical parameters and the control are the same as in \eqref{sensor_based_ctrl}. The next lemma provides the set of equilibria of the closed-loop system \eqref{12}-\eqref{sensor_based_ctrl} in the case of convex obstacles.
\begin{lemma}\label{lem9}
    All trajectories of the closed-loop system \eqref{12}-\eqref{sensor_based_ctrl} converge to the set $\tilde{\zeta}=\{x_d\}\cup\left(\cup_{i\in\mathbb{I}}\tilde{\mathcal{L}}_d(x_d,\mathrm{x}_i,R)\right)$, where $\tilde{\mathcal{L}}_d(x_d,\mathrm{x}_i,R):=\mathcal{L}_d(x_d,\mathrm{x}_i)\cap\tilde{\mathcal{D}}(x_d,i,R)$.
\end{lemma}
\begin{proof}
See Appendix \ref{appendix:Lemma 9}.
\end{proof}
In addition to the destination, Lemma \ref{lem9} shows that the equilibrium points are the positions aligned with their projection $\tilde c$ and the destination.
\begin{figure}[!h]
\centering
\includegraphics[scale=0.5]{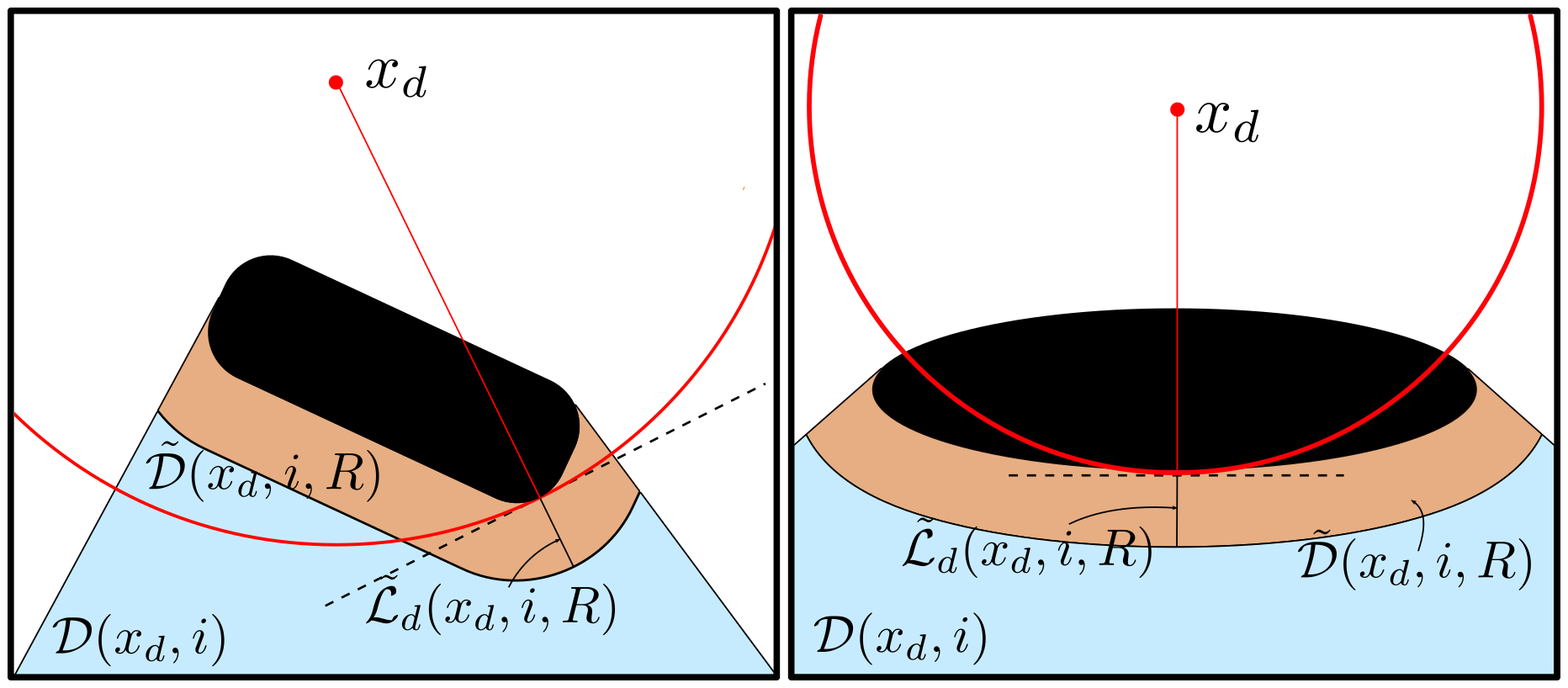}
\caption{Representation of the shadow region and the practical shadow region of a convex obstacle. In the figure on the left, the obstacle satisfies the curvature condition, while in the figure on the right, the obstacle does not satisfy this condition.}
\label{fig:conv_par}
\end{figure}
The same sensor-based procedure applied in sphere worlds is used, except that the elements of the list of arcs are not arcs but convex portions of the detected obstacles. In the case of convex obstacles with non-smooth boundaries, the procedure is modified where the endpoints of each detected portion are dilated with a ball of radius $r>0$, and the cone enclosing the segment crossed by $\mathcal{L}(x,x_d)$ is enlarged, as shown in Fig. \ref{fig:conv_buff}. The objective of dilating the endpoints is to smooth the corners of the obstacles. When an endpoint coincides with the vertex of an obstacle, the dilated endpoint will help to generate a smooth trajectory. If the robot rotates around an obstacle and applies the endpoint dilation on its boundary, a dilated version of this obstacle, given by  $\tilde{\mathcal{O}}_i^{r}=\tilde{\mathcal{O}}_i\oplus\mathcal{B}(0,r)$, will be generated. Therefore, the new free space will be $\mathcal{F}_{r}:=\mathcal{W}\setminus\bigcup\limits_{i=1}^{m} \mathring{\tilde{\mathcal{O}}}_i^{r}$ and the minimum separation distance will be greater than $2r$.    
\begin{figure}[!h]
\centering
\includegraphics[scale=0.35]{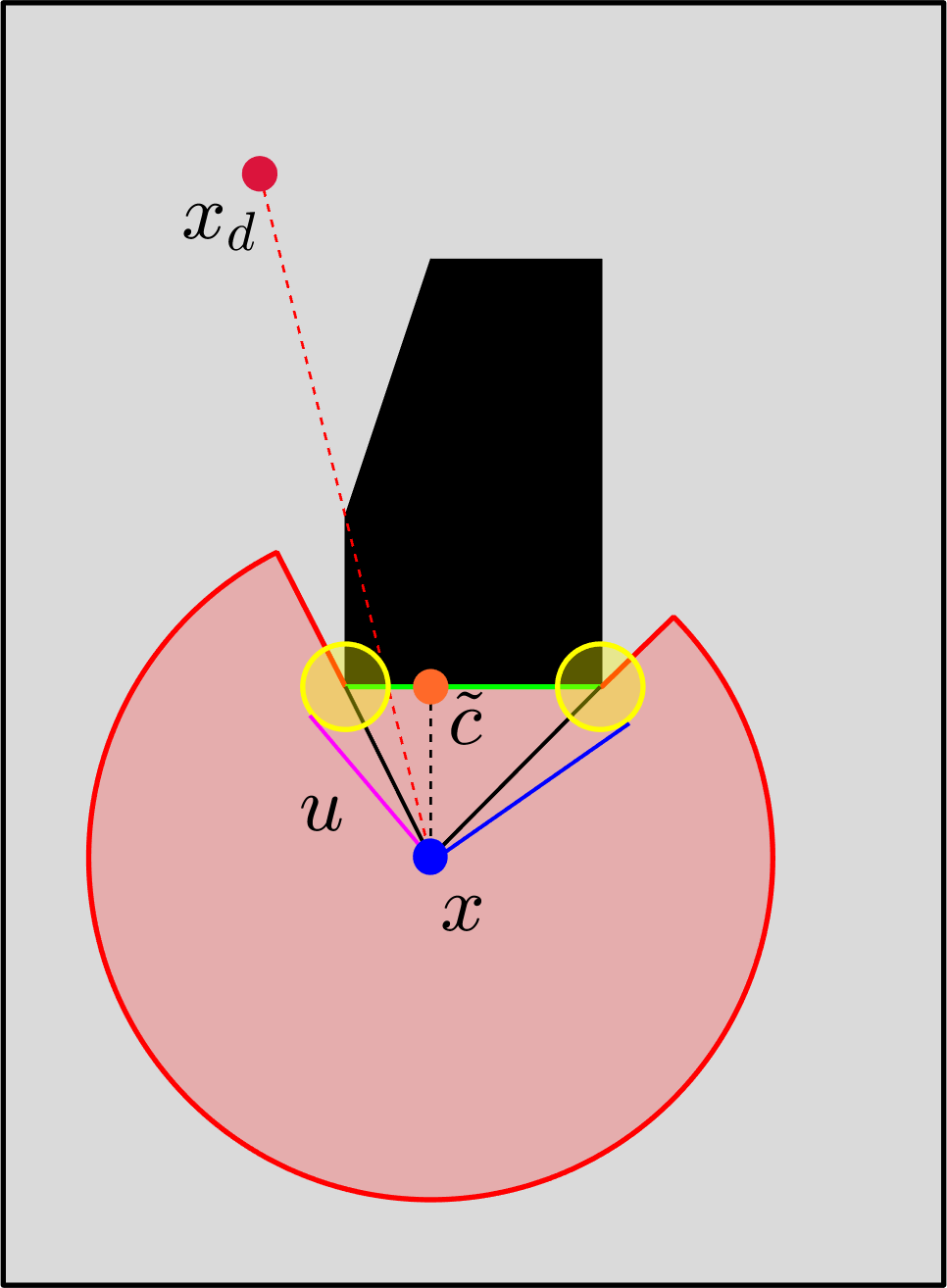}
\caption{Smoothing of polygonal obstacle corners.}
\label{fig:conv_buff}
\end{figure}
\begin{theorem}\label{the4}
    Consider the free space $\mathcal{F}\subset\mathbb{R}^n$ described in \eqref{8}, in the case of convex obstacles and $n=2$ , and the closed-loop system \eqref{12}-\eqref{sensor_based_ctrl}. Under Assumptions \ref{as:1}, \ref{as:2}, and \ref{as:4},  the following statements hold:
\begin{itemize}
\item [i)] The set $\mathcal{F}$ is forward invariant.
\item [ii)] All trajectories converge to the set $\tilde{\zeta}=\{x_d\}\cup\left(\cup_{i\in\mathbb{I}}\tilde{\mathcal{L}}_d(x_d,\mathrm{x}_i,R)\right)$.
\item [iii)] The set of undesired equilibria $\cup_{i\in\mathbb{I}}\tilde{\mathcal{L}}_d(x_d,\tilde{x}_i,R)$ is unstable.
\item [iv)] The equilibrium point $x_d$ is almost globally asymptotically stable on $\mathcal{F}$.
\end{itemize}
\end{theorem}
\begin{proof}
See Appendix \ref{appendix:the4}.
\end{proof}
Theorem \ref{the4} shows that the sensor-based strategy designed for sphere worlds extends to convex worlds with obstacles satisfying the curvature condition of Assumption \ref{as:4}, and the results are preserved. For convex obstacles with non-smooth boundaries, we consider their dilated version $\tilde{\mathcal{O}}_i^{r}$ and the free space $\mathcal{F}_r$ which amounts to the case of obstacles with smooth boundaries.
\section{Numerical simulation}
 To explore the extent of what our {\it quasi-optimal} trajectories can offer in terms of the shortest path in the multiple obstacle case, we compare the trajectories generated by our approach with the shortest paths obtained with Dijkstra's algorithm (DA) on a tangent visibility graph (TVG). We used 10 different and highly congested two-dimensional environments and 100 randomly selected initial positions in each environment. The percentage of perfect matches of the paths is reported in Table \ref{table_1}, which shows a high rate of success. Fig. \ref{10spaces} shows a sample of 10 trajectories generated from 10 randomly selected initial positions in two of the ten environments used in our simulations. A simulation video can be found at \url{https://youtu.be/SE8w8vabxVE}. 
 The effect of successive projections on the optimality of the path generated by our approach is illustrated in Fig. \ref{Long_Short}, where one can see that the path generated by our approach coincides with the shortest path in a single-obstacle workspace, while it does not in a two-obstacle workspace.
\begin{figure}[h!]
\renewcommand*\thesubfigure{\arabic{subfigure}}
     \centering
     \subfloat[]{\includegraphics[width=0.26\linewidth,height=0.26\linewidth,keepaspectratio]{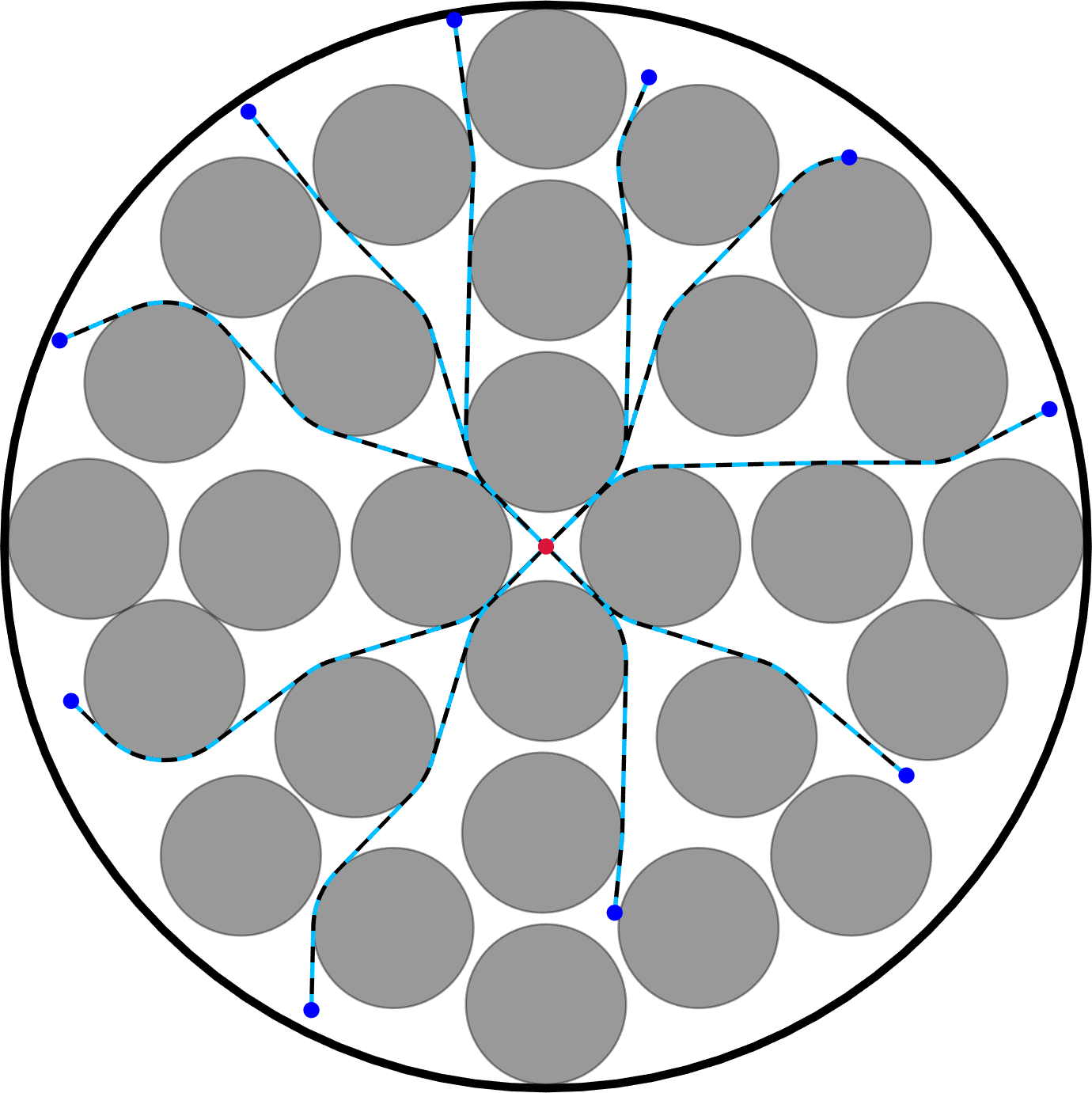}\label{s1}}
     \subfloat[]{\includegraphics[width=0.26\linewidth,height=0.26\linewidth,keepaspectratio]{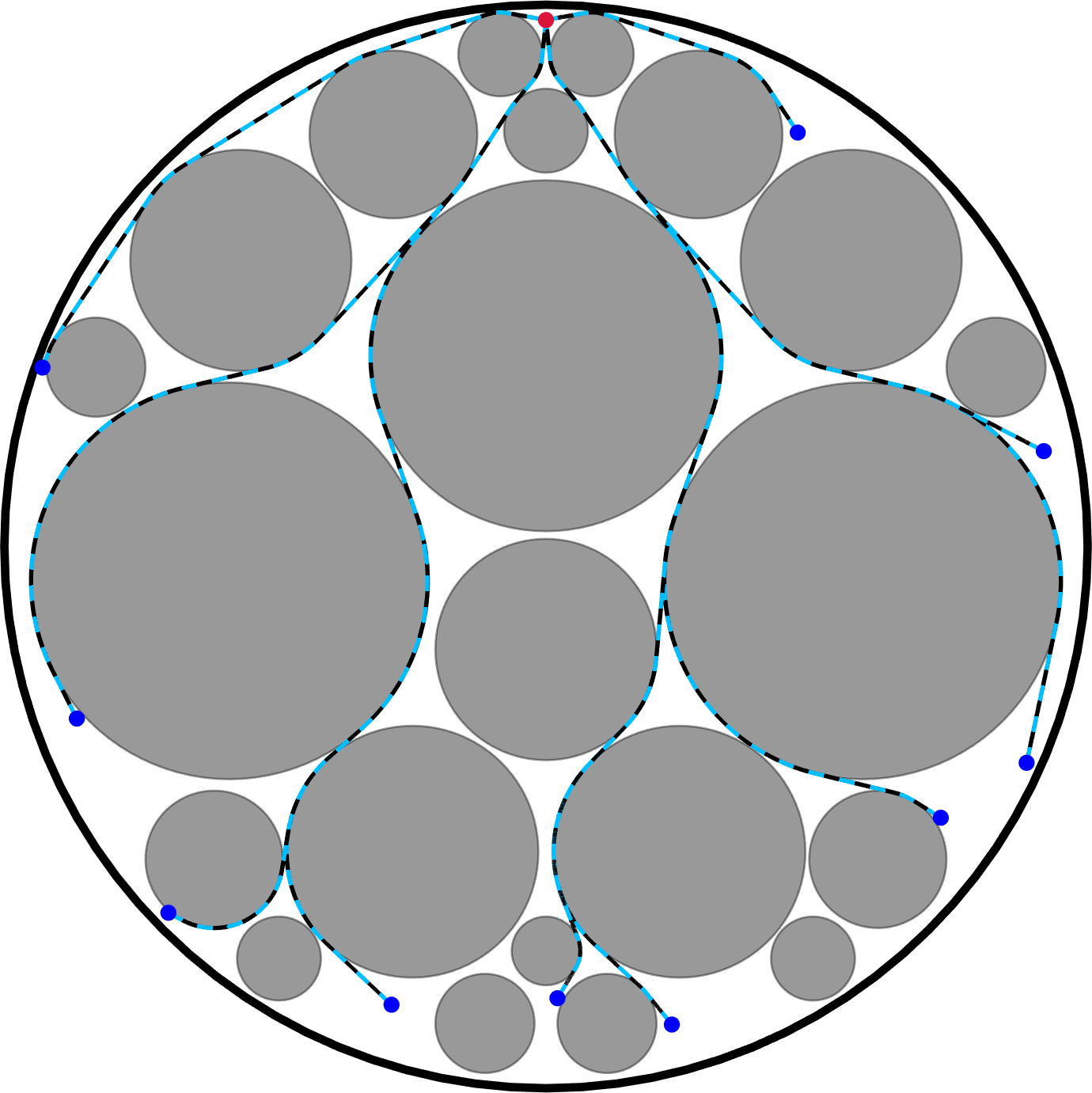}\label{s3}}
     \subfloat[]{\includegraphics[width=0.26\linewidth,height=0.26\linewidth,keepaspectratio]{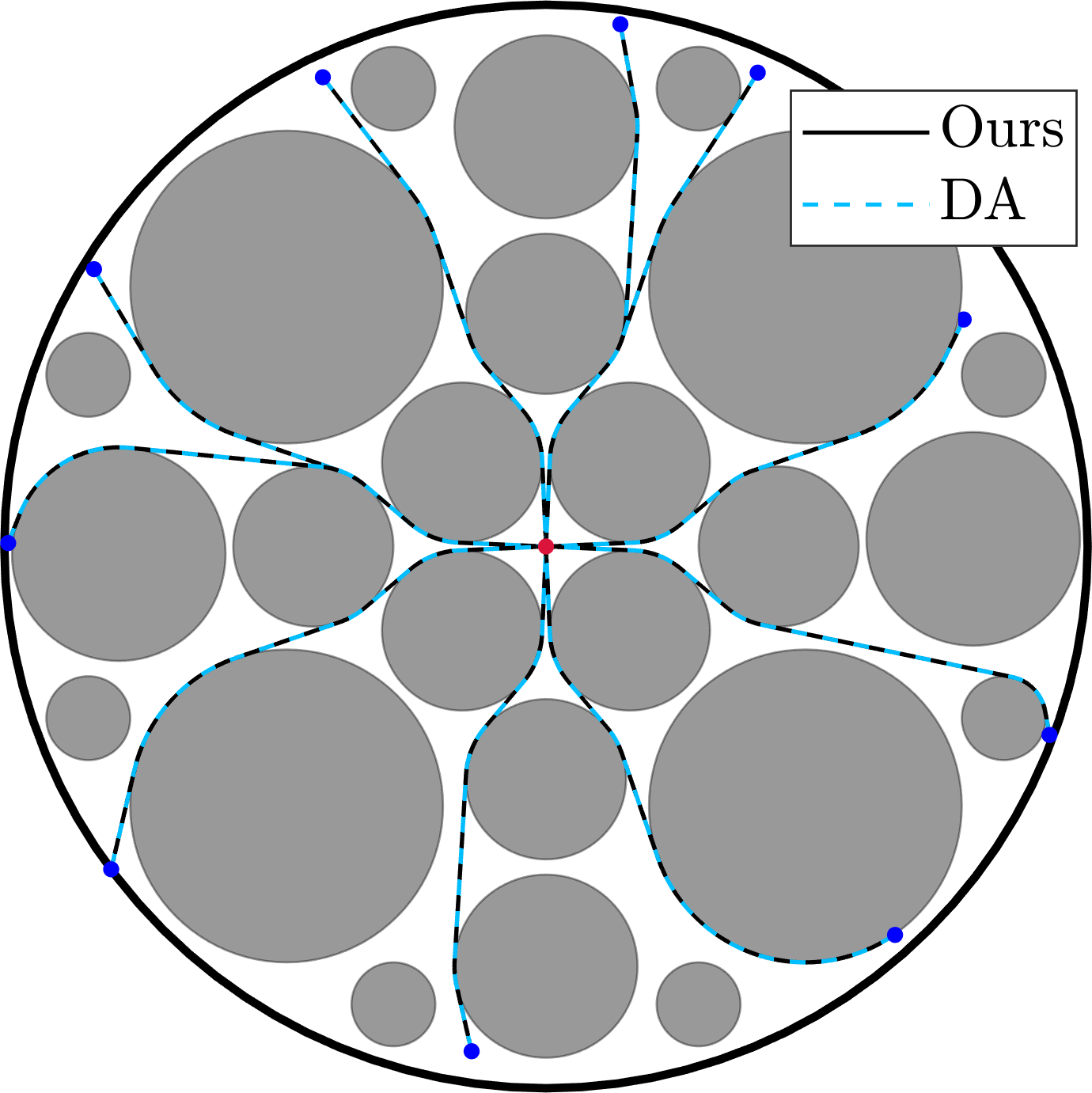}\label{s2}}\\
     \subfloat[]{\includegraphics[width=0.26\linewidth,height=0.26\linewidth,keepaspectratio]{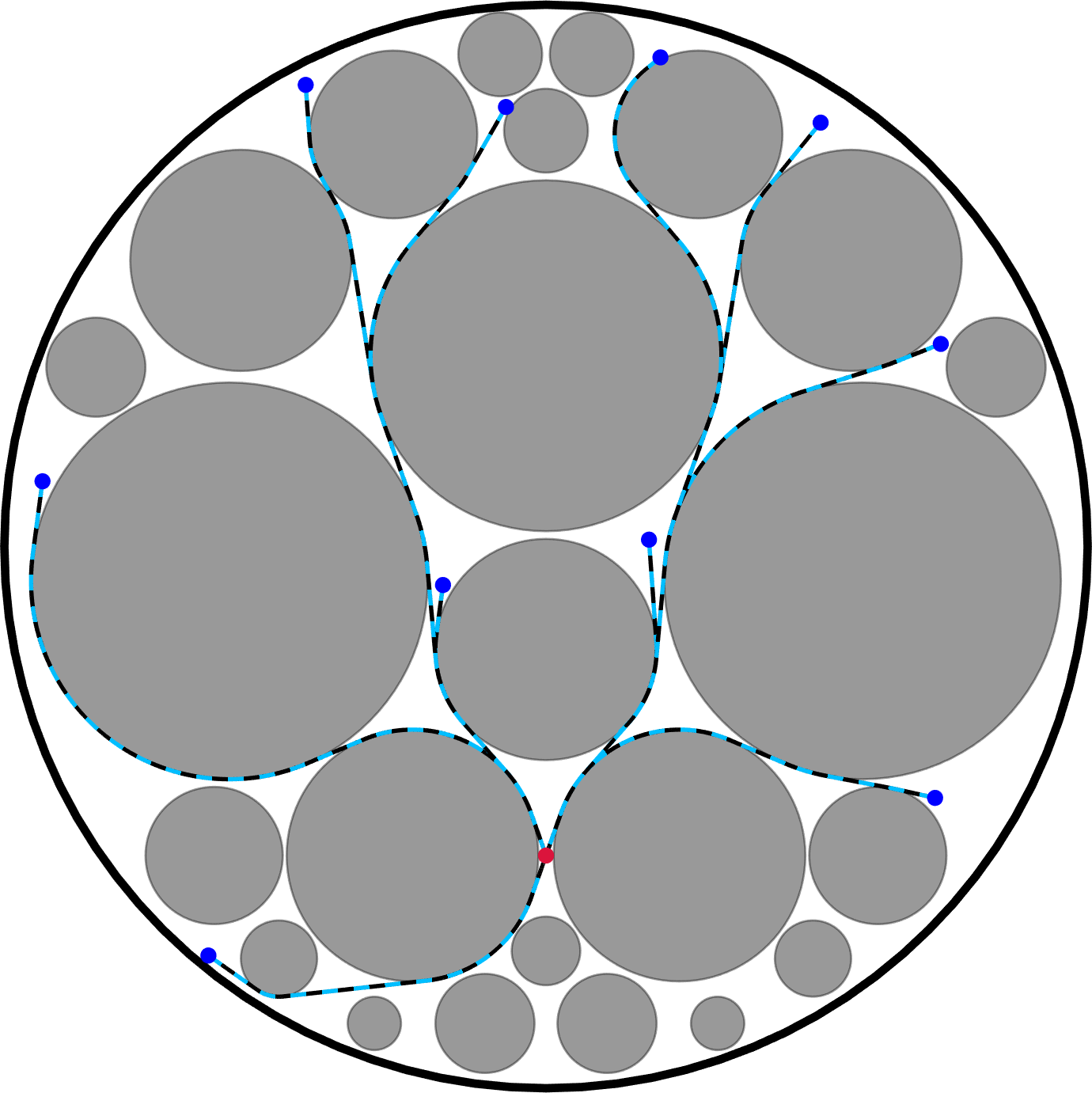}\label{s4}}
     \subfloat[]{\includegraphics[width=0.26\linewidth,height=0.26\linewidth,keepaspectratio]{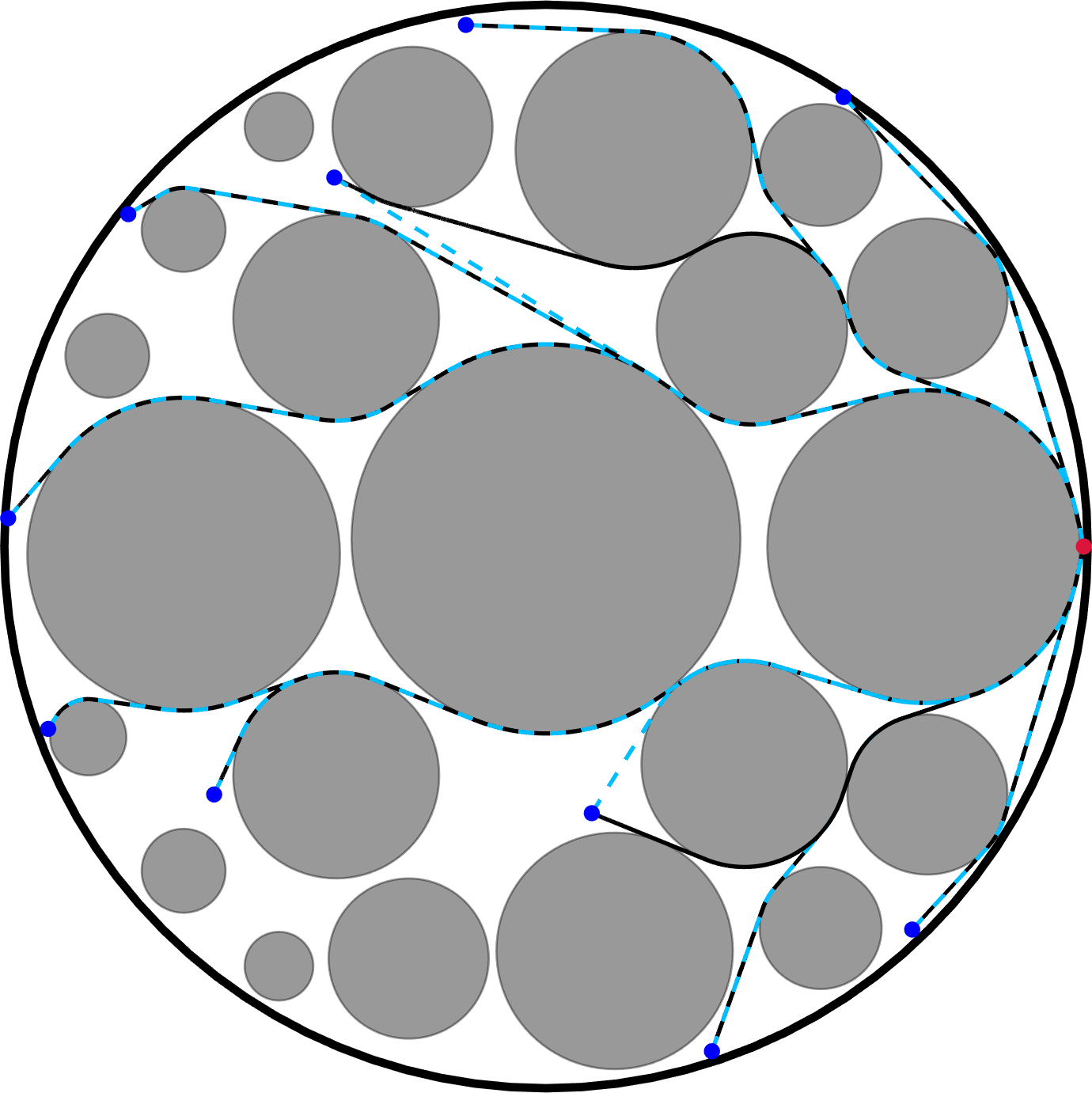}\label{s5}}
     \subfloat[]{\includegraphics[width=0.26\linewidth,height=0.26\linewidth,keepaspectratio]{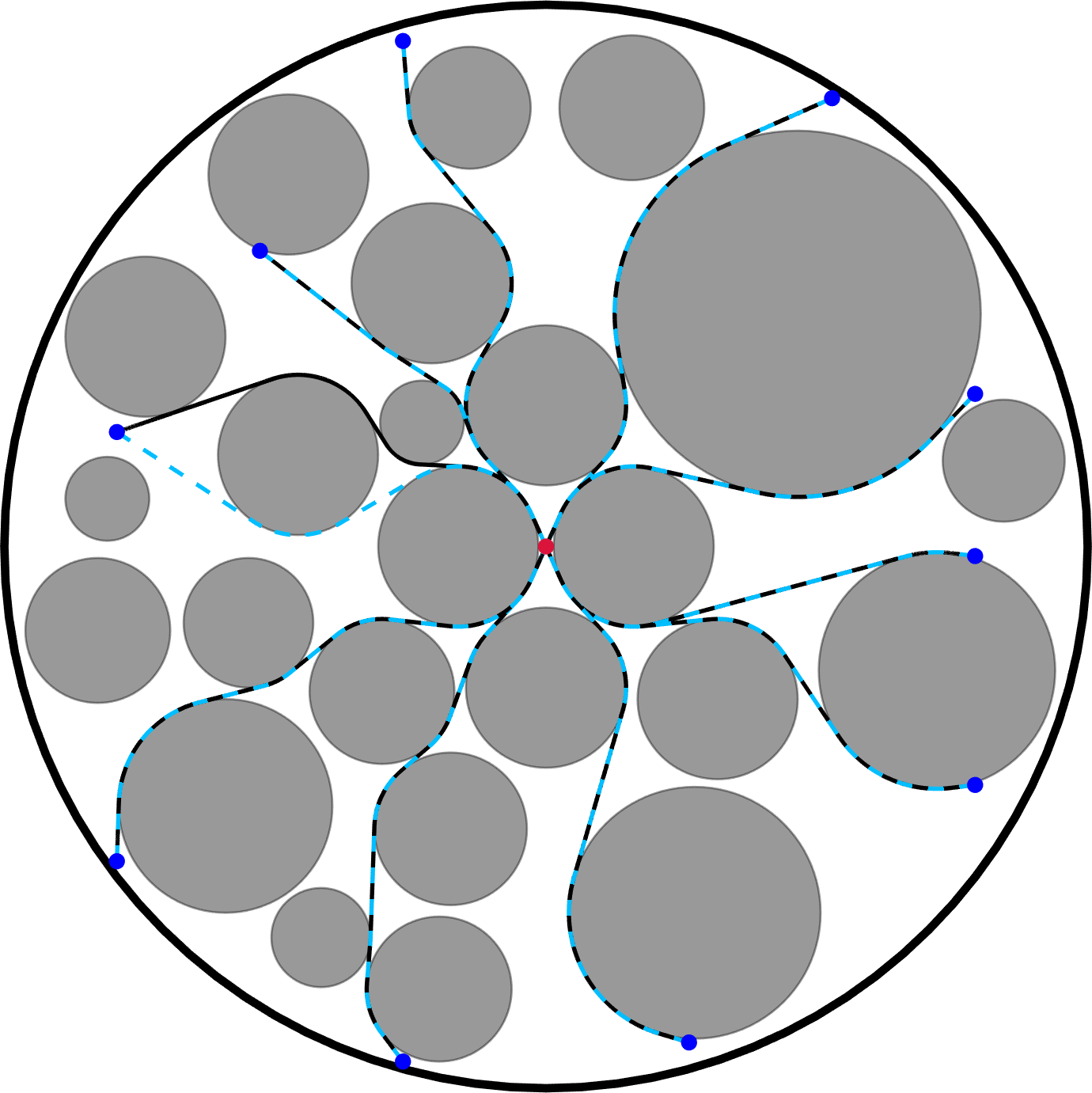}\label{s6}}\\
     \subfloat[]{\includegraphics[width=0.26\linewidth,height=0.26\linewidth,keepaspectratio]{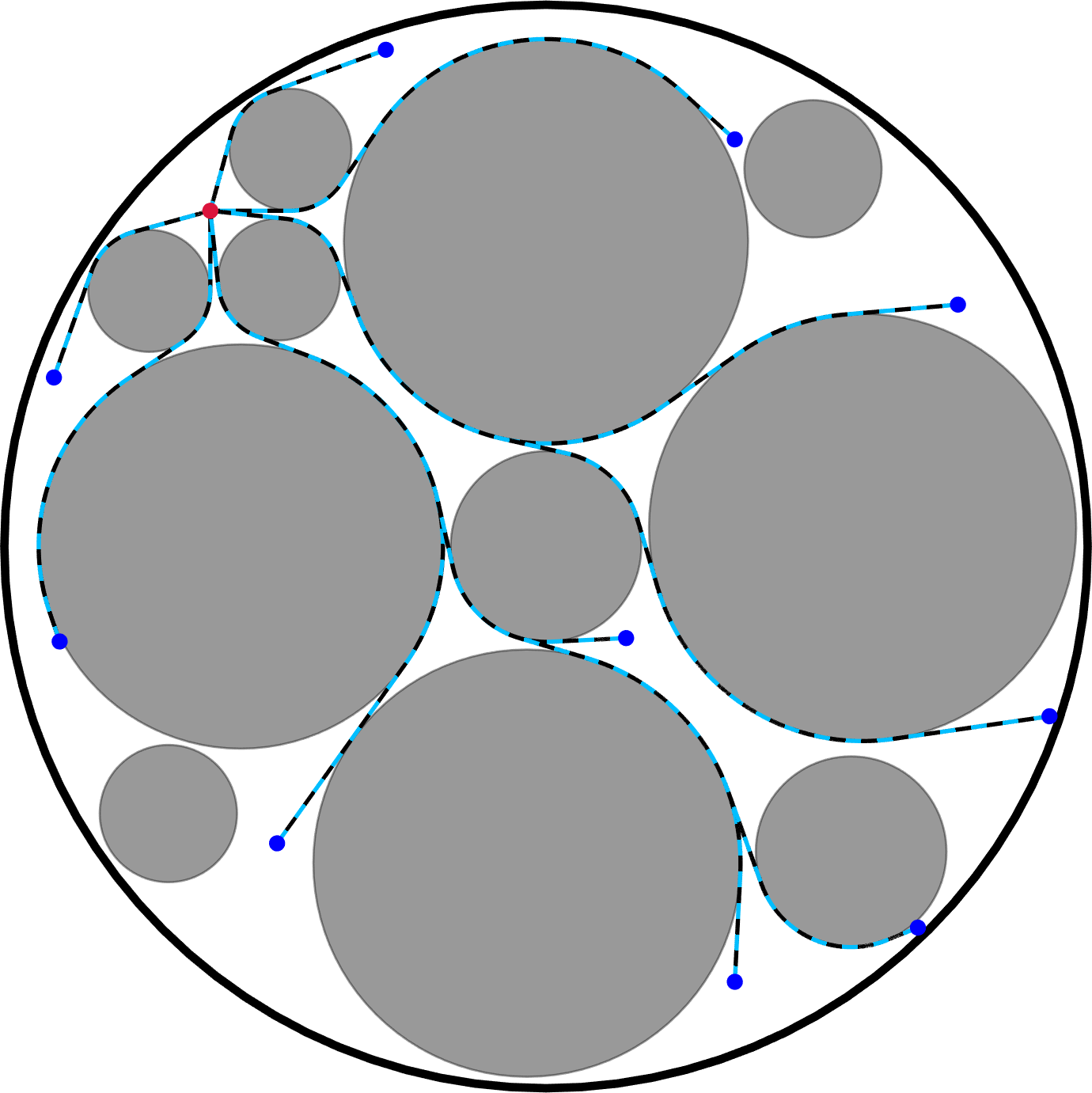}\label{s7}}
     \subfloat[]{\includegraphics[width=0.26\linewidth,height=0.26\linewidth,keepaspectratio]{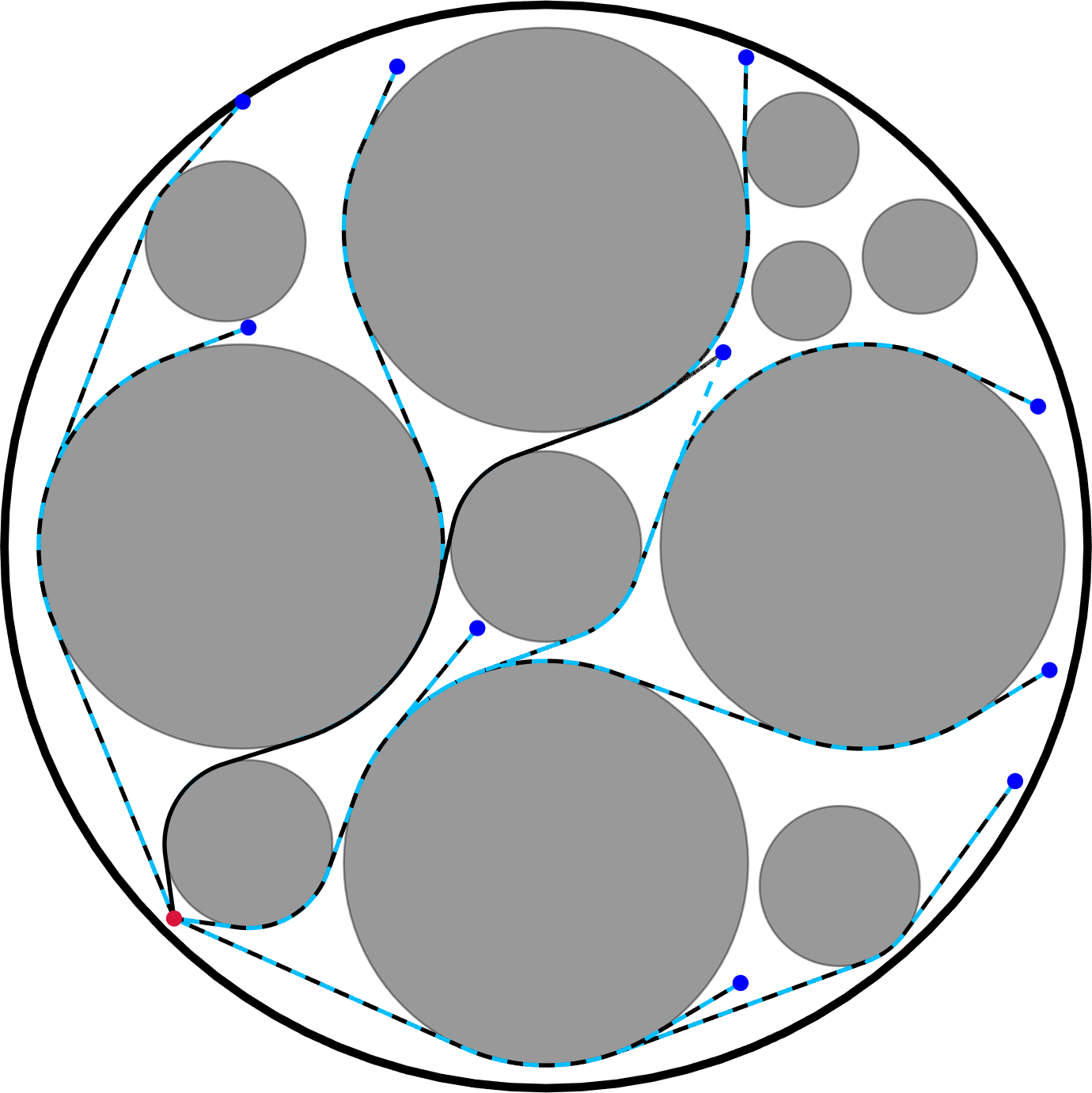}\label{s8}}
     \subfloat[]{\includegraphics[width=0.26\linewidth,height=0.26\linewidth,keepaspectratio]{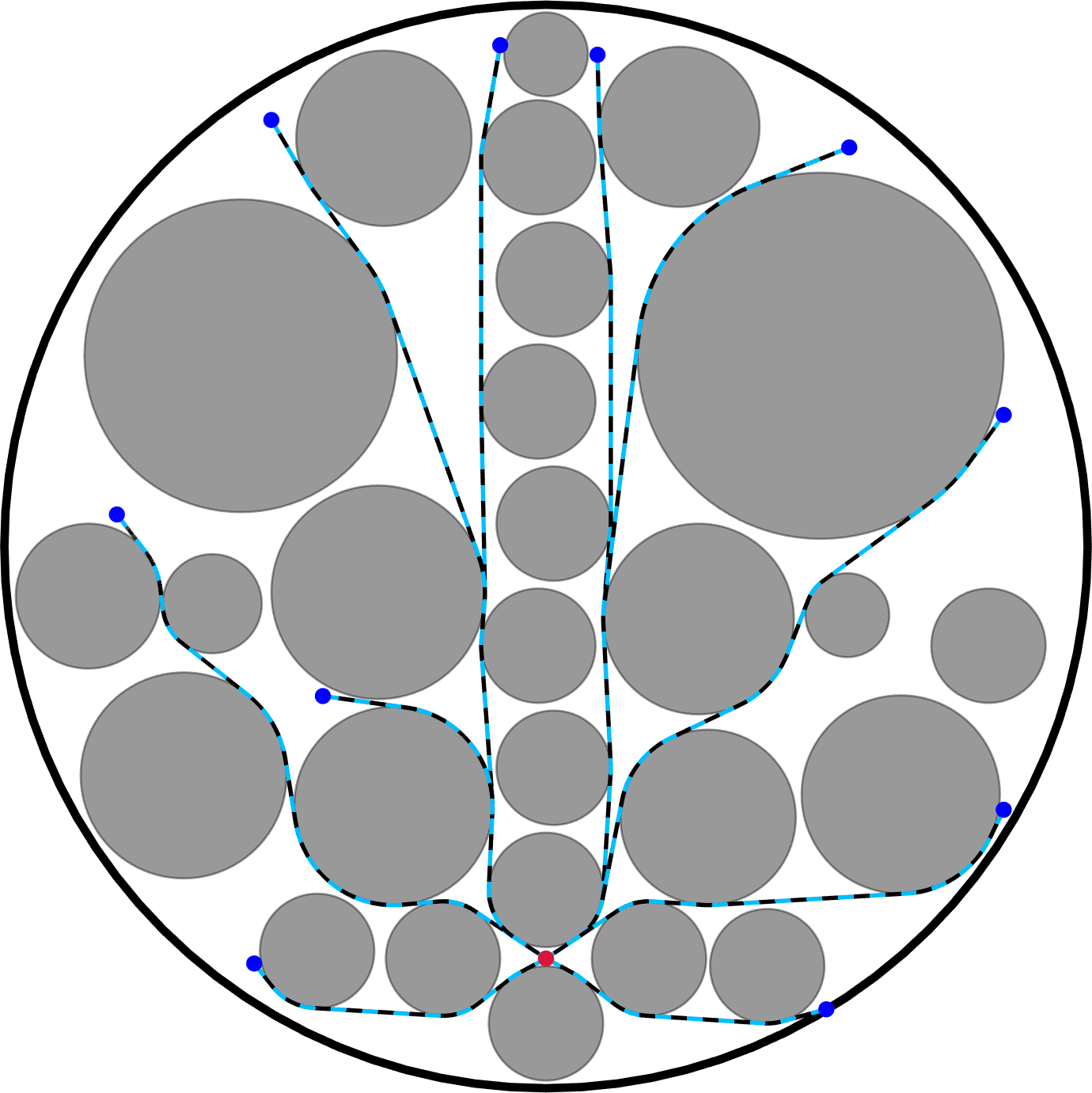}\label{s9}}\\
     \subfloat[]{\includegraphics[width=0.26\linewidth,height=0.26\linewidth,keepaspectratio]{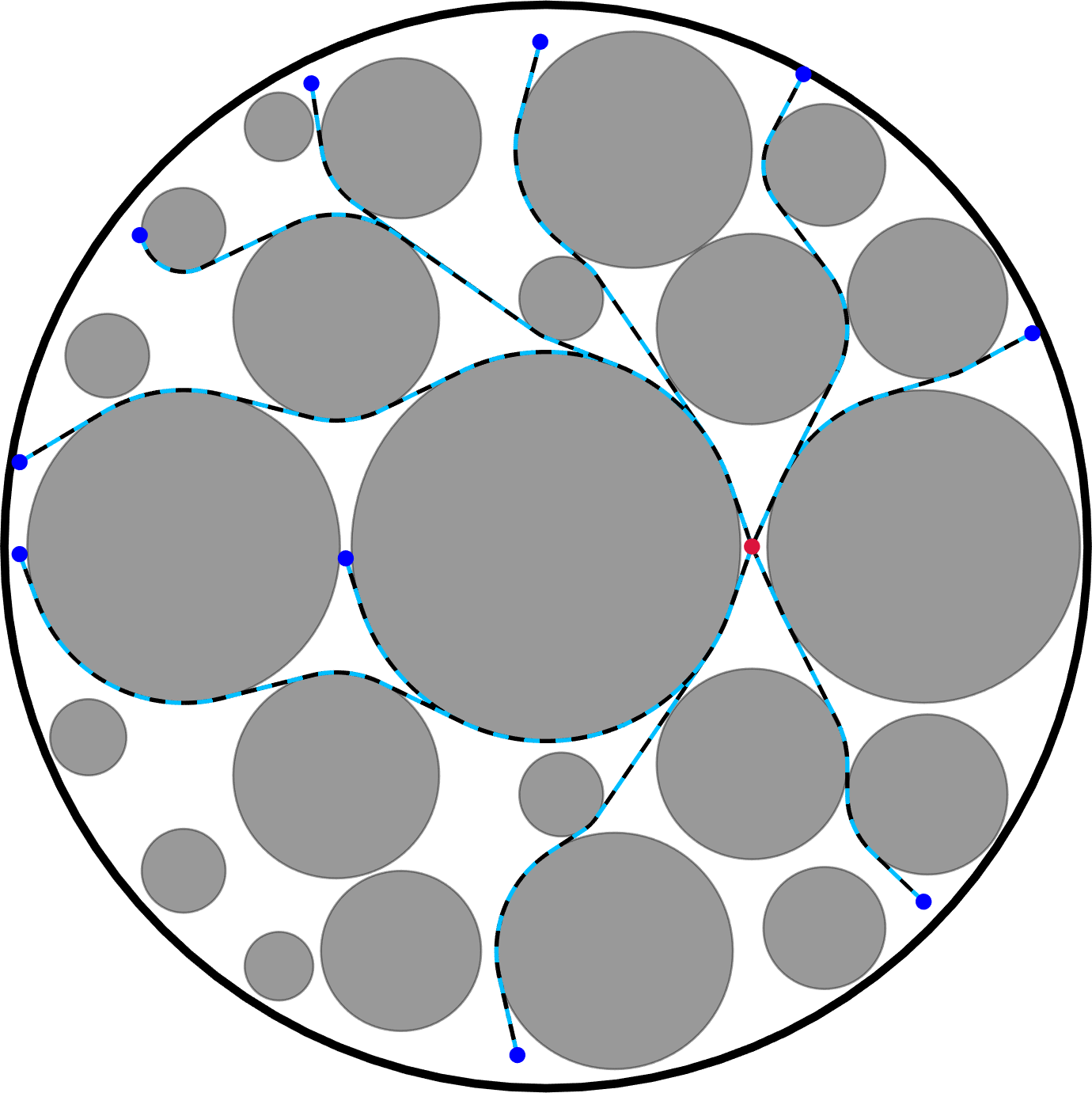}\label{s10}}
     \caption{Trajectories generated by our approach (Algorithm \ref{alg1}) and the optimal path found by DA in tangent visibility graphs, for 10 different initial positions in 10 different environments. The target location is indicated with a red dot.}
     \label{10spaces}
\end{figure}
\begin{table}[h!]
\caption{Number of perfect matches between the paths generated by our approach (Algorithm \ref{alg1}) and those found by DA in tangent visibility graphs, for 100 runs with 100 randomly selected initial positions. }
\label{table_1}
\begin{center}
\begin{tabular}{|c||c||c||c||c|}
\hline
\textbf{Space 1} & \textbf{Space 2} & \textbf{Space 3} & \textbf{Space 4} & \textbf{Space 5}\\
\hline
100$\%$ & 98$\%$ &100$\%$ &100$\%$ &81$\%$\\
\hline
\textbf{Space 6} & \textbf{Space 7} & \textbf{Space 8} & \textbf{Space 9} & \textbf{Space 10}\\
\hline
96$\%$ & 99$\%$ &94$\%$ &94$\%$ &99$\%$\\
\hline
\end{tabular}
\end{center}
\end{table}
\begin{figure}[h!]
     \centering
     \subfloat[]{\includegraphics[width=0.45\linewidth,height=0.45\linewidth,keepaspectratio]{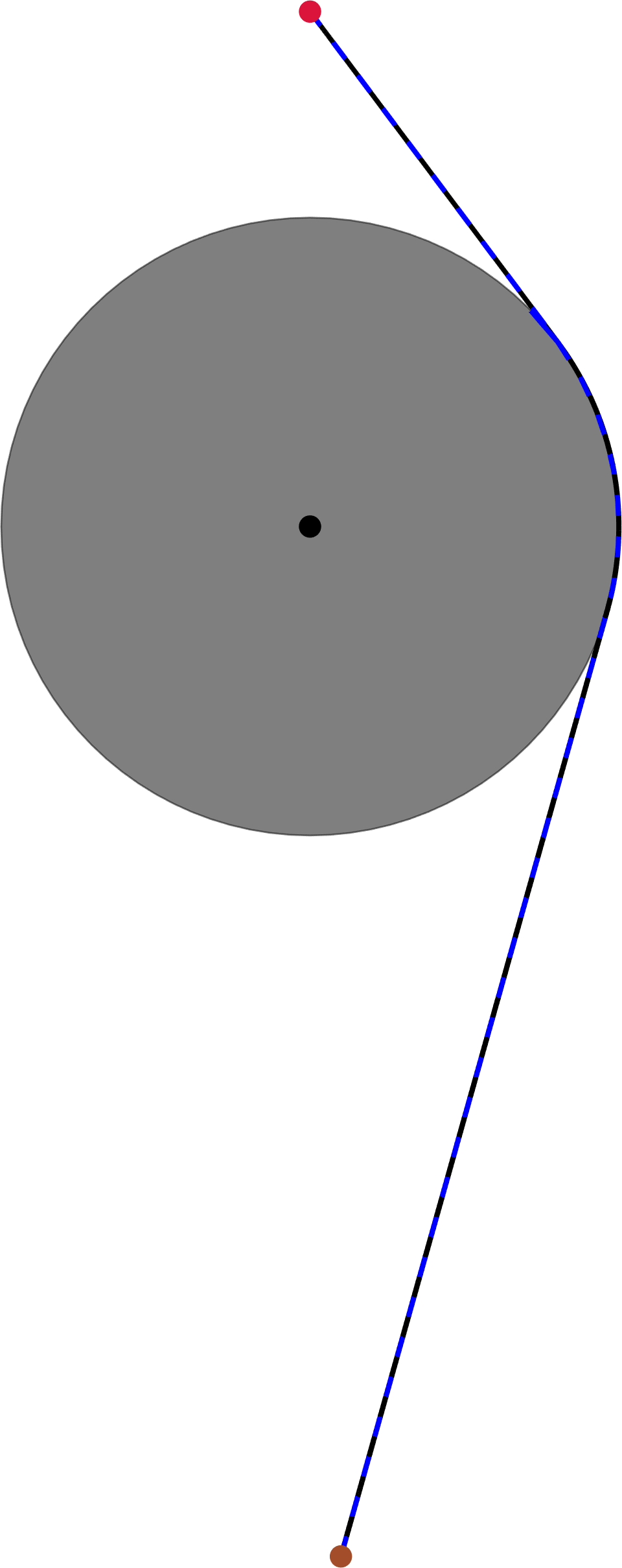}\label{L1}}
    \hspace{0.01cm}
     \subfloat[]{\includegraphics[width=0.45\linewidth,height=0.45\linewidth,keepaspectratio]{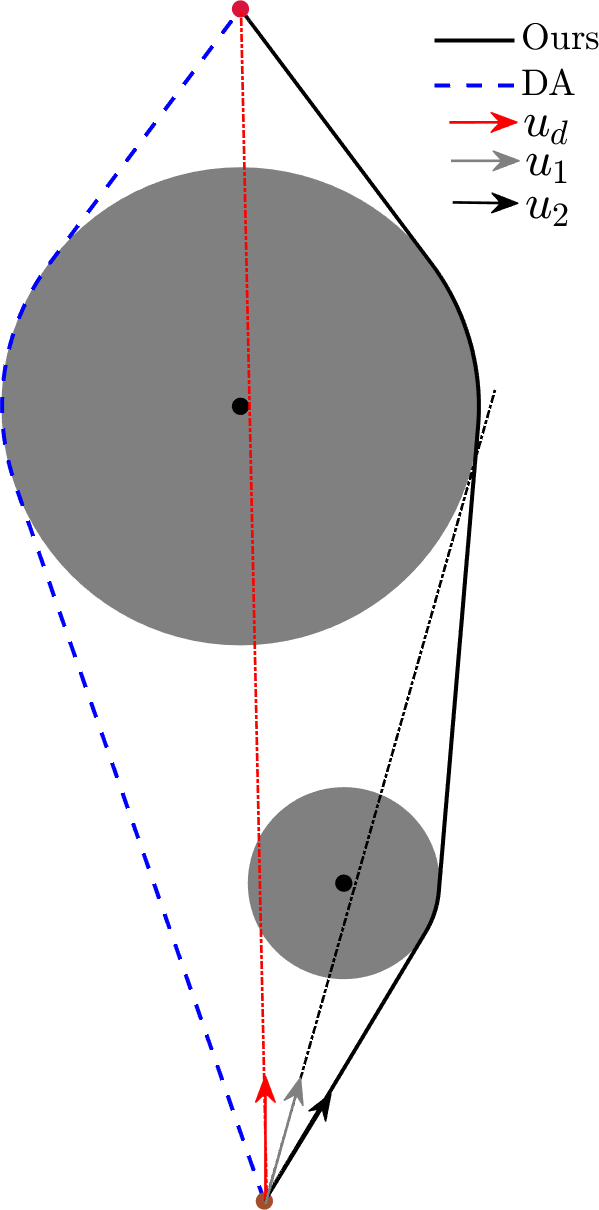}\label{L2}}\\
     \caption{Effect of the nested projections on the optimality of the generated trajectory. In Fig. (a), a single obstacle is considered, and the trajectory generated by our approach corresponds to the DA trajectory. In Fig. (b), a second obstacle is considered. The trajectory generated by our approach differs from the DA trajectory.}
     \label{Long_Short}
\end{figure}
\begin{rem}
    The combination TVG-DA has been used as a benchmark to test the optimality of the paths generated by our approach. The advantages of our approach w.r.t. the TVG-DA are as follows:
    \begin{itemize}
        \item We solve the problem from a control perspective, as our solution is feedback-based, allowing us to solve the navigation problem in one go, whereas TVG-DA only gives the shortest path to be tracked by another feedback controller.
        \item We propose a closed-form solution which is more suitable for real-time implementations (computationally efficient) than searching tangent visibility graphs.
        \item We solve the navigation problem in $n$-dimensional sphere worlds while the TVG-DA is limited to paths in two-dimensional sphere worlds as the TVG is infinite in higher dimensions. 
    \end{itemize}
\end{rem}
To visualize the properties of our approach, we consider two different scenarios. In the first scenario, we assume that the robot evolves in $\mathbb{R}^2$ where the workspace contains twenty-six obstacles, and the destination is $x_d=[0\;0]^\top$. We run the simulation from fifteen different initial positions. In the second scenario, the considered space is $\mathbb{R}^3$, where the workspace contains eighteen obstacles, and the goal is $x_d=[0\;0\;0]^\top$. We run the simulation from eighteen different initial positions. A comparison of our approach with the navigation function approach (NF) \citep{k_R_90} and the separating hyperplane approach with the Voronoi-adjacent obstacle sensing model (SH) \citep{Arslan2019} is established in the two-dimensional space. The simulation results in Fig. \ref{fig:2dcomp} and \ref{fig:3dsim}  show that all the trajectories generated by our control strategy are safe and converge to the red target. In addition, Fig. \ref{fig:2dcomp} shows the superiority of our approach over the two other methods in terms of the length of the generated collision-free paths where it generates the same paths as DA. Moreover, Table \ref{table_2} reports the relative length difference of the paths generated by the NF and SH approaches with respect to our approach. For each initial position $p_i$, $i\in\{1,\dots,15\}$, in Fig. \ref{fig:2dcomp}, we computed the relative length difference $RLD_i^a=100(l_i^a-l_i^{0})/l_i^{0}$, $a\in\{NF,SH\}$, where $l_i^{NF}$ (resp. $l_i^{SH}$) is the length of the $i$th path generated by the NF approach (resp. SH approach), and $l_i^{0}$ is the length of the path generated by our approach. The positive numbers in Table \ref{table_2} indicate that, for all 15 initial conditions, our approach generated shorter paths than the NF and SH approaches. This superiority is mainly due to the uncontrolled repulsion exerted by the obstacles on the robot in the NF and SH approaches. It becomes clear in the single obstacle case where the robot is repelled even if it has a clear line-of-sight to the destination, which is shown in the simulation result in Fig. \ref{fig:1obs_2d_comp}, where the pink initial positions are in the visible set while the green initial positions are in the shadow region. The trajectories generated by our approach are the shortest in terms of distance, as shown in Lemma \ref{lem1}. The simulation video of Fig. \ref{fig:3dsim} can be found at \url{https://youtube.com/shorts/yJCdRLdQHnc}.\\
\begin{figure}[h!]
\centering
\includegraphics[scale=0.4]{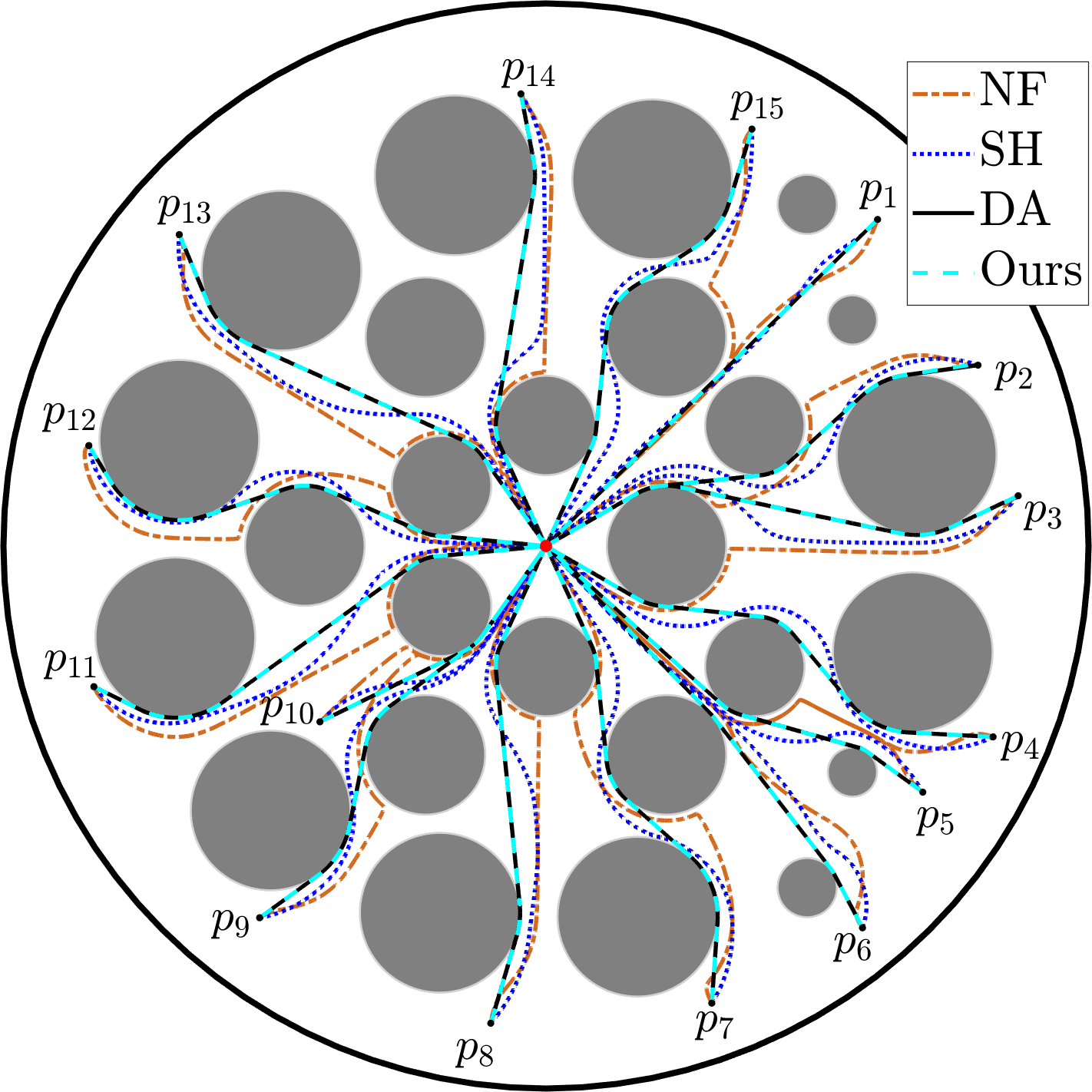}
\caption{Trajectories generated by our approach, SH, NF and DA in a two-dimensional sphere world.}
\label{fig:2dcomp}
\end{figure}
\begin{table}[h!]
\caption{The relative length difference of the paths, shown in Fig. \ref{fig:2dcomp}, generated by the NF and SH approaches with respect to our approach. 
}
\label{table_2}
\begin{center}
\begin{tabular}{|c||c||c|}
\hline
Paths & $RLD^{NF}$ (\%)& $RLD^{SH}$ (\%)\\
\hline
$p_1$ & 1.18 &0.27\\
\hline
$p_2$ & 11.36 &7.59 \\
\hline
$p_3$ & 8.6 &5.2\\
\hline
$p_4$ & 5.93 &7.23 \\
\hline
$p_5$ & 6.57 &3.64\\
\hline
$p_6$ & 4.26 &2.43 \\
\hline
$p_7$ & 13.35 &7.15\\
\hline
$p_8$ & 6.6 &3.79 \\
\hline
$p_9$ & 11.34 &5.47\\
\hline
$p_{10}$ & 6.63 &2.98 \\
\hline
$p_{11}$ & 9.79 &3.91\\
\hline
$p_{12}$ & 14.08 &5.05 \\
\hline
$p_{13}$ & 9.24 &4.24 \\
\hline
$p_{14}$ & 7.23 &4.31\\
\hline
$p_{15}$ & 4.65 &6.96 \\
\hline
\end{tabular}
\end{center}
\end{table}
\begin{figure}[h!]
\centering
\begin{minipage}{.4\textwidth}
  \centering
  \includegraphics[width=1\linewidth]{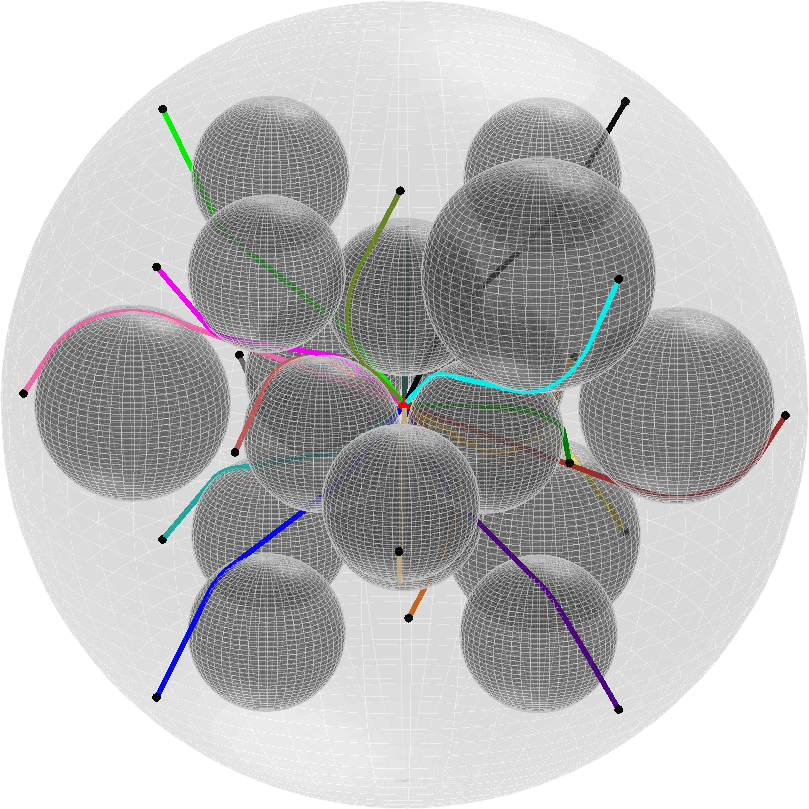}
  \captionof{figure}{Robot safe navigation from eighteen different initial positions in a three-dimensional sphere world.}
  \label{fig:3dsim}
\end{minipage}
\begin{minipage}{.4\textwidth}
  \centering
  \includegraphics[width=1\linewidth]{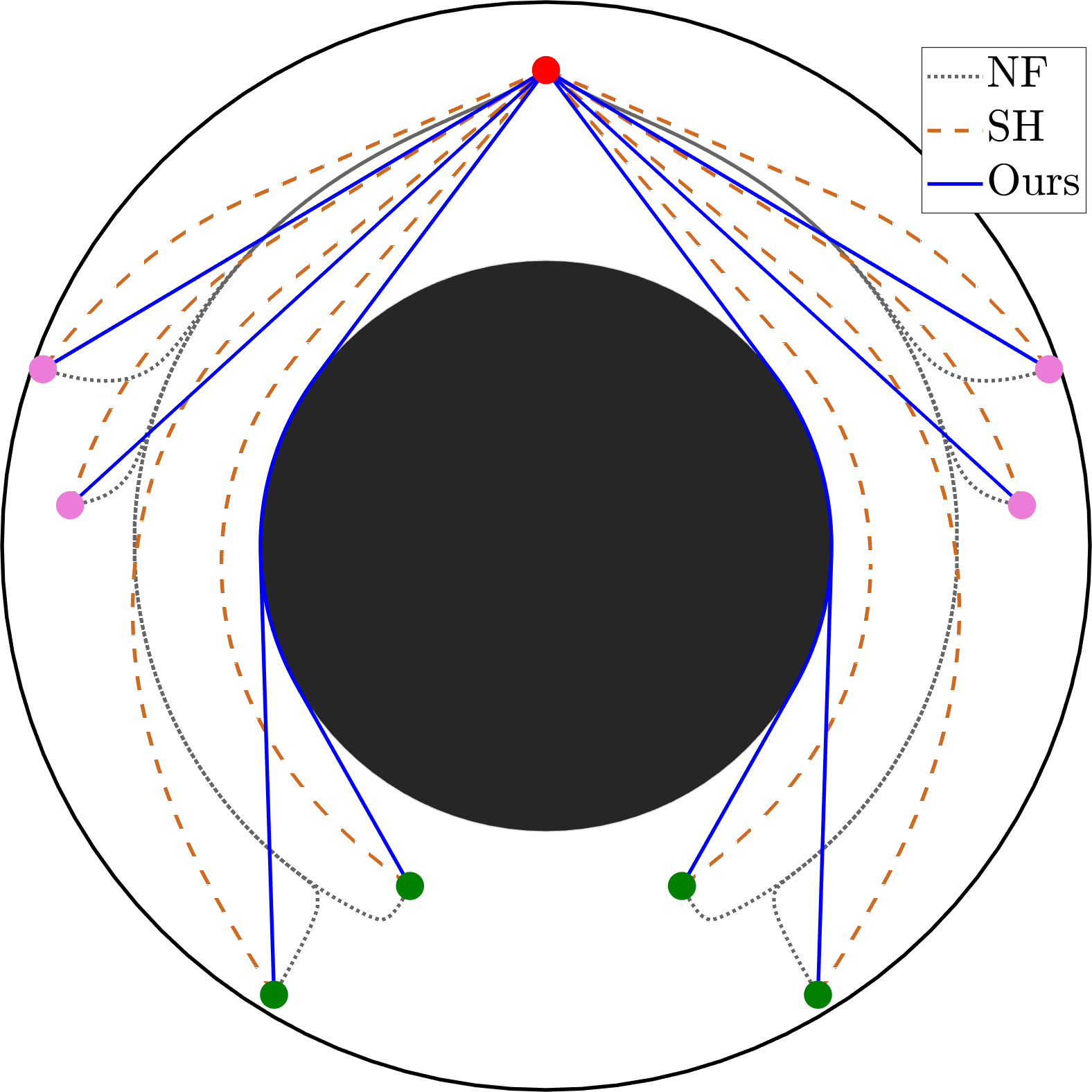}
  \captionof{figure}{Comparison of paths generated by our approach, SH, and NF in a single two-dimensional sphere world.}
  \label{fig:1obs_2d_comp}
\end{minipage}
\end{figure}
Let us test our control in a two-dimensional space that does not satisfy Assumption \ref{as:3}. We consider six different initial positions. Three are inside the nests, and the remaining three are outside but in the vicinity of their boundaries (undesired equilibria). The results of the simulation are shown in Fig. \ref{fig:2d_3obs_sim}. The trajectories starting from the nests stay inside, while the three remaining trajectories reach their destination safely. We can see that nests are indeed the attraction region of the undesired equilibria.    
\begin{figure}[h!]
\centering
\includegraphics[scale=0.4]{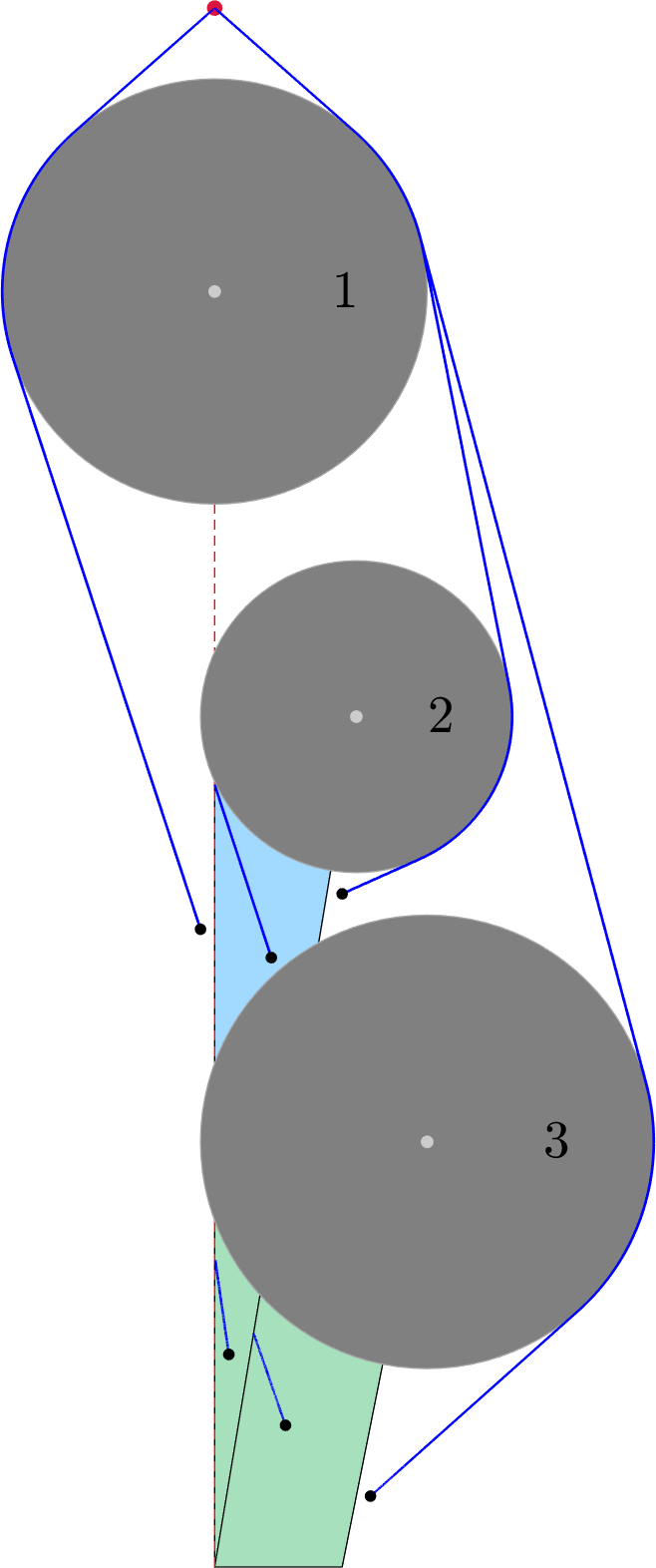}
\caption{Invariance of the nests.}
\label{fig:2d_3obs_sim}
\end{figure}
\subsection{Sensor-based implementation in \textit{a priori} unknown environments}
\subsubsection{MATLAB Simulation}
We tested our sensor-based strategy \eqref{sensor_based_ctrl} in a sphere world, as shown in Fig. \ref{fig:Sensor_sim}, where we used a $360^\circ$ LiDAR model with $1^\circ$ resolution and two radial ranges $R=2m$ and $R=4m$. We plotted the trajectories generated by Algorithm 2 and Algorithm 3. The results in Fig. \ref{fig:Sensor_sim}, clearly show a decrease in performance, in terms of path length, when navigating in an {\it a priori} unknown environment relying on a sensor. This is expected since, in the sensor-based approach, the information available about the workspace is limited to the sensor's detection zone. It is, therefore, impossible to predict \textit{a priori} the obstacles to be avoided before the sensor detects them. In fact, under control \eqref{36}, where global information on the environment is available, after avoiding an obstacle, the robot already knows the next obstacle to avoid, resulting in a {\it quasi-optimal} trajectory. Overall, the sensor-based approach provides a \textit{local} optimal solution in the sense that each local avoidance maneuver is optimal when the robot is close enough to the detected obstacle. The simulation video can be found at \url{https://youtu.be/cnWoxi-lGvw}.\\
\begin{figure}[h!]
\centering
\includegraphics[scale=0.38]{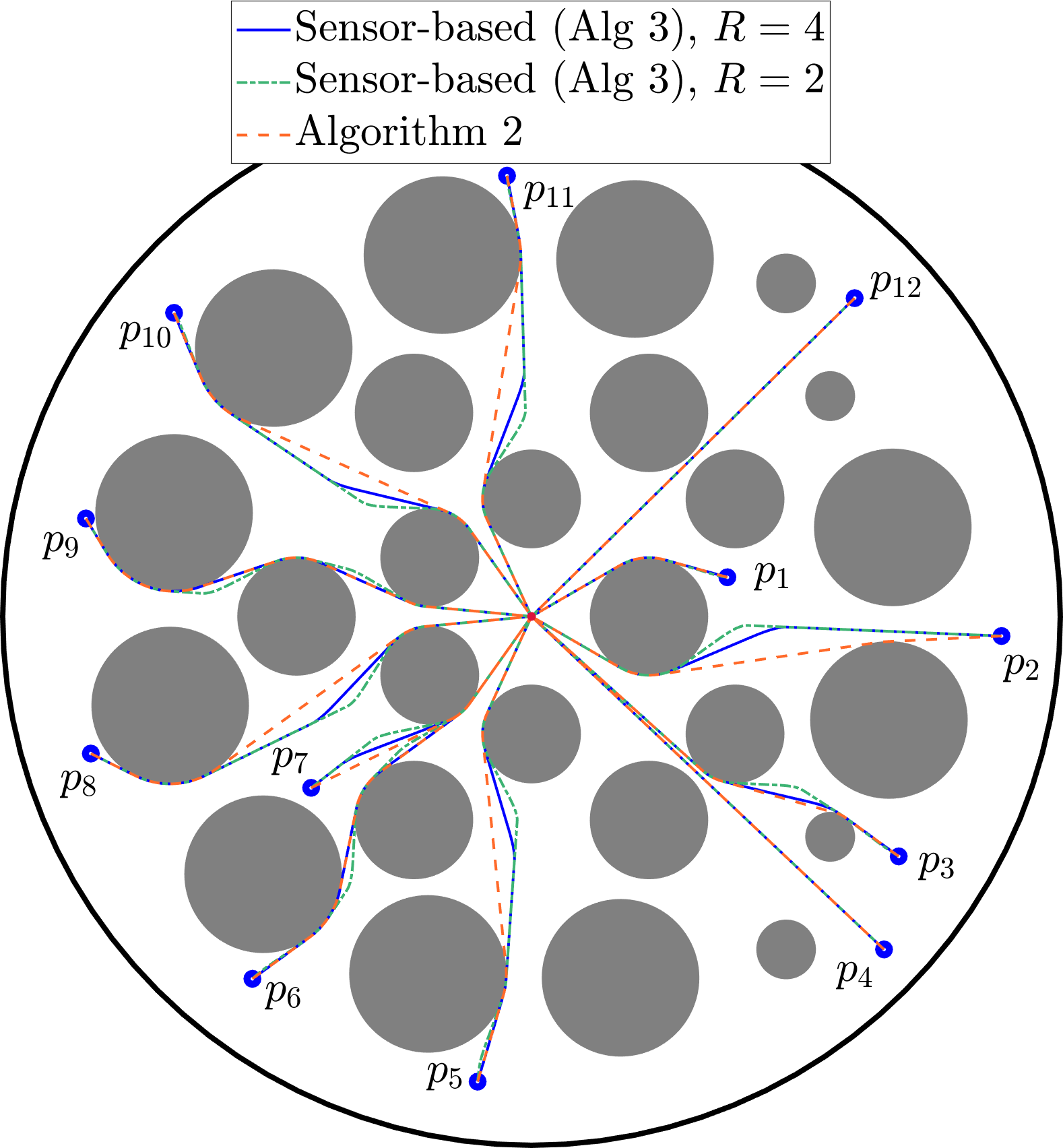}
\caption{Navigation in 2D sphere world.}
\label{fig:Sensor_sim}
\end{figure}
\begin{table}[h!]
\caption{The relative length difference of the paths, shown in Fig. \ref{fig:Sensor_sim}, generated by Algorithm \ref{alg_sensor}, with two sensor radial ranges ($R=2m$ and $R=4m$), with respect to Algorithm \ref{alg1}.}
\label{table_3}
\begin{center}
\begin{tabular}{|c||c||c|}
\hline
Paths & RLD$^1$ (\%)& RLD$^2$ (\%)\\
\hline
$p_1$ & 1.37 &2.37\\
\hline
$p_2$ & 0.12 &0.8 \\
\hline
$p_3$ & 0 &0\\
\hline
$p_4$ & 0 &0.02 \\
\hline
$p_5$ & 0.72 &1.5\\
\hline
$p_6$ & 0 &0.72 \\
\hline
$p_7$ & 0.43 &1.24\\
\hline
$p_8$ & 0.69 &1.37 \\
\hline
$p_9$ & 0 &0.68\\
\hline
$p_{10}$ & 0.7 &1.4 \\
\hline
$p_{11}$ & 0.94 &1.84\\
\hline
$p_{12}$ & 0 &0 \\
\hline
\end{tabular}
\end{center}
\end{table}
We also performed simulations in environments with 
convex obstacles satisfying the curvature condition in Assumption \ref{as:4}. The first simulation was performed in an environment with ellipsoidal obstacles and the second one in an environment with polygonal obstacles. The results in Fig. \ref{Long_Short} show the effectiveness of the proposed approach in convex worlds with smooth and non-smooth boundaries where all the trajectories converge safely to the target (red dot). Note that the robot's navigation was successful in the environment shown in Fig. \ref{Long_Short}-\subref{fig:poly_sim}, although it contains an L-shaped non-convex obstacle. In fact, in view of the position of the target and during all the avoidance maneuvers, only a convex  curve is detected from the boundary of this L-shaped obstacle. Simulation videos can be found at \url{https://youtu.be/Y5dho-ptkm8} and \url{https://youtu.be/FZ0qxx6Gsog}.
\begin{figure}[h!]
     \centering
     \subfloat[Convex world with elliptical obstacles.]{\includegraphics[width=0.35\linewidth,height=0.35\linewidth,keepaspectratio]{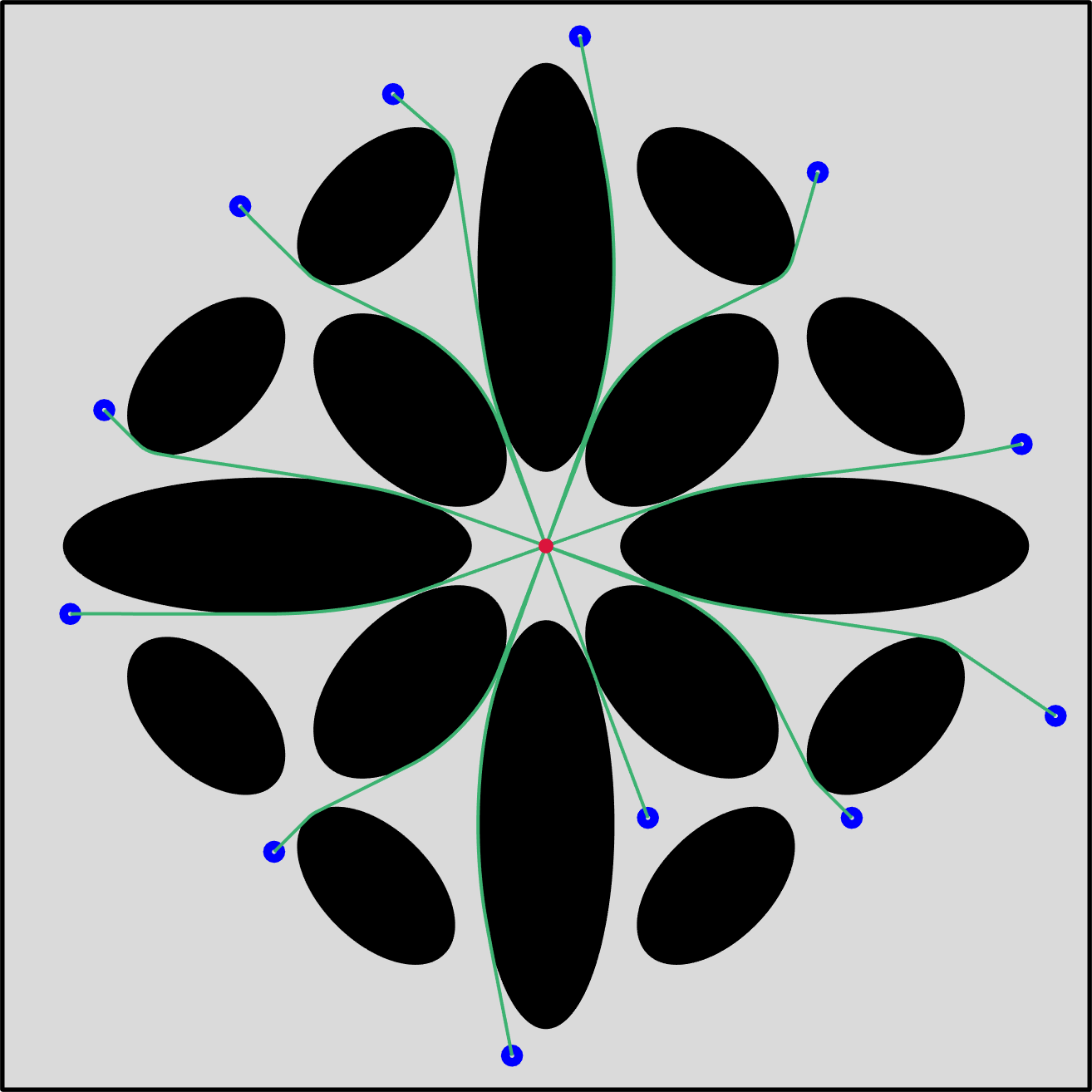}\label{fig:conv_sim}}
      \hspace{0.01cm}
     \subfloat[Convex world with polygonal obstacles.]{\includegraphics[width=0.3505\linewidth,height=0.3505\linewidth]{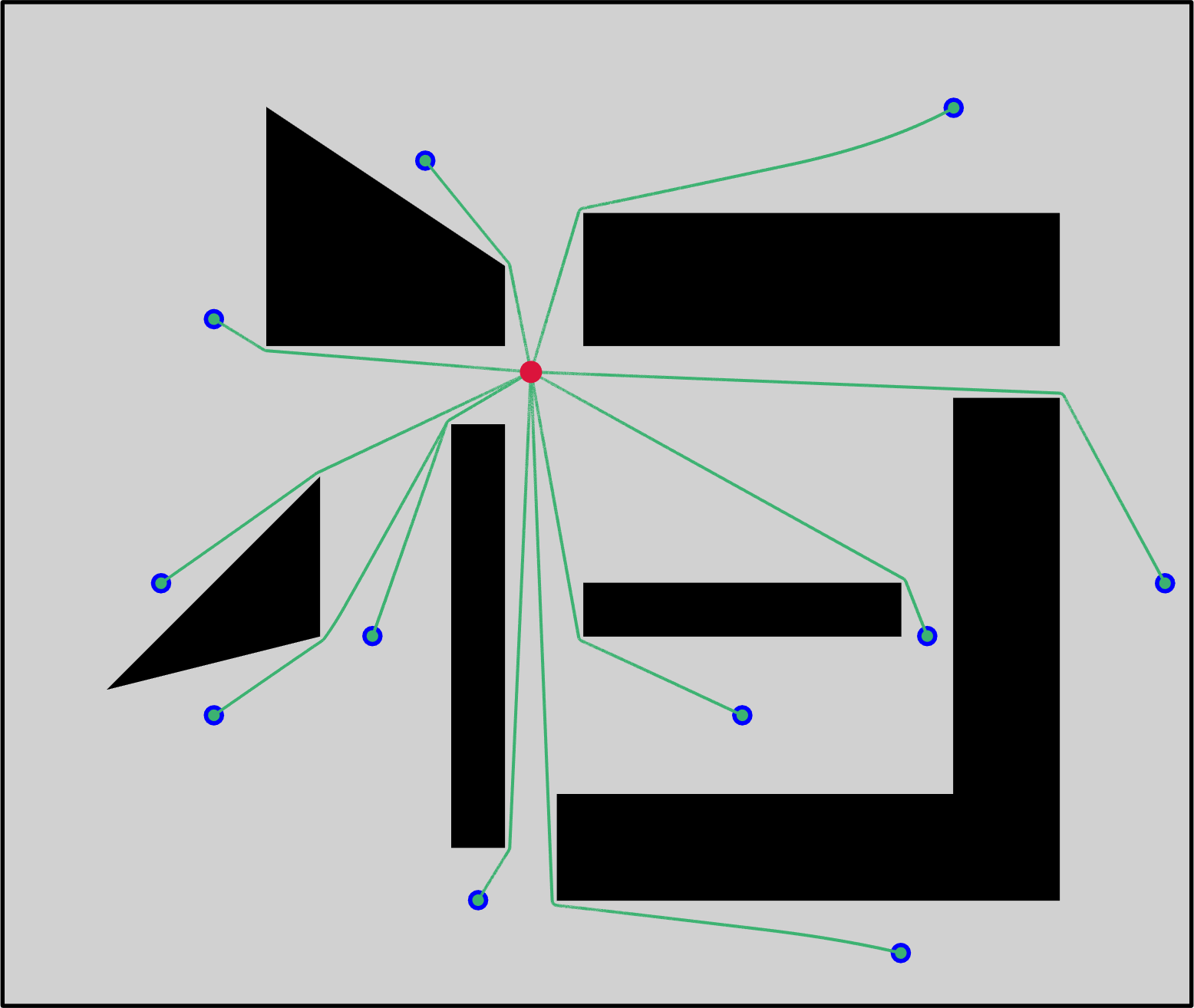}\label{fig:poly_sim}}
     \caption{Sensor-based navigation in unknown 2D convex worlds.}
     \label{Long_Short}
\end{figure}
\subsubsection{Gazebo simulation}
For experimental validation of our sensor-based approach, we used the meta-operating system ROS (Noetic) running on Ubuntu (20.04.6) to implement the control strategy \eqref{sensor_based_ctrl} on a Turtlebot3 model and simulate real-world scenarios with Gazebo (3D dynamic simulator). Our code is written in Python, and the data analysis is performed in MATLAB. The TurtleBot model includes a $360^\circ$ Lidar with a resolution of $1^\circ$, a maximum range $R=3.4\,m$, and a minimum range of $0.12 \,m$. The robot's position and orientation are obtained by subscribing to the odometry topic provided by ROS. Zero mean Gaussian noise is added to the sensors' data where the standard deviation for the Lidar is $0.02\,m$, the standard deviation for the position is $0.03\,m$, and for the orientation, the standard deviation is $0.035\,rd$. As the TurtleBot has a disk-shaped base of radius $r_b=0.14\,m$, we consider the eroded workspace $\mathcal{W}_r:=\mathcal{W}\setminus\left(\partial\mathcal{W}\oplus\mathcal{B}(0,r)\right)$, and the dilated obstacles $\tilde{\mathcal{O}_i^r}$. The eroded free space is then defined as $\mathcal{F}_r:=\mathcal{W}_r\setminus\cup_{i\in\mathbb{I}}\tilde{\mathcal{O}}_i^r$ and for all $x\in\mathcal{F}_r,\,\mathcal{B}(x,r_b)\subset\mathcal{F}$. Considering that $x$ (the center of the robot's base) evolves in the eroded free space $\mathcal{F}_r$, and choosing the dilation parameter $r=r_b+r_s$, where $r_s=0.11\,m$ is a security margin, the robot is guaranteed to evolve in the free space $\mathcal{F}$. TurtleBot 3 is a differential drive robot whose kinematic model is represented by
\begin{align}
    \begin{cases}
      \Dot{x}=v[\cos(\psi)\,\sin(\psi)]^\top, \\
      \Dot{\psi}=\omega,
    \end{cases}
\end{align}
where $\psi\in (-\pi,\pi]$ is the robot's orientation, and $v\in\mathbb{R}$ and $\omega\in\mathbb{R}$ are, respectively, the robot's linear and angular velocity inputs. As the control law \ref{sensor_based_ctrl} was designed for fully actuated robots, a transformation is required to generate adequate velocity inputs for our robot. The principal idea is to rotate the robot so that its orientation coincides with the direction of $u(x)$ obtained from \eqref{sensor_based_ctrl}, and then translate the robot with a linear velocity equal to the magnitude of $u(x)$. The direction of $u(x)$ is denoted by $\psi_d=\overline{\mathrm{atan2}}(u(x))$, and the difference between the robot's orientation and the direction of $u(x)$ is denoted by $\Delta \psi=\psi-\psi_d\in(-\pi,\pi]$. We transform the velocity input of a fully actuated robot to the velocity inputs of a nonholonomic mobile robot with smooth switching between the rotation and translation using the following transformation (inspired from \citep{sawant2024hybrid}):
\begin{align}\label{diff_drive_ctrl}
    \begin{cases}
      v=\min\left(v_{max},k_v\|u(x)\|\cos(\frac{\Delta\psi}{2})^{2p}\right), \\
      \omega=\omega_{max}\sin(\frac{\Delta\psi}{2}),
    \end{cases}
\end{align}
where $k_v>0$, $p\geq1$, $v_{max}=0.26\,m/s$ and $\omega_{max}=1.82\,rd/s$ are the maximum supported velocities by the robot's actuators. The higher the exponent $2p$, the safer our robot is, and its trajectory is closer to the one generated by the control \eqref{sensor_based_ctrl}. We created a Gazebo world cluttered with obstacles whose dilated versions adhere to Assumption \ref{as:4}, and we added non-convex walls facing the destination. In this Gazebo world, we implemented the transformed control law \eqref{diff_drive_ctrl} on the TurtleBot 3, where we set the gain to $k_v=0.8$ and the exponent $p=3$. The results are shown in Fig. \ref{TurtlBot}, and the simulation video can be found at \url{https://youtu.be/g1Ya9RFSgJc}.
\begin{figure}[h!]
\renewcommand*\thesubfigure{\arabic{subfigure}}
     \centering
     \subfloat[]{\includegraphics[width=0.37\linewidth,height=0.37\linewidth,keepaspectratio]{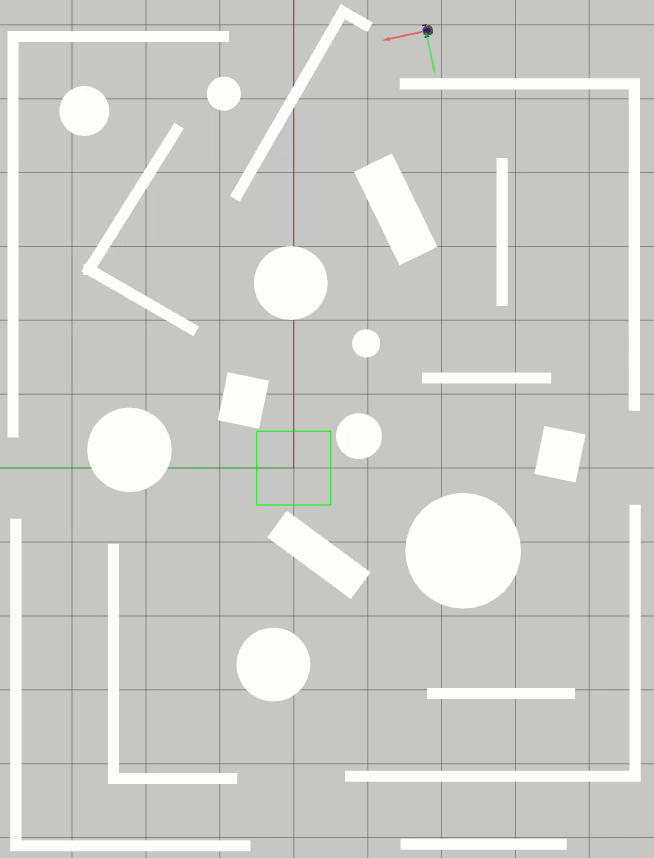}\label{Gt0}}
    \hspace{0.01cm}
     \subfloat[]{\includegraphics[width=0.37\linewidth,height=0.37\linewidth,keepaspectratio]{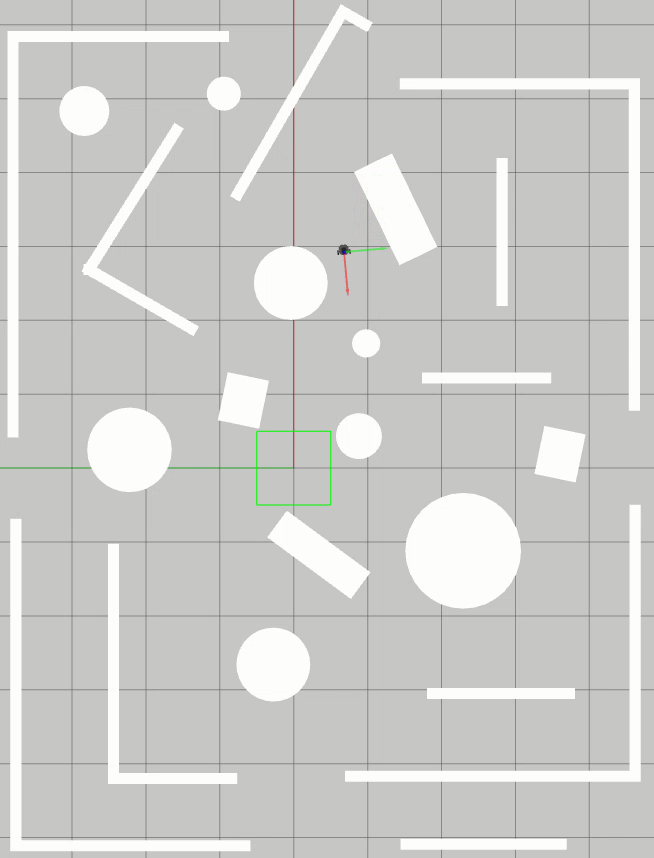}\label{Gt15}}
     \hspace{0.01cm}
     \subfloat[]{\includegraphics[width=0.37\linewidth,height=0.37\linewidth,keepaspectratio]{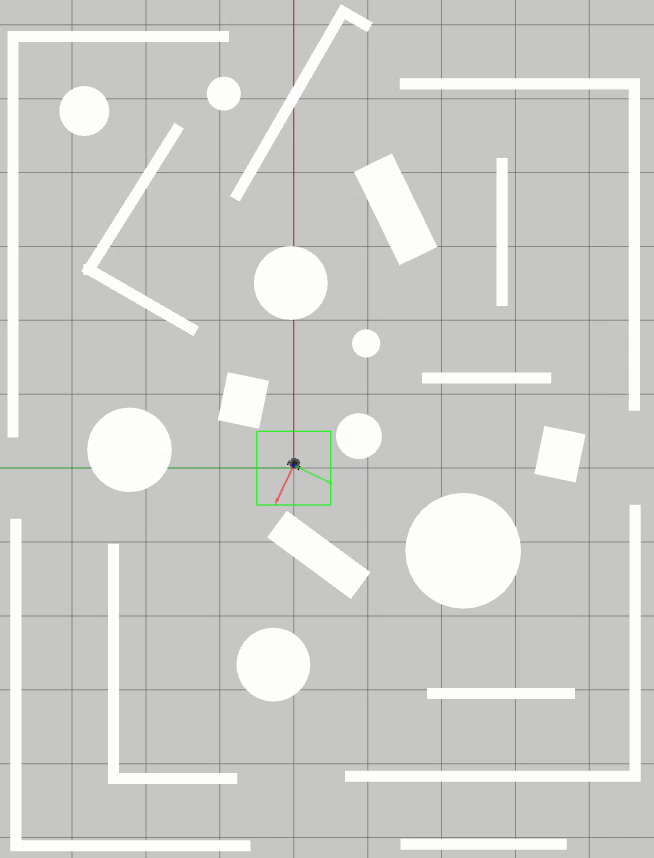}\label{Gt61}}\\
     \subfloat[]{\includegraphics[width=0.35\linewidth,height=0.35\linewidth,keepaspectratio]{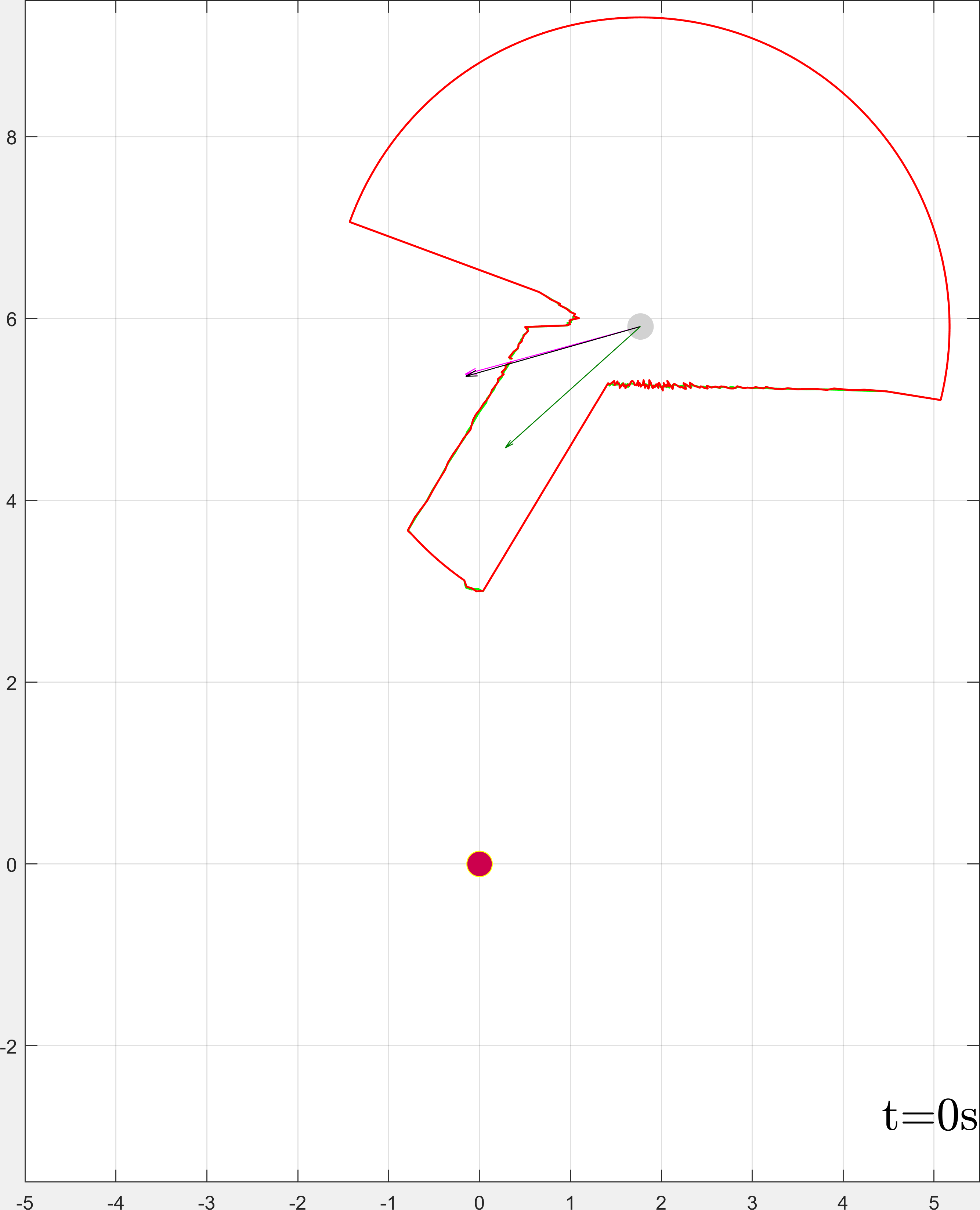}\label{t0}}
    \hspace{0.01cm}
    \subfloat[]{\includegraphics[width=0.35\linewidth,height=0.35\linewidth,keepaspectratio]{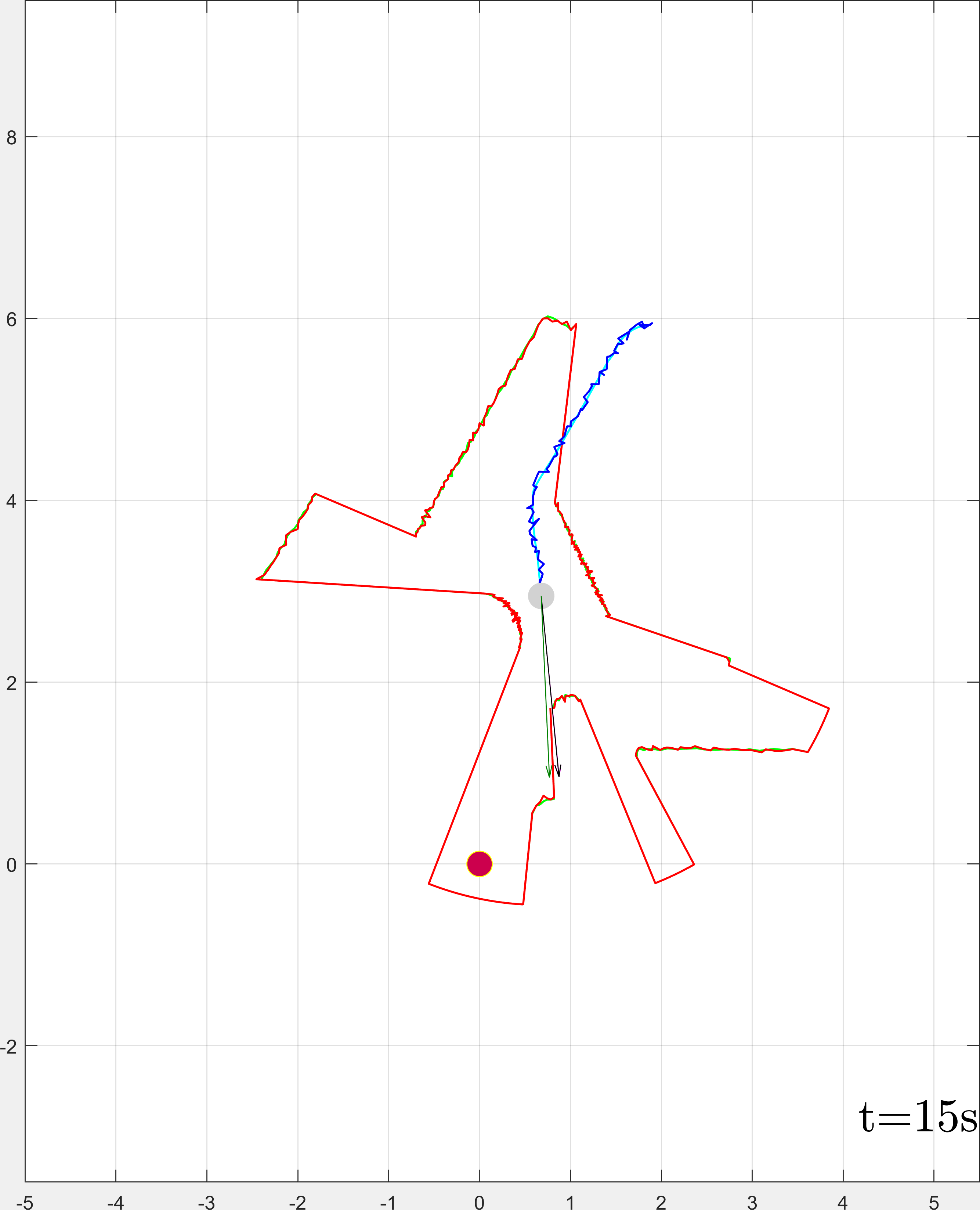}\label{t15}}
     \hspace{0.01cm}
     \subfloat[]{\includegraphics[width=0.35\linewidth,height=0.35\linewidth,keepaspectratio]{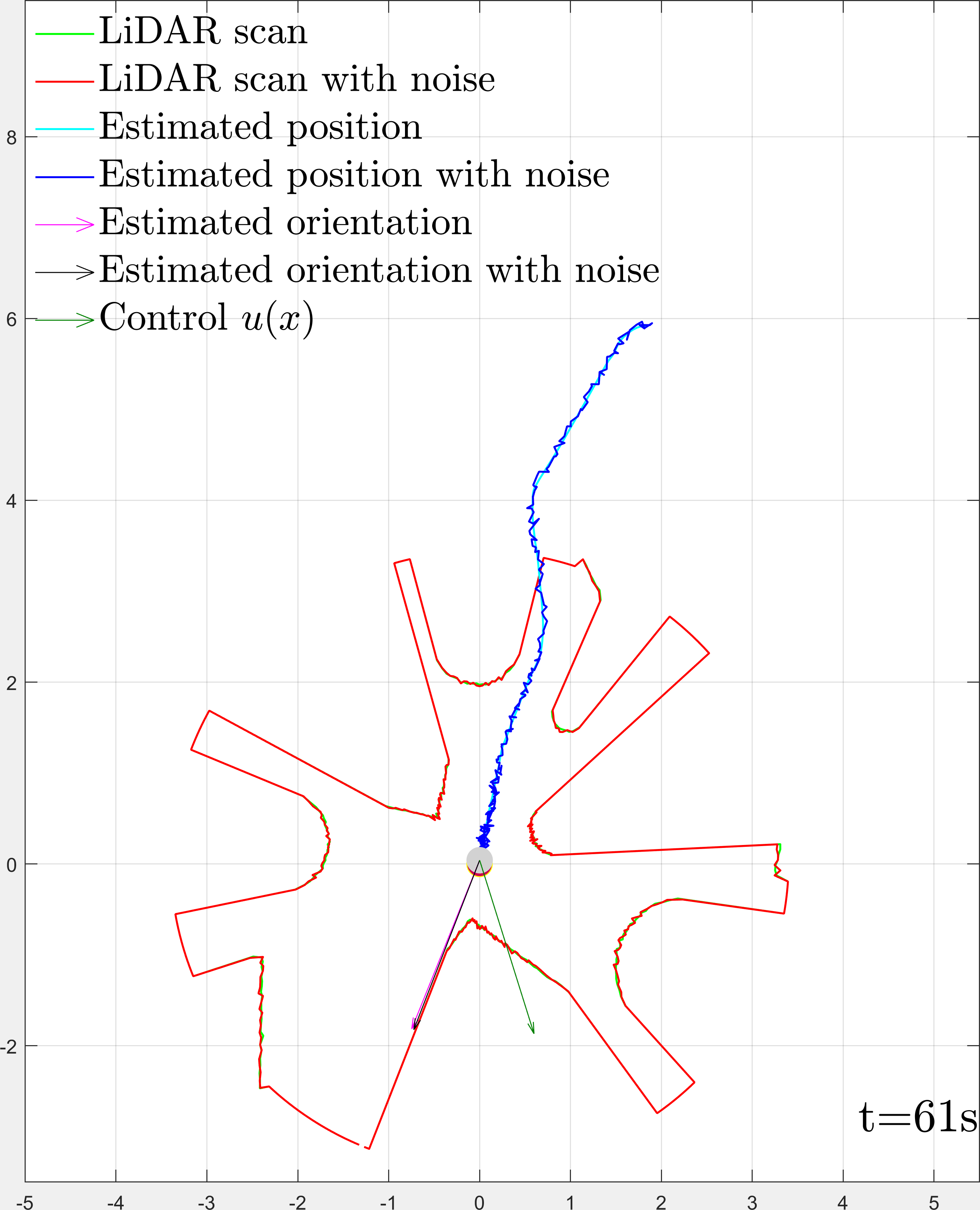}}\label{t61}
     \caption{Time-stamped shots of Turtlebot 3 navigating a Gazebo world.}
     \label{TurtlBot}
\end{figure}
\section{Conclusion}
A continuous feedback control strategy for the autonomous navigation problem in an $n$-dimensional sphere world has been proposed, with safety guarantees and {\it quasi-optimal} trajectory generation. The proposed strategy consists in steering the robot tangentially to the blocking obstacles through successive projections of the nominal control onto the obstacles enclosing cones. Consequently,  the deviations from the nominal direction to the target are minimized with respect to each blocking obstacle, resulting in a {\it quasi-optimal} overall collision-free trajectory. The price to pay for the almost global asymptotic stability result, in two-dimensional sphere worlds, is a somewhat restrictive assumption on the configuration of the obstacles (Assumption 3) that has been lifted in the sensor-based version, where the robot can navigate to the target location from almost everywhere in the free space without prior knowledge of environment containing sufficiently curved convex obstacles. Extending the proposed approach to arbitrarily shaped obstacles, with global asymptotic stability guarantees, is another interesting problem that will be the main focus of our future work.

\section*{Appendix}
\subsection*{Proof of Lemma \ref{lem1}}\label{appendix:Lemma 1}
Minimizing the angle $\angle(x_d-x,v_i)$ is equivalent to minimizing the cost function $g(v_i)=1-V_d^\top\frac{v_i}{\|v_i\|}$ with $V_d=(x_d-x)/\|x_d-x\|$ under the constraint $\Gamma(v_i)=\frac{v_i^\top V_{ci}}{\|v_i\|}-\cos(\theta_i)=0$ with $V_{ci}=(c_i-x)/\|c_i-x\|$. Define the Lagrangian associated to the optimization problem \eqref{min} by $ L_{\lambda}(v_i)=g(v_i)-\lambda \Gamma(v_i)$ where $\lambda$ is the Lagrange multiplier. The optimum is the solution of 
\begin{align*}
    \nabla_{v_i}L_{\lambda}(v_i)=0,
    \nabla_{\lambda}L_{\lambda}(v_i)=0,
\end{align*}
which gives
\begin{align}
        \pi^{\bot}(v_i)(V_d+\lambda V_{ci})=0,\quad 
        \frac{v_i^\top V_{ci}}{\|v_i\|}-\cos(\theta_i)=0.\label{32}
\end{align}
From the first equation, one has $v_i=\alpha(V_d+\lambda V_{ci})$ for some $\alpha\in\mathbb{R}$. Substituting this into the second equation, one gets
\begin{align}\label{aa}
    \alpha(\cos(\beta_i)+\lambda)=\cos(\theta_i)\|\alpha(V_d+\lambda V_{ci})\|.
\end{align}
Squaring \eqref{aa} and substituting $\|\alpha(V_d+\lambda V_{ci})\|^2=\alpha^2(\lambda^2+2\lambda\cos(\beta_i)+1)$, one can solve for $\lambda$ 
\begin{align}
    \lambda_{1,2}=-\frac{\sin(\theta_i\pm\beta_i)}{\sin(\theta_i)}.
\end{align}
Consequently, one can obtain $v_i^1$ and $v_i^2$ as follows:
\begin{align}
      v_i^{1,2}=\pm|\alpha|\left(V_d-\frac{\sin(\theta_i\pm\beta_i)}{\sin(\theta_i)}V_{ci}\right).
\end{align}
The value of $g$ at the two solutions is as follows:
\begin{align*}
    g(v_i^1)=1+\cos(\theta_i+\beta_i),\;
    g(v_i^2)=1-\cos(\theta_i-\beta_i),
\end{align*}
and $g(v_i^1)-g(v_i^2)=2\cos(\theta_i)\cos(\beta_i)\geq0$ which implies that
\begin{align}
  \mathcal{U}(x)=\{\Bar{\alpha}(V_d-\sin^{-1}(\theta_i)\sin(\theta_i-\beta_i)V_{ci})|\; \Bar{\alpha}\geq0\}.
\end{align}
When $x\in\mathcal{S}(x_d,c_i)$, $u_d(x)\in\mathcal{V}(c_i-x,\theta_i)$ which implies that $\theta_i=\beta_i$, and $u_d(x)\in\mathcal{U}$. Therefore, for all $x\in\mathcal{S}(x_d,c_i)$, $u(x)\in\mathcal{U}$ implies that $u(x)=\Bar{\alpha}V_d$, and if in addition $u(x)=u_d(x)$, then $\Bar{\alpha}=\gamma\|x_d-x\|$. One can conclude that the solution is unique and is given by
\begin{align*}
\small
    u(x)&=\gamma\|x_d-x\|\left(V_d-\frac{\sin(\theta_i-\beta_i)}{\sin(\theta_i)}V_{ci}\right)=\xi(u_d(x),x,i), 
\end{align*}
where the last equation is obtained after some straightforward manipulations.
\subsection*{Proof of Lemma \ref{lem2}}\label{appendix:Lemma 2} 
Let $x(0)\in\mathcal{F}\setminus\mathcal{L}_d(x_d,c_i)$. Then, one has two situations. First, when $x(0)\in\mathcal{VI}$, the trajectory $x(t)$ is a line-segment which is the closest path. Now, when $x(0)\in\mathcal{D}(x_d,c_i)$, there are two types of possible trajectories: trajectories inside the enclosing cone $\mathcal{C}^{\leq}_{\mathcal{F}}(x,c_i-x,\theta_i)$ and trajectories outside this cone. One can show that the trajectory generated by the closed-loop system \eqref{12}-\eqref{25}, on the enclosing cone $\mathcal{C}^{=}_{\mathcal{F}}(x,c_i-x,\theta_i)$, has minimum length. For the first type of trajectory, one only considers the ones between the line segment $\mathcal{L}_s(x(0), x_d)$ and the closest tangent to it (blue segment in Fig. \ref{fig:fig6}) among the cone enclosing the obstacle (the red trajectory in Fig. \ref{fig:fig6} is an example). All these trajectories will merge with our trajectory, which is on the closest tangent (as shown in Lemma \ref{lem1}), at the intersection point of the tangent with the obstacle. Since, before the intersection point, our trajectory is a line segment, one can conclude that it is the shortest path. The best that can be achieved outside the cone for a smooth trajectory is a dilated version of our trajectory (larger radius of curvature) which is longer than ours (black path in Fig. \ref{fig:fig6}).
\subsection*{Proof of Lemma \ref{lem3}}\label{appendix:lemma 3}
First, we prove that the closed-loop system admits a unique solution. The control is Lipschitz on $\mathcal{VI}$ since $u(x)=u_d(x)$ is continuously differentiable. When  $x\in\mathcal{BL}$, for simplicity,  $\sin(\theta_{\iota_x(p)}(q)-\beta_{\iota_x(p)}(u_{p-1}(x),q))$ is denoted by $s_{\iota_x(p)}^s(q)$ and $\frac{\sin(\beta_{\iota_x(p)}(u_{p-1}(x),q))}{\sin(\theta_{\iota_x(p)}(q))}$ by $s_{\iota_x(p)}^d(q)$, where $p\in\{1,\dots,h(x)\}$. After manipulation, the control \eqref{36} can be expressed as $u(x)=u_d(x)-\gamma\|x-x_d\|\sum\limits_{p=1}^{h(x)}\prod\limits_{k=1}^{p-1} s_{\iota_x(k)}^d(x)\frac{s_{\iota_x(p)}^s(x)}{r_{\iota_x(p)}}(c_{\iota_x(p)}-x)$, which is shown to be one-sided Lipschitz as follows:
    \begin{align*}
        \left(u(x)-u(y)\right)^\top(x-y)&=-\gamma\|x-y\|^2
        -\gamma\|x_d-x\|\sum\limits_{p=1}^{h(x)}\prod\limits_{k=1}^{p-1} s_{\iota_x(k)}^d(x)\frac{s_{\iota_x(p)}^s(x)}{r_{\iota_x(p)}}(c_{\iota_x(p)}-x)^\top(x-y)\\&\qquad+\gamma\|x_d-y\|\sum\limits_{p=1}^{h(y)}\prod\limits_{k=1}^{p-1} s_{\iota_x(k)}^d(y)\frac{s_{\iota_x(p)}^s(y)}{r_{\iota_x(p)}}(c_{\iota_x(p)}-y)^\top(x-y),\\
        &\leq-\gamma\|x-y\|^2
        +\gamma\|x_d-x\|\|x-y\|\sum\limits_{p=1}^{h(x)}\prod\limits_{k=1}^{p-1} s_{\iota_x(k)}^d(x)\frac{s_{\iota_x(p)}^s(x)}{r_{\iota_x(p)}}\|c_{\iota_x(p)}-x\|\\&\qquad+\gamma\|x_d-y\|\|x-y\|\sum\limits_{p=1}^{h(y)}\prod\limits_{k=1}^{p-1} s_{\iota_x(k)}^d(y)\frac{s_{\iota_x(p)}^s(y)}{r_{\iota_x(p)}}\|c_{\iota_x(p)}-y\|.
    \end{align*}
    Note that $\forall x\in\mathcal{BL}$ and $\forall p\in\{1,\dots,h(x)\}$, $0\leq s_{\iota_x(p)}^d(x)\leq1$, $0\leq s_{\iota_x(p)}^s(x)\leq1$, $\|c_{\iota_x(p)}-x\|\leq 2r_0-r_{\iota_x(p)}$ and $\|x_d-x\|\leq 2r_0$, which implies that there exists $M>0$ such that $\|x_d-x\|\sum_{p=1}^{h(x)}\frac{\|c_{\iota_x(p)}-x\|}{r_{\iota_x(p)}}\leq M\|x-y\|$. Therefore,
    \begin{align*}
        \left(u(x)-u(y)\right)^\top(x-y)&\leq-\gamma\|x-y\|^2+\gamma M_1\|x-y\|^2+\gamma M_2\|x-y\|^2\\&\leq\gamma(-1+M_1+M_2)\|x-y\|^2\\&\leq L\|x-y\|^2.
    \end{align*}
One can take $L=\gamma(-1+M_1+M_2)$ where $M_1>0$, $M_2>0$ and $M_1+M_2>1$. The control \eqref{36} is one-sided Lipschitz \citep{uniqueness} when $x\in\mathcal{BL}$, and is Lipschitz when $x\in\mathcal{VI}$. Thus, according to \cite[Proposition 2]{uniqueness}, the closed-loop system \eqref{12}-\eqref{36} has a unique solution for all $x(0)\in\mathcal{F}$.
Now, we prove forward invariance using Nagumo's theorem. We only need to verify Nagumo's condition at the free space boundary as it is trivially met when $x\in\mathring{\mathcal{F}}$ where $\mathcal{T}_{\mathcal{F}}(x)=\mathbb{R}^n$. Since the free space is a sphere world, the tangent cone on its boundary is the half-space $\mathcal{C}^{\leq}_{\mathbb{R}^n}(x,-x,\frac{\pi}{2})$ when $x\in\partial\mathcal{W}$ and $\mathcal{C}^{\geq}_{\mathbb{R}^n}(x,c_i-x,\frac{\pi}{2})$ when $x\in\partial\mathcal{O}_i$ (see Fig. \ref{fig:Tangent_cones}). Let us consider an obstacle $\mathcal{O}_i$  and verify Nagumo's condition in three regions of the free space.\\
In the first region, When $x\in\partial\mathcal{W}$, $\mathcal{T}_{\mathcal{F}}(x)=\mathcal{C}^{\leq}_{\mathbb{R}^n}(x,-x,\frac{\pi}{2})$ and two sub-regions must be considered.
            \begin{itemize}
             \item $x\in\partial\mathcal{W}\cap\partial\mathcal{BL}$ (brown arc in Fig. \ref{fig:Tangent_cones}): Since $u(x)\in \mathcal{C}^=_{\mathcal{F}}(x,c_i-x,\theta_i)$ and $\mathcal{C}^=_{\mathcal{F}}(x,c-x,\theta)\subseteq\mathcal{C}^{\leq}_{\mathbb{R}^n}(x,-x,\frac{\pi}{2})$, one concludes that $u(x)\in\mathcal{T}_{\mathcal{F}}(x)$.
             \item $x\in\partial\mathcal{W}\setminus\partial\mathcal{BL}$ (grey arc in Fig. \ref{fig:Tangent_cones}): Since $u_d(x)\in\mathcal{F}$ and $\mathcal{F}\subseteq\mathcal{C}^{\leq}_{\mathbb{R}^n}(x,-x,\frac{\pi}{2})$, one concludes that $u(x)=u_d(x)\in\mathcal{T}_{\mathcal{F}}(x)$.
            \end{itemize}
In the second region, $x\in\partial\mathcal{O}_i\cap\mathcal{AR}_i$ (green arc in Fig. \ref{fig:Tangent_cones}) and $\mathcal{T}_{\mathcal{F}}(x)=\mathcal{C}^{\geq}_{\mathbb{R}^n}(x,c_i-x,\frac{\pi}{2})$. Since $u(x)\in\mathcal{C}^{=}_{\mathbb{R}^n}(x,c_i-x,\frac{\pi}{2})\subset\mathcal{C}^{\geq}_{\mathbb{R}^n}(x,c_i-x,\frac{\pi}{2})$, one concludes that $u(x)\in\mathcal{T}_{\mathcal{F}}(x)$.
Finally, in the last region, $x\in\partial\mathcal{O}_i\setminus\mathcal{AR}_i$ (blue arc in Fig. \ref{fig:Tangent_cones}) and $\mathcal{T}_{\mathcal{F}}(x)=\mathcal{C}^{\geq}_{\mathbb{R}^n}(x,c_i-x,\frac{\pi}{2})$. Since $x\notin\mathcal{AR}_i$, $\forall p\in\{0,\cdots,h(x)\}$, obstacle $\mathcal{O}_i$ is not selected in the successive projections $(\iota_x(p)\neq i)$ and $u_p(x)\notin\mathcal{C}^{\leq}_{\mathbb{R}^n}(x,c_i-x,\frac{\pi}{2})$. Therefore, $u(x)$ must be in the complement of the enclosing cone to the obstacle $\mathcal{O}_i$. Thus, one can conclude that $u(x)\in\mathcal{C}^{\geq}_{\mathbb{R}^n}(x,c_i-x,\frac{\pi}{2})=\mathcal{T}_{\mathcal{F}}(x)$. 
Since $\forall x\in\mathcal{F},\;\; u(x)\in\mathcal{T}_{\mathcal{F}}(x)$ 
and the solution of the closed-loop system \eqref{12}-\eqref{36} is unique, it follows that the free space $\mathcal{F}$ is positively invariant and the closed-loop system \eqref{12}-\eqref{36} is safe.
\subsection*{Proof of Lemma \ref{lem5}}\label{appendix:Lemma 5}
Let $\mathcal{L}^e_i$ be the central half-line associated to obstacle $\mathcal{O}_i$, and let $y\in\mathcal{L}^e_i\setminus\mathcal{O}_i$. Since the control is tangent to the obstacle, for any $x\in\mathcal{AR}_i\setminus\mathring{\mathcal{H}}(y,c_i)$ the control, at a step $p$, can never point inside the hat $\mathcal{H}(y,c_ki)$, {\it i.e.,} there is no position $q\in\mathcal{AR}_i\cap\mathring{\mathcal{H}}(y,c_i)$ such that $\angle (q-x,u_p(x))=0$, where $p=\iota^{-1}_x(i)$. Assume that $\mathcal{M}_i\neq\varnothing$ and consider an obstacle $k\in\mathcal{M}_i$ such that $k=\kappa_i^{-1}(1)$, $c_k\in\mathring{\mathcal{H}}(y,c_i)$, and $x^*_{k,i}=y$. Assume that $\mathring{\mathcal{H}}(x^*_{k,i},c_i)\cap\mathcal{O}_l=\varnothing$ for all $l\in\mathbb{I}\setminus\{i,k\}$, which ensures that no other obstacle interferes and changes the direction of the control $u_p$ towards the hat $\mathring{\mathcal{H}}(x^*_{k,i},c_i)$. Consequently, there is no $x\in\mathcal{AR}_k\cap\mathcal{AR}_i$ such that $\angle (c_k-x,u_p)=\beta(c_k-x,u_p)=0$ where $p=\iota_x^{-1}(i)$. Thus, obstacle $k$ does not generate a set of undesired equilibria $\mathcal{L}_k$ (see Fig. \ref{fig:lemma4pr}). Following the same reasoning, one can show that obstacle $j=\kappa_i^{-1}(2)$ will not generate a set of undesired equilibria by considering the obstacles $\mathcal{O}_i$ and $\mathcal{O}_k$ as single obstacle whose hat is the union $\mathring{\mathcal{H}}(x^*_{j,i},c_i)\cup\mathring{\mathcal{H}}(x^*_{j,i},c_k)$ and $\mathcal{L}^{e}_i$ as its single central half-line since obstacle $k=\kappa_i^{-1}(1)$ is free of undesired equilibria. These considerations reduce to the first case where if $c_j\in\mathring{\mathcal{H}}(x^*_{j,i},c_i)\cup\mathring{\mathcal{H}}(x^*_{j,i},c_k)$ and $(\mathring{\mathcal{H}}(x^*_{j,i},c_i)\cup\mathring{\mathcal{H}}(x^*_{j,i},c_k))\cap\mathcal{O}_l=\varnothing$ for all $l\in\mathbb{I}\setminus\{i,k\}$, obstacle $j$ does not generate undesired equilibria. The same can be applied up to obstacle $\kappa_i^{-1}(p)$ by considering the union of the hat of obstacle $i$ and obstacles $\mathcal{M}_i^{p-1}$, and the central half-line $\mathcal{L}^e_i$. Therefore, the obstacles of indices in the set $\mathcal{M}_i^{p}$ are free of undesired equilibria if, for all $j\in \mathcal{M}_i^{p}$ with $p\in\{1,\dots,N_i\}$, 1) $c_j\in\cup_{l\in\mathcal{M}_i^{p1}}\mathring{\mathcal{H}}(x^*_{j,i},c_l)$ and 2) $\cup_{l\in\mathcal{M}_i^{p-1}}\mathring{\mathcal{H}}(x^*_{j,i},c_l)\cap\mathcal{O}_z=\varnothing$ for all $z\in\mathbb{I}\setminus(\mathcal{M}_i^j\cup\{j\})$. If, in addition, $p=N_i$, or obstacle $\kappa_i^{-1}(p+1)$, with $p<N_i=\textbf{card}(M_i)$, does not satisfy conditions 1) and 2), we say that $\Bar{N}_i=p$ is the order of the set $\mathcal{M}_i$ and the number of obstacles free of undesired equilibria with indices grouped in the set $\mathcal{M}_i^{\Bar{N}_i}$, which concludes the proof.
\begin{figure}[!h]
\centering
\includegraphics[scale=0.38]{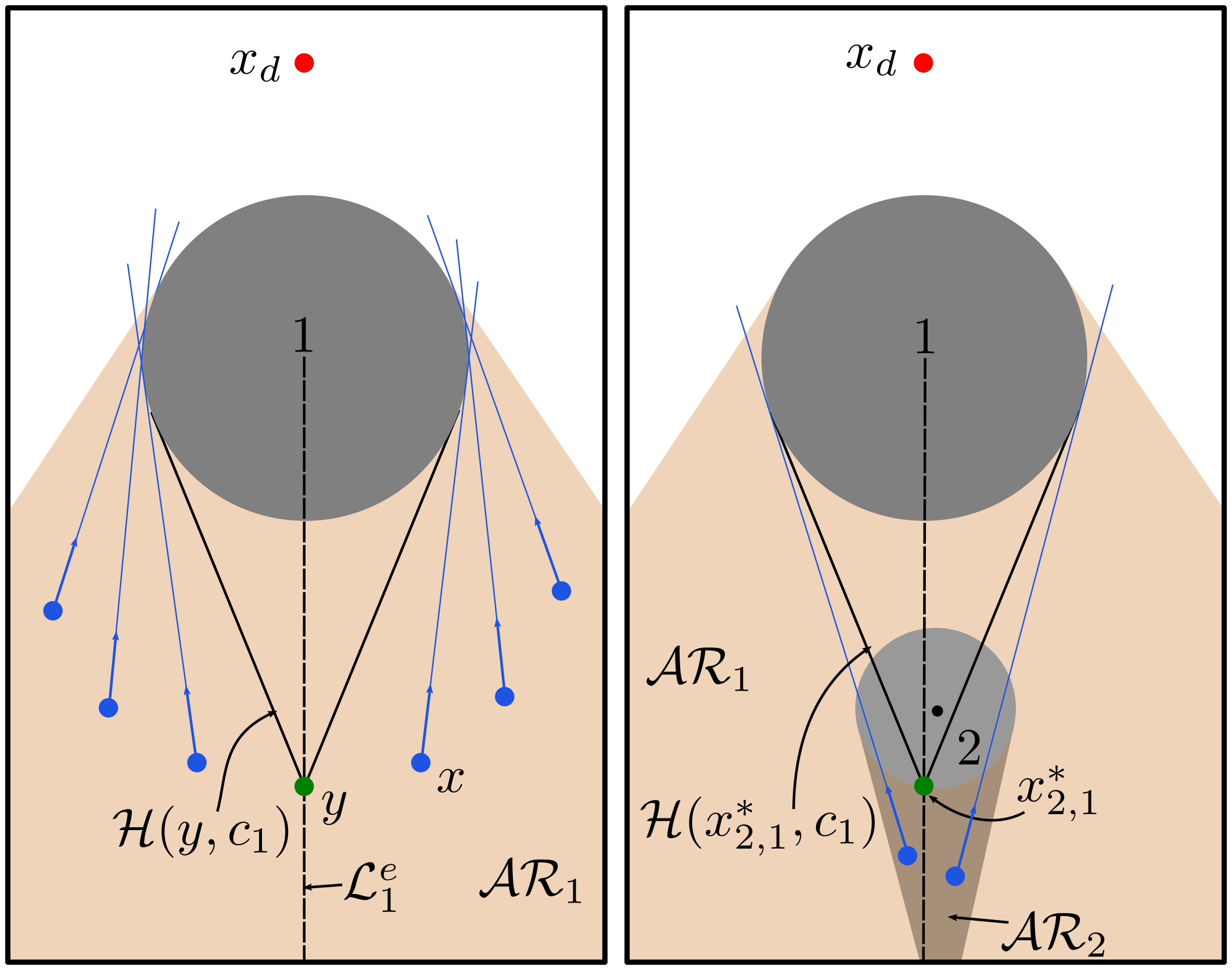}
\caption{The left figure shows that for all position $x\in\mathcal{AR}_1\setminus\mathring{\mathcal{H}}(y,c_1)$, the control cannot point inside $\mathcal{H}(y,c_1)$. In the right figure, obstacle $2$ is placed such that $c_2\in\mathring{\mathcal{H}}(y,c_1)$, the intermediary control $u_1$ cannot point inside $\mathcal{H}(y,c_1)$ and $y=x^*_{2,1}$. Then, the intermediary control $u_1$ cannot point into the center $c_2$ at any position $x\in\mathcal{AR}_2$ which implies that $\mathcal{L}_2$ is an empty set.}
\label{fig:lemma4pr}
\end{figure}
\subsection*{Proof of Theorem \ref{the1}}\label{appendix:the1}
Item i) and item ii) follow directly from lemma \ref{lem3} and lemma \ref{lem4} respectively. Now let us prove item iii).
Consider the set of undesired equilibria $\mathcal{L}_i$, $i\in\mathcal{Z}$. Define the tube surrounding $\mathcal{L}_i$ inside the free space $\mathcal{TU}_i:=\{x\in\mathcal{F}|d(x,\mathcal{L}_i)\leq e_i\}$ where $e_i$ is small such that $\mathcal{TU}_i\subset\mathcal{AR}_i$, and $e_i<r_i$. Let $V(x)=\frac{1}{2}d^2(\mathcal{L}_i,x)=\frac{1}{2}(x-c_i)^{\top}\pi^{\bot}(\bar v_i)(x-c_i)$ where $\bar v_i=(\bar x_i-c_i)/\|\bar x_i-c_i\|$, $\bar x_i\in\mathcal{L}_i$, $V(x)=0$ for all $x\in\mathcal{L}_i$, and $V(x)>0$ for all $x\in\mathcal{TU}_i\setminus\mathcal{L}_i$. Let $U:=\mathcal{TU}_i\cap\mathcal{P}_{\leq}(\bar x_i,\bar v_i)\setminus\mathcal{L}_i$ where $\bar x_i\in\mathcal{L}_i$ is such that $U\cap\mathcal{L}_k=\varnothing$ for all $k\in\mathcal{Z}\setminus\{i\}$, and $U\subset\mathcal{AR}_i^h$ with $\mathcal{AR}_i^h:=\{q\in\mathcal{AR}_i|\iota_q^{-1}(i)=h(x)\}$. Note that $V(x)>0$ for all $x\in\mathcal{U}$. The time-derivative of $V(x)$ is given by $\Dot{V}(x)=\frac{\partial V(x)}{\partial x}^\top\Dot{x}=(x-c_i)^{\top}\pi^{\bot}(\bar v_i)u(x)$. Since $e_i<r_i$ and for all $x\in U$, $u(x)$ lies on the cone enclosing obstacle $\mathcal{O}_i\subset\mathcal{P}_{\leq}(\bar x_i,\bar v_i)$, $0<\angle(\bar v_i,x-c_i)<\pi/2$ and $\pi/2<\angle(\bar v_i,u(x))<\pi$. Therefore, $\Dot{V}(x)>0$ for all $x\in U$. As $U$ is a compact set, $V(x)$ is increasing on $U$, and $V(x)=0$ on $\mathcal{L}_i$ (the axis of the tube), $x(t)$ must leave the set $U$. The set $U$ is bounded on top by obstacle $i$, its lateral boundary is the surface of the tube, and is bounded from the bottom by the hyperplane $\mathcal{P}_=(\bar x_i, \bar v_i)$. Due to the safety of the system, as per Lemma \ref{lem3}, $x(t)$ can not leave $U$ from the upper boundary (the boundary of obstacle $i$), and since $\pi/2<\angle(\bar v_i,u(x))<\pi$, $x(t)$ can only leave $U$ from the surface of the tube. 
Now, let us prove item iv). Since $x_d\in\mathcal{VI}$, there exists $r_d>0$ such that $\mathcal{B}(x_d,r_d)\subset\mathcal{VI}$. As the closed-loop system \eqref{12}-\eqref{36} reduces to $\Dot{x}=-\gamma(x-x_d)$ on the visible set $\mathcal{VI}$, the equilibrium $x=x_d$ is locally exponentially stable. Finally, let us prove item v). Consider a trajectory starting from $x(0)\in\mathcal{VL}$ generated by the closed-loop system \eqref{12}-\eqref{36}. Since the control on the visible set $\mathcal{VL}$ is the nominal one $u_d(x)=\gamma(x_d-x)$, the generated trajectory is the line segment connecting $x(0)$ to $x_d$, which has the shortest length. Now consider a trajectory with initial condition $x(0)\in\mathcal{BL}$, generated by the closed-loop system \eqref{12}-\eqref{36}. The velocity of a vehicle moving along the considered trajectory at an instant $t\in[0,\infty)$ is given by $\Dot{x}(t)=u_{h(x(t))}(x(t))$ where $h(x(t))=\mathbf{card}(\mathcal{I}(x(t)))$. Since the virtual destination at position $x(t)$ is the point given by $P(x(t))=x(t)+u_{h(x(t))}(x(t))$, then, the direction from $x(t)$ to the virtual destination is the vehicle's velocity $P(x(t))-x(t)=u_{h(x(t))}(x(t))=\Dot{x}(t)$. Therefore, one can conclude that for $x(0)\in\mathcal{BL}$, the tangent to the trajectory generated by the closed-loop system \eqref{12}-\eqref{36}, at any position $x(t)$, points to the virtual destination $P(x(t))$.
\subsection*{Proof of Lemma \ref{lem6}}\label{appendix:lemma 6}
Let $i\in\mathbb{L}$. Since $\textbf{Cell}_i$ is bounded by line segments of undesired equilibria $(\cup_{k\in\mathcal{Z}}\mathcal{L}_k)$ and the free space boundary, $u(x)$ points inside the cell when $x\in\partial \textbf{Cell}_i\cap\partial\mathcal{F}$, as per Lemma \ref{lem3}, and $u(x)=0$ when $x\in\partial\textbf{Cell}_i\cap(\cup_{k\in\mathcal{Z}}\mathcal{L}_k)$. Consequently, $\textbf{Cell}_i$ is forward invariant for the closed-loop system \eqref{12}-\eqref{36}.
\subsection*{Proof of Lemma \ref{lem7}}\label{appendix:lemma 7}
Since the nests are invariant, as per Lemma \ref{lem6}, and all the undesired equilibria are contained inside the nests, it remains to show that if we start outside nests, we will never get back in. We begin with the special nest ($\textbf{Nest}_0=\cup_{i\in\mathcal{Z}}\mathcal{L}_i\setminus\cup_{j\in\mathbb{L}}\textbf{Cell}_j$) formed by segments of undesired equilibria and show their repellency. These segments can be defined as $\mathcal{CL}_i:=\left\{ q\in\mathcal{L}_i|q\notin\cup_{k\in\mathbb{L}}\textbf{Cell}_k\right\}\subset\textbf{Nest}_0$ for $i\in\mathcal{Z}$.\\
Consider obstacle $i\in\mathcal{Z}$ and segment $\mathcal{CL}_i$ in the following three possible cases illustrated in (Fig. \ref{fig:2dnests}):\\ 
{\bf Case 1:} Consider the region $\mathcal{AR}_i^h:=\left\{q\in\mathcal{AR}_i|\iota_q^{-1}(i)=h(q) \right\}$ where obstacle $i$ is the last on the list of projections. Define the tube $\mathcal{TU}_i:=\{x\in\mathcal{F}|d(x,\mathcal{CL}_i)\leq e_i\}$ where $e_i$ is small enough to have $\mathcal{TU}_i\cap\mathcal{L}_j=\varnothing$ for all $j\in\mathcal{Z}\setminus\{i\}$ and $\mathcal{TU}_i\subset\mathcal{AR}_i$. Let $V(x)=1-\frac{(\Bar{x}_i-c_i)^\top}{\|\Bar{x}_i-c_i\|}\frac{(x-c_i)}{\|x-c_i\|}$ where $\Bar{x}_i\in\mathcal{CL}_i\cap\mathcal{AR}^{h}_i$. Note that $V(\Bar{x}_i)=0$, and $V(x)>0$ for all $x\in \mathcal{TU}_i\setminus\mathcal{CL}_i$.
Define the set $U:=(\mathcal{TU}_i\cap\mathcal{AR}^h_i
)\setminus\mathcal{CL}_i$. The time-derivative of $V(x)$ is given by
\begin{align*}
    \Dot{V}(x)&=\frac{\partial V(x)}{\partial x}^\top\Dot{x},\\&=-\frac{(\Bar{x}_i-c_i)^\top}{\|\Bar{x}_i-c_i\|}J_x\left(\frac{(x-c_i)}{\|x-c_i\|}\right)u(x),\\
    &=-K\Bar{V}^\top_{ci}\pi^{\bot}(V_{ci})\bar\xi_i,
\end{align*}
where $K=\frac{\|u(x)\|}{\|x-c_i\|}>0$, $V_{ci}=\frac{(c_i-x)}{\|c_i-x\|}$, $\bar{V}_{ci}=\frac{(\Bar{x}_i-c_i)}{\|\Bar{x}_i-c_i\|}$ and $\bar\xi_i=\frac{\sin(\theta_i)u_{h(x)-1}}{\sin(\beta_i)\|u_{h(x)-1}\|}-\frac{\sin(\theta_i-\beta_i)}{\sin(\beta_i)}V_{ci}$. 

The segment $\mathcal{CL}_i$ divides the set $\mathcal{AR}^{h}_i$ into two symmetric regions, a left-side $\mathcal{P}_{<}(c_i,R_1\Bar{V}_{ci})\cap\mathcal{AR}^{h}_i$, and a right-side $\mathcal{P}_{>}(c_i,R_1\Bar{V}_{ci})\cap\mathcal{AR}^{h}_i$. On the right side, the control is on the right tangent to obstacle $i$, while on the left, the control is on the right tangent to obstacle $i$. Therefore, the control can not point inside the region enclosed by the vectors $\Bar{V}_{ci}$, $V_{ci}$ ({\it i.e.,} $\forall x\in \mathcal{AR}^{h}_i\setminus\mathcal{CL}_i,\,u(x)\notin\{v\in\mathbb{R}^n|v=a\Bar{V}_{ci}+bV_{ci},\,a>0,b>0\}$). Thus, for all $x\in U$, $\Bar{V}_{ci}^{\top}\bar\xi_i=\cos(\sigma_i+\theta_i)$ where $0<\sigma_i=\angle(\Bar{V}_{ci},V_{ci})<\pi$ and $0<\theta_i=\angle(V_{ci},\Bar{\xi}_i)\leq\frac{\pi}{2}$. Then,
\begin{align*}
    \Dot{V}(x)&=
    -K\left(\cos(\sigma_i+\theta_i)-\cos(\sigma_i)\cos(\theta_i) \right),\\
    &=K\sin(\sigma_i)\sin(\theta_i).
\end{align*}
Therefore, $\Dot{V}(\Bar{x}_i)=0$ and $\Dot{V}(x)>0$ for all $x\in U$.\\
{\bf Case 2:} Consider the region $\mathcal{AR}^{h}_k$ where $k\in\mathcal{M}_i^{\Bar{N}_i}$. Define the tube $\mathcal{TU}_i:=\{x\in\mathcal{F}|d(x,\mathcal{CL}_i)\leq e_i\}$ where $e_i$ is small such that $\mathcal{TU}_i\cap\mathcal{L}_j=\varnothing$ for all $j\in\mathcal{Z}\setminus\{i\}$, and $\mathcal{TU}_i\subset\mathcal{AR}_k$. Let $V(x)=1-\frac{(\Bar{x}_i-x^*_{k,i})^\top}{\|\Bar{x}_i-x^*_{k,i}\|}\frac{(x-x^*_{k,i})}{\|x-x^*_{k,i}\|}$ where $\Bar{x}_i\in\mathcal{CL}_i\cap\mathcal{AR}^{h}_k$ and $x^*_{k,i}=\mathrm{arg}\max\limits_{q\in\mathcal{L}_i^e\cap\partial\mathcal{O}_k}||c_i-q||$. Note that $V(\Bar{x}_i)=0$, and $V(x)>0$ for all $x\in (\mathcal{TU}_i\cap\mathcal{AR}_k^h)\setminus  \mathcal{CL}_i$. Define the set $U:=(\mathcal{TU}_i\cap\mathcal{AR}_h^k)\setminus\mathcal{CL}_i$. The time-derivative of $V(x)$ is given by  
\begin{align*}
    \Dot{V}(x)&=\frac{\partial V(x)}{\partial x}^\top\Dot{x},\\&=-\frac{(\Bar{x}_i-x^*_{k,i})^\top}{\|\Bar{x}_i-x^*_{k,i}\|}J_x\left(\frac{(x-x^*_{k,i})}{\|x-x^*_{k,i}\|}\right)u(x),\\
    &=-K\Bar{V}^\top_{k,i}\pi^{\bot}(V_{k,i})\bar\xi_k,
\end{align*}
where $K=\frac{\|u(x)\|}{\|x-x^*_{k,i}\|}>0$, $V_{k,i}=\frac{(x^*_{k,i}-x)}{\|x^*_{k,i}-x\|}$, $\bar{V}_{k,i}=\frac{(\Bar{x}_i-x^*_{k,i})}{\|\Bar{x}_i-x^*_{k,i}\|}$ and $\bar\xi_k=\frac{\sin(\theta_k)u_{h(x)-1}}{\sin(\beta_k)\|u_{h(x)-1}\|}-\frac{\sin(\theta_k-\beta_k)}{\sin(\beta_k)}V_{ck}$. Similar to the previous case, segment $\mathcal{CL}_i$ divides set $\mathcal{AR}^{h}_k$ into two regions, a left-side $\mathcal{P}_{<}(c_i,R_1\Bar{V}_{k,i})\cap\mathcal{AR}^{h}_k$ and a right-side $\mathcal{P}_{>}(c_i,R_1\Bar{V}_{k,i})\cap\mathcal{AR}^{h}_k$. On the right side, the control is on the right tangent to obstacle $k$, while on the left, the control is on the left tangent to obstacle $k$. Therefore, by considering $V_{k,i}$ instead of $V_{ck}$ where $V_{k,i}=a\Bar{V}_{k,i}+bV_{ck},\,a>0,\,b>0$, $0<\sigma_{k,i}=\angle(\Bar{V}_{k,i},V_{k,i})< \pi$, and $\theta_k<\theta_k^*=\angle(\Bar{V}_{k,i},\Bar{\xi}_k)<\pi$, we can write $\Bar{V}_{k,i}^{\top}\Bar{\xi}_k=\cos(\sigma_{k,i}+\theta_k^*)$ for all $x\in U$. Then, 
\begin{align*}
    \Dot{V}(x)&=
    -K\left(\cos(\sigma_{k,i}+\theta_k^*)-\cos(\sigma_{k,i})\cos(\theta_k^*) \right),\\
    &=K\sin(\sigma_{k,i})\sin(\theta_k^*),
\end{align*}
Therefore, $\Dot{V}(\Bar{x}_i)=0$ and $\Dot{V}(x)>0$ for all $x\in U$.\\
{\bf Case 3:} Consider the region $\mathcal{AR}_k^h$ where $\mathcal{L}_i^e\cap\mathcal{O}_k=\varnothing$ and $\mathcal{CL}_i\cap\mathcal{AR}_k^h\neq\varnothing$. Let $\bar\Omega_i=R_b(c_i-\bar x_i)$ where $\Bar{x}_i\in\mathcal{CL}_i\cap\mathcal{AR}^{h}_k$, $R_b=\big(\begin{smallmatrix}
  0 & b\\
  -b & 0
\end{smallmatrix}\big)$, and $b\in\{-1,1\}$ is chosen such that $\bar\Omega_i^{\top}(c_k-\bar x_i)>0$. Since $c_k\in\mathcal{P}_>(\bar x_i,\bar\Omega_i)$ and $\mathcal{L}_i^e\cap\mathcal{O}_k=\varnothing$, $\mathcal{O}_k\subset\mathcal{P}_>(\bar x_i,\bar\Omega_i)$. Define the tube $\mathcal{TU}_i:=\{x\in\mathcal{F}|d(x,\mathcal{CL}_i)\leq e_i\}$ where $e_i$ is small such that $\mathcal{TU}_i\cap\mathcal{P}_>(\bar x_i,\bar\Omega_i)\cap\mathcal{L}_j=\varnothing$ for all $j\in\mathcal{Z}\setminus\{i\}$, and $\bar\Omega_i^{\top}u(x)>0$ for all $x\in\mathcal{TU}_i\cap\mathcal{P}_>(\bar x_i,\bar\Omega_i)$. Let $V(x)=\bar\Omega_i^{\top}(x-\bar x_i)$ where $V(\Bar{x}_i)=0$ and $V(x)>0$ for all $x\in\mathcal{P}_>(\bar x_i,\bar\Omega_i)$. Define the set $U:=\mathcal{TU}_i\cap\mathcal{P}_>(\bar x_i,\bar\Omega_i)\cap\mathcal{AR}^h_k$. 
\begin{align*}
    \Dot{V}(x)&=\frac{\partial V(x)}{\partial x}^\top\Dot{x}=\bar\Omega_i^{\top}u(x),
\end{align*}
where $\Dot{V}(x)>0$ for all $x\in U$ and $\Dot{V}(x)=0$ for $x\in\mathcal{CL}_i$.\\
In the three treated cases, $U$ is compact, $V(x)=0$ on the undesired equilibria $\mathcal{CL}_i$, and $V$ is increasing on $U$. Therefore, $x(t)$ must leave $U$.\\
\begin{figure}[!h]
\centering
\includegraphics[scale=0.4]{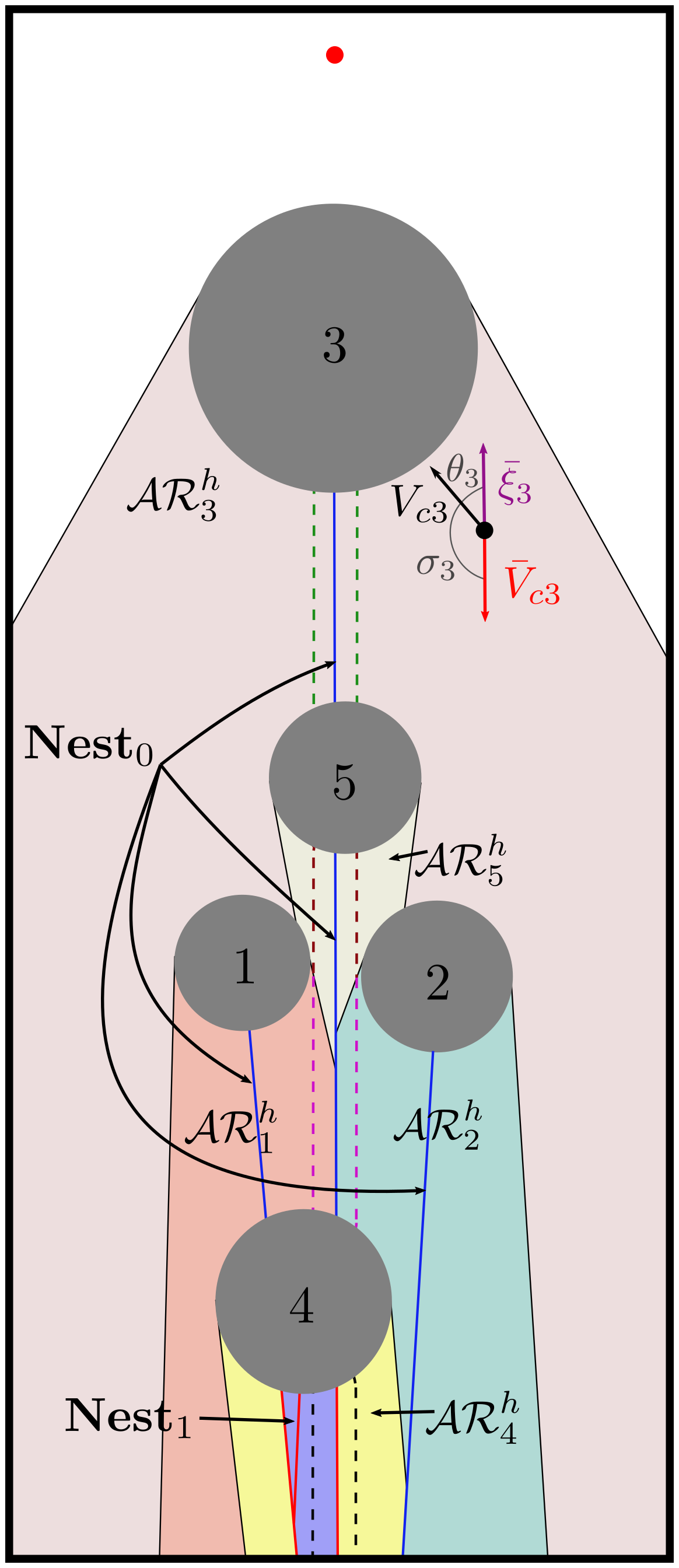}
\caption{Illustration of the nests (attraction region of the undesired equilibria.)}
\label{fig:2dnests}
\end{figure}
Now, we will show that if we start outside a given nest $\textbf{Nest}_k,\,k>0$, we will never get inside any nest.\\
 Consider the nest $\textbf{Nest}_k$, obstacle $j\in\mathbb{L}$, and the set of undesired equilibria $\mathcal{L}_i$ such that $k>0$, $i\in\mathcal{R}_j$, and $\partial\textbf{Nest}_k\cap\mathcal{L}_i\cap\mathcal{AR}_j^h\neq\varnothing$ ({\it i.e.,} segment (or segments) of $\mathcal{L}_i$ forms a portion of the boundary of the nest $\textbf{Nest}_k$ when the last projection involves obstacle $j$). Define the tube $\mathcal{TU}_i:=\{x\in\mathcal{F}|d(x,\mathcal{L}_i\cap\mathcal{AR}_j^h)\leq e_i\}$ where $e_i$ is small such that $(\mathcal{TU}_i\setminus\textbf{Nest}_k)\cap\mathcal{L}_p=\varnothing$ for all $p\in\mathcal{Z}\setminus\{i\}$, and $\mathcal{TU}_i\subset\mathcal{AR}_j$. This case amounts to case 2) with $U:=(\mathcal{TU}_i\setminus\textbf{Nest}_k)\cap\mathcal{AR}_j^h$. 
Since the nests are invariants,  $\cup_{i\in\mathcal{Z}}\mathcal{L}_i\subset\cup_k\textbf{Nest}_k$, and for all $x(0)\notin\cup_k\textbf{Nest}_k$, $\lim_{t \to+\infty}x(t)\notin\textbf{Nest}_k$, the set of nests $\cup_k\textbf{Nest}_k$, is the attraction region of the undesired equilibria.
\subsection*{Proof of Theorem \ref{the2}}\label{appendix:the2}
Item i) follow directly from Lemma \ref{lem7}. According to item v) of Theorem \ref{the1}, all trajectories converging to $x_d$ are generated by a {\it quasi-optimal} obstacle avoidance maneuver and item i) states that $x_d$ is attractive from all $x(0)\in\mathcal{F}\setminus\cup_k\textbf{Nest}_k$, which proves item ii). Since Assumption \ref{as:3} imposes that $\mathbb{L}=\varnothing$, no cell will be created, which implies that only the special nest exists. Therefore, $\cup_k\textbf{Nest}_k=\textbf{Nest}_0=\cup_{i\in\mathcal{Z}}\mathcal{L}_i$ is the region of attraction of the undesired equilibria $\cup_{i\in\mathcal{Z}}\mathcal{L}_i$, as per Lemma \ref{lem7}, and has Lebesgue measure zero, which shows the almost global asymptotic stability of $x_d$.
\subsection*{Proof of Lemma \ref{lem8}}\label{appendix:Lemma 8}
Let us look for the equilibria of the closed-loop system \eqref{12}-\eqref{sensor_based_ctrl} by setting $u(x)=0$. Then,
from the first equation of \eqref{sensor_based_ctrl}, the equilibrium point is $x_d$. From the second equation of \eqref{sensor_based_ctrl}, one gets $u_d(x)=\|u_d(x)\|\frac{\sin(\tilde{\theta}-\tilde{\beta})}{\sin(\tilde{\theta})}\frac{\tilde{c}-x}{\|\tilde{c}-x\|}$ which is true if and only if $\tilde\beta=0$ ({\it i.e.,}$\angle(u_d,(\tilde c-x))=0$). As $\widetilde{\mathcal{BL}}$ is the union of the disjoints practical shadow regions, there exists a unique $i\in\mathbb{I}$ such that if $x\in\widetilde{\mathcal{BL}}$, $x\in\tilde{\mathcal{D}}(x_d,c_i,R)$, and since $\tilde{c}$ is the projection of $x$ onto obstacle $i$, $\frac{\tilde c-x}{\|\tilde c -x\|}=\frac{c_i-x}{\|c_i -x\|}$. Therefore, $u(x)=0$ for all $x\in\tilde{\mathcal{L}}_d(x_d,c_i,R):=\mathcal{L}_d(x_d,c_i)\cap\tilde{\mathcal{D}}(x_d,c_i,R)$ where $\mathcal{L}_d(x_d,c_i)$ is defined in Lemma \ref{lem2}. 
    Finally, one can conclude that the set of undesired equilibria for the closed-loop system \eqref{12}-\eqref{sensor_based_ctrl} is $\zeta=\{x_d\}\cup\left(\cup_{i\in\mathbb{I}}\tilde{\mathcal{L}}_d(x_d,c_i,R)\right)$.
    
\subsection*{Proof of Theorem \ref{the3}}\label{appendix:the3}
Since the sensor-based case is a special case of the map-based scenario when each obstacle is considered as a unique obstacle in the workspace, item i) follows from item i) of Theorem \ref{the1}. Item ii) follows directly from Lemma \ref{lem8}. Now, let us prove item iii). 
Consider obstacle $i\in\mathbb{I}$ and the set of equilibrium points $\tilde{\mathcal{L}}_d(x_d,c_i,R)$. Define the tube $\mathcal{TU}_i:=\{x\in\tilde{\mathcal{D}}(x_d,c_i,R)|d(x,\tilde{\mathcal{L}}_d(x_d,c_i,R))\leq e_i\}$ surrounding $\tilde{\mathcal{L}}_d(x_d,c_i,R)$ inside the practical shadow region where $e_i$ is small such that $\tilde{c}=\mathrm{arg}\min\limits_{y\in\mathcal{O}_i}\|x-y\|$.  Let $V(x)=1-\frac{(\tilde x_i-c_i)^\top}{\|\tilde x_i-c_i\|}\frac{(x-\tilde c)}{\|x-\tilde c\|}$ where
$\tilde x_i\in\tilde{\mathcal{L}}_d(x_d,c_i,R)$. Note that $V(\tilde x_i)=0$ and $V(x)>0$ for all $x\in \mathcal{TU}_i\setminus  \tilde{\mathcal{L}}_d(x_d,c_i,R)$. Define the set $U:=\{x\in\mathcal{TU}_i|V(x)>0\}$. The time-derivative of $V(x)$ on $\mathcal{TU}_i$ is given by
\begin{align*}
    \Dot{V}(x)&=\frac{\partial V(x)}{\partial x}^\top\Dot{x},\\&=-\frac{(\tilde x_i-c_i)^\top}{\|\tilde x_i-c_i\|}J_x\left(\frac{(x-\tilde c)}{\|x-\tilde c\|}\right)u(x),\\
    &=\frac{-1}{\|\tilde x_i-c_i\|\|x-\tilde c\|}(\tilde x_i-c_i)^\top\pi^{\bot}\left(\frac{(x-\tilde c)}{\|x-\tilde c\|}\right)u(x),\\
    &=\frac{-\gamma}{\|\tilde x_i-c_i\|\|x-\tilde c\|}(\tilde x_i-c_i)^\top\pi^{\bot}\left(\frac{(x-\tilde c)}{\|x-\tilde c\|}\right)(x_d-x),\\
    &=-K(\tilde x_i-c_i)^\top\pi^{\bot}(V_{ci})(x_d-x),
\end{align*}
where $K=\frac{\gamma}{\|\tilde x_i-c_i\|\|x-\tilde c\|}$, and as $\tilde c$ is the projection of $x$ onto obstacle $i$ for all $x\in\mathcal{TU}_i$, $V_{ci}=\frac{(c_i-x)}{\|c_i-x\|}=\frac{(\tilde c-x)}{\|\tilde c-x\|}$. Since $\tilde x_i=c_i+\delta\frac{c_i-x_d}{\|c_i-x_d\|}$ with $\delta\geq r_i$, one gets
\begin{align*}
    \Dot{V}(x)&=-\frac{\delta K}{\|c_i-x_d\|}(c_i-x_d)^{\top}\pi^{\bot}(V_{ci})(x_d-x),\\
    &=-\frac{\delta K}{\|c_i-x_d\|}\bigl((c_i-x)+(x-x_d)\bigr)^{\top}\pi^{\bot}(V_{ci})(x_d-x),\\
    &=\frac{\delta K}{\|c_i-x_d\|}(x_d-x)^\top\pi^{\bot}(V_{ci})(x_d-x).
\end{align*}
where we used the fact that $(c_i-x)^{\top}\pi^{\bot}(V_{c_i})(x_d-x)=0$. It is clear that $\Dot{V}(x)>0$ for all $x\in U$, and $\Dot{V}(x)=0$  for all $x\in\tilde{\mathcal{L}}_d(x_d,c_i,R)$ over the set $\mathcal{TU}_i$. Since $U$ is a compact set, $V(x)$ is increasing on $U$, and $V(x)=0$ on $\tilde{\mathcal{L}}_d(x_d,c_i,R)$ (the tube axis), $x(t)$ must leave the set $U$. Note that the set $U$ is bounded by the free space boundary and the lateral surface of tube $\mathcal{TU}_i$. Due to the safety of the system, as per item i), $x(t)$ can not leave $U$ from the free space boundary and can only leave it from the surface of the tube for all $x(0)\in U$. Therefore, the set of equilibria $\tilde{\mathcal{L}}_d(x_d,c_i,R)$ is unstable. 
Lastly, we prove item iv). Consider the equilibrium point $x_d$ and the positive definite function $V_1(x)=\frac{1}{2}||x-x_d||^2$ whose time-derivative is given by
\begin{align*}
    \Dot{V}_1(x)&=\frac{\partial V_1(x)}{\partial x}^\top\Dot{x},\\&=(x-x_d)^\top u(x),\\&=\begin{cases}
    -\gamma\|x-x_d\|^2,&x\in\widetilde{\mathcal{VI}}\\
    -\gamma\|x-x_d\|^2+ 
     \gamma\|x-x_d\|\frac{\sin(\tilde\theta-\tilde\beta)}{\sin(\tilde\theta)}\frac{(x_d-x)^\top(\tilde c-x)}{\|\tilde c-x\|},
 & x\in\widetilde{\mathcal{BL}}
\end{cases}\\
      &=\begin{cases}
      -\gamma\|x-x_d\|^2,&x\in\widetilde{\mathcal{VI}}\\
      -\gamma\|x-x_d\|^2+
      \gamma\|x-x_d\|^2\frac{\sin(\tilde\theta-\tilde\beta)}{\sin(\tilde\theta)}\cos(\tilde\beta),
      &
      x\in\widetilde{\mathcal{BL}}
      \end{cases}\\
      &=\begin{cases}
      -\gamma\|x-x_d\|^2,&x\in\widetilde{\mathcal{VI}}\\
      -\gamma\|x-x_d\|^2\frac{\sin(\tilde\beta)}{\sin(\tilde\theta)}\cos(\tilde\theta-\tilde\beta),&x\in\widetilde{\mathcal{BL}}
      \end{cases}
\end{align*}
where we used the fact that $\sin(\tilde{\theta})-\sin(\tilde{\theta}-\tilde{\beta})\cos(\tilde{\beta})=\sin(\tilde\beta)\cos(\tilde\theta-\tilde\beta)$, $0<\tilde\theta\leq\frac{\pi}{2}$ and $0\leq\tilde\beta\leq\tilde\theta$. Therefore, $\Dot{V}_1(x)=0$ only for $x\in\zeta$ and $\Dot{V}_1(x)<0$ for all $x\in\mathcal{F}\setminus\zeta$. Since the practical shadow regions are disjoint by construction, and the undesired equilibria $\tilde{\mathcal{L}}_d(x_d,c_i,R)$ are limited to the shadow regions as per Lemma \ref{lem8}, $\mathbb{L}=\varnothing$ and $\textbf{Cell}_i=\varnothing$ for all $i\in\mathbb{I}$. Thus, for all $i\in\mathbb{I}$, the attraction region of the set of undesired equilibria $\tilde{\mathcal{L}}_d(x_d,c_i,R)$ reduces to the line segment $\mathcal{L}_d(x_d,c_i)\cap\mathcal{D}^t(x_d,c_i)$. Since the attraction region of the undesired equilibria is a set of measure zero, the equilibrium point $x_d$ is almost globally asymptotically stable in $\mathcal{F}$. 
\subsection*{Proof of Lemma \ref{lem9}}\label{appendix:Lemma 9}
Following the same procedure as in \ref{appendix:Lemma 8}, $u(x)=0$ if and only if $x=x_d$ or $x\in\tilde{\mathcal{L}}_d(x_d,\mathrm{x}_i,R)$, where for all $x\in\tilde{\mathcal{L}}_d(x_d,\mathrm{x}_i,R)$, $\tilde c=\mathrm{x}_i$ and 
$\frac{(\tilde{c}-x)^{\top}}{\|\tilde{c}-x\|}\frac{(x_d-x)^{\top}}{\|x_d-x\|}=1$. Therefore, one can conclude that the set of equilibria is $\tilde{\zeta}=\{x_d\}\cup\left(\cup_{i\in\mathbb{I}}\tilde{\mathcal{L}}_d(x_d,\mathrm{x}_i,R)\right)$.
\subsection*{Proof of Theorem \ref{the4}}\label{appendix:the4}
Since the considered convex obstacles have smooth boundaries, the tangent cone on the boundaries of the obstacles are half-planes, which is similar to the spherical obstacles case. Therefore, item i) follows from item i) of Theorem \ref{the3}. Item ii) follows directly from Lemma \ref{lem9}. Now, let us prove item iii). Consider obstacle $i\in\mathbb{I}$ and the equilibrium point $\bar x_i\in \tilde{\mathcal{L}}_d(x_d,\mathrm{x}_i,R)$.
Define the tube $\mathcal{TU}_i:=\{x\in\tilde{\mathcal{D}}(x_d,i,R)|d(x,\tilde{\mathcal{L}}_d(x_d,\mathrm{x}_i,R))\leq e_i\}$ surrounding $\tilde{\mathcal{L}}_d(x_d,\mathrm{x}_i,R)$ inside the practical shadow region where $e_i$ is small such that $\tilde{c}=\mathrm{arg}\min\limits_{y\in\tilde{\mathcal{O}}_i}\|x-y\|$.  Let $V(x)=1-\frac{(\bar x_i-x_d)^\top}{\|\bar x_i-x_d\|}\frac{(x-\tilde c)}{\|x-\tilde c\|}$ where
$\bar x_i\in\tilde{\mathcal{L}}_d(x_d,\mathrm{x}_i,R)$. Note that $V(\bar x_i)=0$ and $V(x)>0$ for all $x\in \mathcal{TU}_i\setminus  \tilde{\mathcal{L}}_d(x_d,\mathrm{x}_i,R)$. Define the set $U:=\{x\in\mathcal{TU}_i|V(x)>0\}$. The time-derivative of $V(x)$ on $\mathcal{TU}_i$ is given by
\begin{align*}
    \Dot{V}(x)&=\frac{\partial V(x)}{\partial x}^\top\Dot{x}=-\frac{(\bar x_i-x_d)^\top}{\|\bar x_i-x_d\|}J_x\left(\frac{(x-\tilde c)}{\|x-\tilde c\|}\right)u(x),\\
    &=\frac{-1}{\|x-\tilde c\|}\frac{(\bar x_i-x_d)^\top}{\|\bar x_i-x_d\|}\pi^{\bot}\left(\frac{(x-\tilde c)}{\|x-\tilde c\|}\right)u(x),\\
    &=\frac{\gamma}{\|x-\tilde c\|}\frac{(\bar x_i-x_d)^\top}{\|\bar x_i-x_d\|}\pi^{\bot}\left(\frac{(x-\tilde c)}{\|x-\tilde c\|}\right)(x-x_d),\\
    &=K(\cos(\varphi_d)-\cos(\tilde{\varphi})\cos(\tilde \beta)),
\end{align*}
where $K=\gamma\frac{\|x-x_d\|}{\|x-\tilde c\|}$, $\tilde{\varphi}=\angle(\bar x_i-x_d,x-\tilde c)$, $0<\varphi_d=\angle(\bar x_i-x_d,x-x_d)\leq\varphi_d^{max}$, and $\varphi_d^{max}=\arcsin(e_i/\|x-x_d\|)\in(0,\pi/2]$. Since obstacle $i$ satisfies the curvature condition, as per Assumption \ref{as:4}, and $\tilde c$ is the projection of $x$ onto obstacle $i$, $0<\tilde\varphi=\varphi_d+\tilde\beta<\pi$, where $ 0\leq\tilde{\beta}=\angle(\tilde{c}-x,x_d-x)\leq\frac{\pi}{2}$. Thus, $ \Dot{V}(x)=K\sin(\tilde{\varphi})\sin(\tilde\beta).$
 It is clear that, over the set $\mathcal{TU}_i$, $\Dot{V}(x)=0 $ for $x\in\tilde{\mathcal{L}}_d(x_d,\mathrm{x}_i,R)$({\it i.e.,} $\tilde{\beta}=0$), and $\Dot{V}(x)>0$ for all $x\in U$. Since $U$ is a compact set, $V(x)$ is increasing on $U$, and $V(x)=0$ on $\tilde{\mathcal{L}}_d(x_d,\mathrm{x}_i,R)$ (the tube axis), $x(t)$ must leave the set $U$. Note that the set $U$ is bounded by the free space boundary and the lateral surface of tube $\mathcal{TU}_i$. Due to the safety of the system, as per item i), $x(t)$ can not leave $U$ from the free space boundary and can only leave it from the surface of the tube for all $x(0)\in U$. Therefore, the set of equilibria $\tilde{\mathcal{L}}_d(x_d,\mathrm{x}_i,R)$ is unstable. Finally, proof of item iv) is similar to that of item iv) in Appendix \ref{appendix:the3}.

\bibliographystyle{unsrtnat}
\bibliography{library}  






\end{document}